\documentclass[a4paper, onecolumn, accepted=2021-04-25]{quantumarticle}
\pdfoutput=1
\usepackage[utf8]{inputenc}
\usepackage[T1]{fontenc}
\usepackage{fullpage}
\usepackage{amsmath}
\usepackage{amsfonts}
\usepackage{amssymb}
\usepackage[makeroom]{cancel}
\usepackage{amsthm}
\usepackage{graphicx}
\usepackage{float}
\usepackage{multirow}
\usepackage{mathtools}
\usepackage{bbm}
\usepackage{braket}
\usepackage{ytableau,tikz,varwidth}
\usepackage[enableskew]{youngtab}
\usepackage{hyphenat}
\usepackage{verbatim}
\usepackage{subcaption}
\usepackage{hyperref}
\usepackage{soul}
\usepackage{etoolbox}
\usepackage{nameref}
\usepackage{tikz}
\usetikzlibrary{tikzmark,decorations.pathreplacing, shapes.geometric, arrows, calc, fit, positioning, chains, arrows.meta}

\theoremstyle{definition}

\allowdisplaybreaks

\def\CC {{\mathbb C}}     

\def\mc {\mathcal}
\def\mk {\mathfrak}

\def\mH {\mc{H}}

\def\mkg {\mk{g}}

\def\mku {\mk{u}}

\def\kpsi {\ket{\Psi}}

\def\bone {\mathbbm{1}}

\def\im {{\rm{Im}}}

\newcommand{\floor}[1]{\ensuremath\lfloor #1\rfloor}

\newcommand{\cupdot}{\mathbin{\mathaccent\cdot\cup}}

\DeclarePairedDelimiterX\bk[2]{\langle}{\rangle}{#1 \delimsize\vert #2}
\DeclarePairedDelimiterX\kb[2]{\vert }{\vert }{#1 \delimsize\rangle\langle#2}

\def\={\;=\;} \def\+{\,+\,}

\newtheorem{theorem}{Theorem}
\newtheorem{lemma}[theorem]{Lemma}

\newtheorem{corollary}[theorem]{Corollary}

\newtheorem{remark}{Remark}
\newtheorem{example}{Example}
\newtheorem{fact}{Fact}

\usepackage{tikz}
\newcommand*\circled[1]{\tikz[baseline=(char.base)]{
            \node[shape=circle,draw,inner sep=2pt] (char) {#1};}}
            
\apptocmd\appendix{%
  \addcontentsline{toc}{chapter}{Appendix}%
  \renewcommand*{\thesection}{\Alph{section}}%
  \counterwithin{equation}{section}%
  \counterwithin{figure}{section}%
  \counterwithin{table}{section}%
}{}{}

\tikzset{
    structure/.style = {draw, rectangle,
        minimum height=1cm,
        minimum width=1cm, rounded corners,, text centered, draw=black, fill=white!30},
    arrow/.style = {thick, ->, >=stealth}
}

\begin{document}

\title{Designing locally maximally entangled quantum states with arbitrary local symmetries}
\author{Oskar S\l{}owik$^{1,*}$, Adam Sawicki$^1$, and Tomasz Maci\k{a}\.{z}ek$^{2}$}
\date{%
    $^1$ Center for Theoretical Physics, Polish Academy of Sciences, Al. Lotników 32/46, 02-668 Warsaw, Poland\\%
$^2$ School of Mathematics, University of Bristol, Fry Building, Woodland Road, Bristol BS8 1UG, UK \\[2ex]%
$^*$ Corresponding author: oslowik@cft.edu.pl\\[2ex]%
   26 April 2021
}
\maketitle

\begin{abstract}
One of the key ingredients of many LOCC protocols in quantum information is a multiparticle (locally) maximally entangled quantum state, aka a critical state, that possesses local symmetries. We show how to design critical states with arbitrarily large local unitary symmetry. We explain that such states can be realised in a quantum system of distinguishable traps with bosons or fermions occupying a finite number of modes. Then, local symmetries of the designed quantum state are equal to the unitary group of local mode operations acting diagonally on all traps. Therefore, such a group of symmetries is naturally protected against errors that occur in a physical realisation of mode operators. We also link our results with the existence of so-called strictly semistable states with particular asymptotic diagonal symmetries. Our main technical result states that the $N$th tensor power of any irreducible representation of $\mathrm{SU}(N)$ contains a copy of the trivial representation. This is established via a direct combinatorial analysis of Littlewood-Richardson rules utilising certain combinatorial objects which we call telescopes.
\end{abstract}

\section{Introduction}
Multipartite entangled states \cite{NiCh09, HoHo09} play important roles in different areas of physics including quantum computation (e.g. measurement based \cite{RaBr01}), quantum communication \cite{HiBu99, Go00} and quantum metrology \cite{GiLl11} as well as condensed matter physics \cite{AmFa08}. One can take on an operational point of view on multipartite entanglement and look at the problem of convertibility of quantum states under Local Operations assisted by Classical Communication (LOCC). One of the key points of such a formulation is that LOCC operations cannot increase entanglement \cite{HoHo09}. Hence, entanglement plays the role of a resource in quantum information processing \cite{resource}.

Another important motivation for our work comes form the fact that any operationally useful quantum state must have local symmetries \cite{GoWa11}. Let us next briefly review the main reasons why it is the case. To this end, we will be considering not only LOCC operations, but also certain larger and more tractable classes of operations called SLOCC (Stochastic LOCC) and local separable transformations (SEP), see Fig. \ref{fig:classes}.

\begin{figure}[H]
\centering 
\includegraphics[width=.5\linewidth]{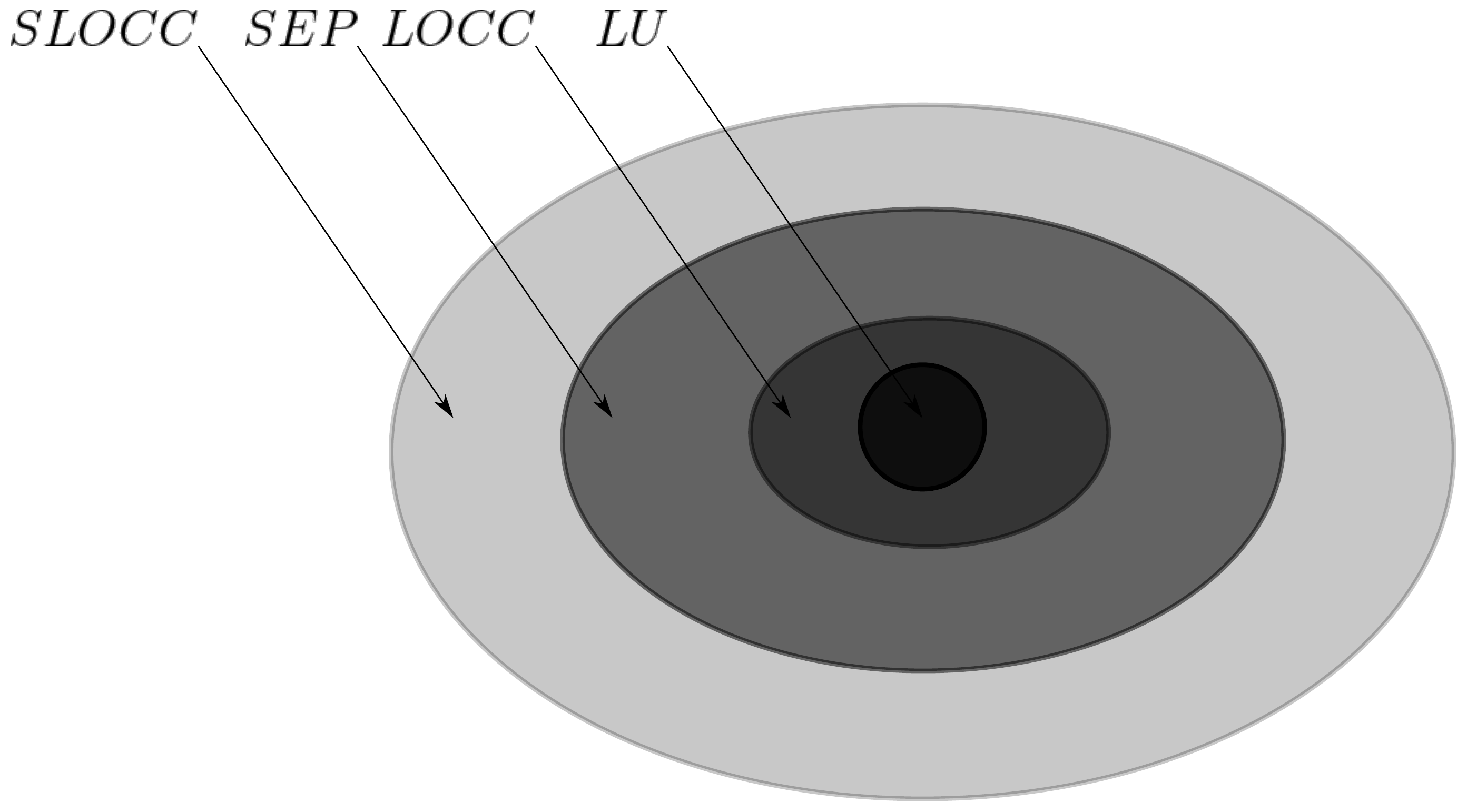}
\caption{The relationship between sets of SLOCC, SEP, LOCC and Local Unitary (LU) operations.}
\label{fig:classes}
\end{figure}

We focus on systems of $N$ qudits described by Hilbert spaces of the form $\mH \coloneqq \CC^{d_1}\otimes\hdots\otimes \CC^{d_N}$ and restrict our attention only to pure states. SLOCC operations are mathematically described by group $G \coloneqq \mathrm{SL}(d_1,\mathbb{C})\times \cdots \times \mathrm{SL}(d_N,\mathbb{C})$  and the action is given by the tensor product
\begin{equation}
(g_1,\ldots,g_N).\ket{\Psi} \coloneqq \frac{(g_1\otimes \ldots \otimes g_N)\ket{\Psi}}{||(g_1\otimes \ldots \otimes g_N)\ket{\Psi}||}\mathrm{.}
\end{equation}
Orbits of this action correspond to sets of SLOCC-equivalent states. Similarly, orbits of the natural action of the local unitary group (LU) $K \coloneqq \mathrm{SU}(d_1,\mathbb{C})\times \cdots \times \mathrm{SU}(d_N,\mathbb{C}) \subset G$ define classes of LU-equivalent states. Furthermore, the set of SEP operations consists of all operations given by Kraus operators in the product form. In other words, a state with density matrix $\rho_1$ is equivalent to a state with density matrix $\rho_2$ under SEP if
\begin{equation}
\rho_1 \mapsto \rho_2=\sum_{m=1}^M K_m \rho_1 K_m^\dagger,\quad \sum_{m=1}^M K_m^\dagger K_m=\mathbb{I},
\end{equation}
and the product form means that, for every $m$, we can write $K_m=\otimes_{n=1}^N K_{m,n}$ for some $K_{n,m}$. Two states,  $\ket{\Psi_1}$ and $\ket{\Psi_2}$, that are connected via SEP  are necessarily  connected via SLOCC. This means  we can write $\ket{\Psi_1}=g_1 \ket{\Psi}/||g_1 \ket{\Psi}||$ and $\ket{\Psi_2}=g_2 \ket{\Psi}/||g_2 \ket{\Psi}||$, where $g_i\in G$. A sufficient condition for SEP conversion from $\ket{\Psi_1}$ to $\ket{\Psi_2}$ was given in  \cite{GoWa11} and it boils down to satisfying the following equality 
\begin{equation}
\sum_{k=1}^L p_k S_k^\dagger \left(g_2g_2^\dagger\right)S_k=\left(g_1g_1^\dagger\right)\frac{||g_2\kpsi||^2}{||g_1\kpsi||^2}\ \mathrm{,}
\end{equation}
where $L \in \mathbb{N}$, for some probabilities $p_k \geq 0$, $\sum_{k=1}^L p_k=1$ and, what is important for us, some choice of operators $S_k \in G_{\ket{\Psi}}$, where $G_{\ket{\Psi}}$ is the stabiliser of $\ket{\Psi}$ defined by 
\begin{equation}
G_{\kpsi} \coloneqq \left\{g \in G |\;  g.\kpsi\propto \kpsi \right\} < G \mathrm{.}
\end{equation}
Thus, we have established two desirable properties of quantum states that make them useful for general quantum information protocols.
\begin{enumerate}
\item State $\kpsi$ should be highly entangled.
\item State $\kpsi$ should be highly symmetric. The larger the stabiliser $G_{\ket{\Psi}}$, the greater the freedom of conversion $\ket{\Psi_1}\mapsto \ket{\Psi_2}$ via SEP.
\end{enumerate}
 This motivates us to look for highly entangled states with large stabilisers. A canonical example of such a family of states are the celebrated stabiliser states (aka graph states) \cite{RaBr01,DAB03,HEB04} that found use in proposals for quantum computers that are naturally robust against decoherence \cite{GottesmanPhD}. Recently, symmetries of stabiliser states have been completely classified \cite{EK20,HESVK19}. Importantly, some new symmetries of stabiliser states have been found showing that there exist new potential applications of such states in quantum information protocols \cite{EK20}. In this work, we take the above approach to the extreme and take up the task of constructing new families of maximally entangled quantum states that have very large local symmetries. In contrast with stabiliser states, our proposed states always have continuous (infinite) symmetries. More specifically, we restrict our attention to so-called Locally Maximally Entangled states (LME), also called critical states. These are states for which all reduced one-qudit reduced density matrices are maximally mixed (proportional to the identity matrix). 
 
It is well-known that LME states of two qudits have large symmetry groups (see e.g. \cite{GoWa11}). Namely, for $\ket{\Psi_{\mathrm{LME}}}=\sum_{i=1}^d\ket{ii}$, stabiliser $G_{\kpsi}$ consists of matrices of the form $A^{-1}\otimes A^T$ for any $A\in \mathrm{SL}(d,\mathbb{C})$. Furthermore, the $GHZ$ states, $W$ states and cluster states \cite{RaBr01} have stabilisers of dimension greater than one. However, this is not a generic property of quantum states, as a generic quantum state (so also a generic LME state) of $N>4$ qubits and $N>3$ qudits has a trivial stabiliser \cite{SaWa18}. Therefore, it is a non-trivial and fundamentally important task to identify quantum states that have nontrivial and possibly large stabilisers.

\subsection{A brief exposition of the main result}

Another reason for considering LME states is that SLOCC entanglement classes containing LME states form an open and dense subset of the space of normalised pure quantum states (which is geometrically the projective space denoted by $\mathbb{P}(\mathcal{H})$). This open and dense subset is called the set of semistable states and denoted by $\mathbb{P}(\mathcal{H})_{ss}$. The remaining  states, i.e. those that belong to the null cone $\mathcal{N} \coloneqq \mathbb{P}(\mathcal{H})\setminus \mathbb{P}(\mathcal{H})_{ss}$ can of course have large symmetry groups but we will not discuss them in this paper (see \cite{NeMu84,MaSa15,MaSa18,SaOs12,SaMa18,WaDo13, SaOs14} for a detailed geometric and information-theoretic analysis of the null-cone). Interestingly, not all quantum systems of many qudits have LME states. There are some exceptional combinations of local qudit dimensions for which LME states do not exist (for details, see \cite{LME19}). Following \cite{LME19}, we will say that a state $\ket{\Psi}\in \CC^{d_1}\otimes  \CC^{d_2}\otimes\hdots\otimes  \CC^{d_N}$ has a {\it diagonal $H$-symmetry} if for some nontrivial representations $\pi_1,\pi_2,\hdots,\pi_N$ of group $H$
\begin{equation}
\pi_k:\ H\to \mathrm{SU}(d_k),\quad 1\leq k\leq N,
\end{equation}
we have 
\begin{equation}\label{diag-h-sym}
\pi_1(h)\otimes \pi_2(h)\otimes\hdots\otimes \pi_N(h)\ket{\Psi}= \ket{\Psi}{\mathrm{\ for\ every\ }}h\in H.
\end{equation}

For example, the $3$-qubit $W$ state $\ket{W}=\frac{1}{\sqrt{3}}\left(\ket{001}+\ket{010}+\ket{100}\right)$, has one-dimensional abelian stabiliser
\begin{equation}
 K_{\ket{W}}=\left\{
\begin{pmatrix}
e^{i\alpha} & 0 \\
0 & e^{-i\alpha}
  \end{pmatrix}
  \otimes
  \begin{pmatrix}
e^{i\alpha} & 0 \\
0 & e^{-i\alpha}
  \end{pmatrix}
  \otimes
  \begin{pmatrix}
e^{i\alpha} & 0 \\
0 & e^{-i\alpha}
  \end{pmatrix} | \ \alpha\in[0,2\pi)\right\} \mathrm{.}   
\end{equation}

Hence $\ket{W}$ has diagonal $U(1)$-symmetry.

Importantly, if all $\pi_1,\pi_2,\hdots,\pi_N$ are irreducible, then state $\kpsi$ satisfying \eqref{diag-h-sym} is automatically an LME state \cite{LME19}. To see this, recall that the equivariance of the one-particle reduced density matrices guarantees that for any $(U_1,\hdots,U_N)\in K$ we have
\begin{equation}
\label{eq:lme1}
\rho_k\left(U_1\otimes\hdots\otimes U_N\kpsi\right)=U_k^\dagger \rho_k(\kpsi) U_k \mathrm{,}
\end{equation}
where $\rho_k (\ket{\Psi})$ is a reduced density matrix of $k$-th qudit when $\ket{\Psi}$ is the state of a whole system.
Hence, if $\kpsi$ has a diagonal $H$-symmetry, then for all $h\in H$ and all $k\in \{1,\hdots,N\}$ we have
\begin{equation}
\label{eq:lme2}
\pi_k(h)^\dagger \rho_k(\kpsi)\pi_k(h)= \rho_k(\kpsi).
\end{equation}
Moreover, because $\pi_k$ for every $k$ is an irreducible representation of $H$, by the virtue of Schur's lemma, $\rho_k(\kpsi)$ must be proportional to identity. Thus $\kpsi$ is LME. 

To recapitulate, we have the following result (see Theorem 3.1 from \cite{LME19} and discussion below it).
\begin{corollary}
\label{cor:lme}
If the tensor product of irreducible representations $\pi_k:H\rightarrow \mathrm{SU}(d_k)$ of $H$, $\pi_1\otimes\hdots\otimes\pi_N$, contains a copy of the trivial representation of $H$, then the corresponding representation space, $\mH_N=\CC^{d_1}\otimes\hdots\otimes\CC^{d_N}$, contains LME states with diagonal $H$-symmetry.
\end{corollary}
In our paper we mainly focus on the specific problem concerning $N$ qudits with $d_1=\ldots=d_N \coloneqq d$, i.e. $\mathcal{H}=\left(\mathbb{C}^d\right)^{\otimes N}$, where LME states always exist and $H=\mathrm{SU}(m)$ with $2\leq m\leq d$.  

Our \textbf{main result} is Theorem \ref{th:main} which asserts that if an irreducible representation $E^{\lambda}$ of $\mathrm{SU}(m)$ (associated to Young diagram $\lambda$) corresponds to the Hilbert space of a qudit (i.e. $E^{\lambda} \simeq \mathbb{C}^d$), then the system with $N=m$ qudits, $\mathcal{H}=(E^\lambda)^{\otimes m}$, contains a copy of the trivial representation of $\mathrm{SU}(m)$. Consequently, $(E^\lambda)^{\otimes m}$ will contain LME states with diagonal $\mathrm{SU}(m)$-symmetry and so-called strictly semistable states that can be asymptotically transformed to them (see Lemma \ref{lem:ssex}). It is crucial to note that for a fixed local qudit dimension $d$ group $\mathrm{SU}(m)$ ($2\leq m\leq d$) has an irreducible representation (irrep) of dimension $d$ only for some particular values of $m$. The admissible $m$ can be found for instance using the Weyl dimension formula or Gelfand-Tsetlin patterns (see e.g. \cite{ArHu11}).

\begin{example}
Consider a system of $N$ qudits with single-particle space of dimension six, $\mathcal{H}=(\mathbb{C}^6)^{\otimes N}$. What are the values of $N$ for which $\mathcal{H}$ contains states with diagonal $\mathrm{SU}(N)$-symmetries, $2\leq N\leq 6$? Using the above corollary, it suffices to show that $\mathrm{SU}(N)$ has an irreducible $6$-dimensional representation. Clearly, it is true for $N=2$ (by treating one qudit as a particle with spin $5/2$) and for $N=6$ (the natural representation of $\mathrm{SU}(6)$). On top of that, there is an irreducible representation of $\mathrm{SU}(3)$ of dimension six (see Subsection \ref{sub:bosons}). There are no other possible diagonal $\mathrm{SU}(N)$-symmetries as $\mathrm{SU}(4)$ and $\mathrm{SU}(5)$ have no nontrivial irreps of dimension $6$. 
\end{example}

The problem we solve is related to recent developments concerning the so-called saturation problem in decomposition of tensor products of representations of $\mathrm{SU}(m)$ due to Knutson and Tao as well as Klyachko and Totaro (the exposition to these results (and more) can be found in the work of Fulton \cite{Fu00}). In particular, Proposition 14 in \cite{Fu00} points out another criterion (which we do not invoke here in full detail) for the existence of a trivial component in a tensor product of representations corresponding to (possibly different) Young diagrams $\lambda(1),\hdots,\lambda(N)$. The criterion requires checking a set of inductively constructed inequalities involving integer partitions related to Young diagrams $\lambda(1),\hdots,\lambda(N)$. Satisfying the criterion is {\it equivalent} to the existence of the trivial component in the product $E^{\lambda(1)}\otimes\hdots \otimes E^{\lambda(N)}$. Our work circumvents the necessity of checking this set of inequalities in the special case when $\lambda(1)=\lambda(2)=\hdots=\lambda(N) \coloneqq \lambda$ by performing a straightforward multiplication of Young diagrams for any $\lambda$. To our best knowledge, besides our Theorem \ref{th:main}, there is no comparably simple and general criterion allowing one to determine the existence of LME states with diagonal $\mathrm{SU}(m)$-symmetries.

The paper is organised as follows. In Subsection \ref{ssyd} we give a brief introduction to Young diagrams and representation theory of $\mathrm{SU}(m)$. In Section \ref{sec2} we present potential physical setups where one could construct our proposed LME states with arbitrary diagonal LU-symmetry. We discuss systems of distinguishable particles as well as systems of distinguishable traps containing bosons or fermions. We also show how one can use our main result, Theorem \ref{th:main}, to tell whether a given system contains so-called strictly semistable states. In Section \ref{sec:main} we recall some basic facts about our main technical tool which are Young diagrams. We briefly review the relationship between Young diagrams and irreducible representations of $\mathrm{SU}(m)$ and explain Littlewood-Richardson rules that concern the problem of decomposing a product of representations into irreducible components. Finally, in Section \ref{sec:main} we formulate our main result, i.e. Theorem \ref{th:main}, and give a sketch of its proof. For the sake of clarity of our presentation, the main technical weight of the proof is contained in the \nameref{app}. Section \ref{sec:final}, discusses alternative approaches and considers the problem of finding stated with diagonal $H$-symmetry with $H$ being a general compact semisimple Lie group. The methods we work out are subsequently applied to find LME states with diagonal $\mathrm{SO}(d)$ symmetry. Section \ref{summary} contains a summary and a discussion of  possible ways to extend our work.

\subsection{A brief introduction of Young diagrams and representation theory of $\mathrm{SU}(m)$}
\label{ssyd}
A \textit{Young diagram} is a finite collection of boxes forming left-justified rows with non-increasing length (from top to bottom). By enlisting the number of boxes in each row (from top to bottom) we obtain a partition $\lambda$ of a non-negative integer $n$ - the total number of boxes forming a diagram. Such a diagram is said to be of shape $\lambda$ and there is a 1-1 correspondence between Young diagrams with $n$ boxes and partitions of $n$, so often by $\lambda$ we denote a diagram itself. We define $|\lambda|$ to be $n$.

By filling the boxes of a Young diagram using symbols from some alphabet we obtain so called \textit{Young tableau}. Figure \ref{fig:yt} shows a Young tableau obtained by filling a Young diagram of shape $\lambda=(5,3,2)$ with some symbols from alphabet $\{1,2,3\}$.

\begin{figure}[H]
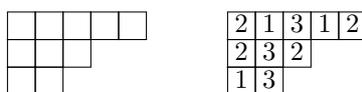

\centering 
\ytableausetup{nosmalltableaux,mathmode,boxsize=1em}\begin{ytableau}
       \null & \null & \null & \null & \null \\
       \null & \null & \null  \\
       \null & \null 
\end{ytableau}\quad\quad\quad
\begin{ytableau}
       2 & 1 & 3 & 1 & 2 \\
       2 & 3 & 2  \\
       1 & 3 
\end{ytableau}
\caption{A Young diagram of shape $\lambda=(5,3,2)$ - left; An example of a Young tableau obtained from $\lambda$ - right.}
\label{fig:yt}
\end{figure}

A \textit{transpose} of a Young diagram $\lambda$ is a Young diagram corresponding to a transpose (or conjugate) partition $\lambda^T$ obtained by enlisting the number of boxes in each column (from left to right) of $\lambda$. On Figure \ref{fig:ytt} we present a transpose of a Young diagram $\lambda=(5,3,2)$ from Figure \ref{fig:yt}, which is $\lambda^T=(3,3,2,1,1)$.

\begin{figure}[H]
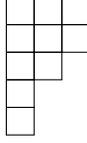

\centering 
\ytableausetup{nosmalltableaux,mathmode,boxsize=1em}
\begin{ytableau}
        \null & \null & \null\\
        \null & \null & \null\\
        \null  & \null\\
        \null \\
        \null
\end{ytableau}
\caption{A transpose $\lambda^T=(3,3,2,1,1)$ of a Young diagram $\lambda=(5,3,2)$.}
\label{fig:ytt}
\end{figure}

Young tableaux are combinatorial objects useful e.g. in representation theory where they can be applied to conveniently describe and study representations of symmetric and general linear groups. In particular, they are handy when working with representations of $\mathrm{SU}(m)$.

All irreducible representations of $\mathrm{SU}(m)$ are in a one-to-one correspondence with Young diagrams that consist of at most $m-1$ rows. The correspondence is established by constructing the so-called Schur module \cite{Fu96}, $E^\lambda$, which gives the canonical form of such a representation. Denote by $E=\mathrm{Span}\{\ket{0},\hdots,\ket{m-1}\}$ the natural representation of $\mathrm{SU}(m)$. Representation $E^\lambda$ of $\mathrm{SU(m})$ described by Young diagram $\lambda=(\lambda_1,\hdots,\lambda_{m-1})$ is a linear subspace of 
\begin{equation}\label{Elambda}
E^\lambda\subset \Lambda^{\mu_1}\left(E\right)\otimes\hdots\otimes\Lambda^{\mu_K}\left(E\right),
\end{equation}
where $\mu_k$ is the length of $k$th column of diagram $\lambda$, i.e. $\mu=(\mu_1,\hdots,\mu_K)=\lambda^T$. The highest weight vector of $E^\lambda$ is of the form
\begin{equation}\label{hw}
\ket{\lambda}=\left(\ket{0}\wedge \ket{1}\wedge\hdots\wedge \ket{\mu_1-1}\right)\otimes\hdots\otimes\left(\ket{0}\wedge\ket{1}\wedge\hdots\wedge \ket{\mu_K-1}\right).
\end{equation}
Two straightforward extreme cases are when $\lambda$ is a single column, i.e. \begin{equation} \lambda=(1,1,\dots,1,0,\dots,0) \coloneqq  (1)^N\mathrm{,}\; N<m \mathrm{,}\end{equation} and when $\lambda$ is a single row, i.e. \begin{equation} \lambda=(N,0,\dots,0) \coloneqq  (N) \mathrm{.}\end{equation} The corresponding representations are 
\begin{equation}
E^{(1)^N}=\Lambda^N \left(E\right),\quad E^{(N)}=S^N \left(E\right) \mathrm{,}
\end{equation}
where by $S^N\left(E\right)$ we denote the $N$th symmetric power of a vector space $E$.

In the following Section \ref{sec2} we make use of the symmetrisation $\widehat{\mc{S}}$ and antisymmetrisation $\widehat{\mc{A}}$ operators. We define their actions on any separable state $\ket{\phi_{1}}_1\otimes\ket{\phi_{2}}_2\otimes\hdots\otimes\ket{\phi_{N}}_N \in \left(\CC^d\right)^{\otimes N}=\mc{H}$ as
\begin{align}
\widehat{\mc{S}}(\ket{\phi_{1}}_1\otimes\ket{\phi_{2}}_2\otimes\hdots\otimes\ket{\phi_{N}}_N) &\coloneqq \sum_{\sigma\in S_N}\ket{\phi_1}_{\sigma(1)}\otimes\ket{\phi_2}_{\sigma(2)}\otimes\hdots\otimes\ket{\phi_{N}}_{\sigma(N)} \mathrm{,}\\
\widehat{\mc{A}}(\ket{\phi_{1}}_1\otimes\ket{\phi_{2}}_2\otimes\hdots\otimes\ket{\phi_{N}}_N) &\coloneqq \sum_{\sigma\in S_N}{\mathrm{sgn}}(\sigma)\ket{\phi_1}_{\sigma(1)}\otimes\ket{\phi_2}_{\sigma(2)}\otimes\hdots\otimes\ket{\phi_{N}}_{\sigma(N)} \mathrm{,}    
\end{align}
where $S_N$ is the permutation group on $N$ elements. Clearly,
$\hat{\mc{S}}(\mc{H})=S^N(\mc{H})$ and $\hat{\mc{A}}(\mc{H})=\Lambda^N(\mc{H})$.

We make an extensive use of Young tableaux in Section \ref{sec:main}, where consider tensor powers of irreducible representations of $\mathrm{SU}(m)$.

\section{Physical settings and other applications}\label{sec2}
In this section we study in detail possible physical settings for which our construction gives quantum states that possess large symmetries. We do this in several steps, starting with the most straightforward construction of multi qudit states with complete diagonal LU-symmetry (Subsection \ref{sub:qudits}). We then move to a slightly more complicated setting that involves distinguishable traps with bosons (Subsection \ref{sub:bosons}) and end up with the most general setting involving traps with fermions possessing internal degrees of freedom (Subsection \ref{sub:fermions}). In Subsection \ref{sub:git} we show that our results have deep implications for the geometric structure of entanglement classes of multi qudit states in terms of geometric invariant theory.


\subsection{Qudit states with complete diagonal LU-symmetry}\label{sub:qudits}
Here we analyse the most straightforward case of $N$ qudits, $\mH=\left(\CC^d\right)^{\otimes N}$ with diagonal $\mathrm{SU}(d)$-symmetry which we call the complete diagonal LU-symmetry. Here, representation $\pi \coloneqq \pi_1=\hdots=\pi_N$ from \eqref{diag-h-sym} is just the natural representation of $\mathrm{SU}(d)$ on $\CC^d$ given by the Young diagram consisting of a single box, i.e. $\lambda=(1)$. State $\kpsi$ is symmetric with respect to this action, i.e. has complete diagonal LU-symmetry \eqref{diag-h-sym}, if and only if $\kpsi$ is annihilated by all generators of $\mathrm{SU}(d)$. Recall that these generators are defined via operators $E_{i,j},\ E_{i,j}^\dagger$, where
\begin{equation}
 E_{i,j} \coloneqq \kb{i}{j},\quad 1\leq i,\, j\leq d \mathrm{,}   
\end{equation}
and $\{\ket{i}|\; 1\leq i \leq d\}$ forms an orthonormal basis in $\mathbb{C}^d$.
More specifically, the Lie algebra $\mk{su}(d)$ is generated by
\begin{equation}
X_{i,j} \coloneqq E_{i,j}+E_{i,j}^\dagger,\quad Y_{i,j} \coloneqq i\left(E_{i,j}-E_{i,j}^\dagger\right),\quad H_{i,j}=E_{i,i}-E_{j,j} \mathrm{,}   
\end{equation}
for all $i,j$ such that $1\leq i< j\leq d$. The above generators are diagonally represented on $\left(\CC^d\right)^{\otimes N}$ as follows
\begin{equation}
\widetilde A \coloneqq A\otimes\bone\otimes\hdots\otimes\bone+\bone\otimes A\otimes\hdots\otimes\bone+\hdots+\bone\otimes \hdots\otimes\bone\otimes A \mathrm{,}    
\end{equation}
for any $A\in \mk{su}(d)$. Hence, $\kpsi$ has complete diagonal LU-symmetry if and only if
\begin{equation}
\widetilde X_{i,j}\kpsi=0,\quad \widetilde Y_{i,j}\kpsi=0,\quad \widetilde H_{i,j}\kpsi=0\quad {\mathrm{for\ all\ }}1\leq i< j\leq d \mathrm{.}   
\end{equation}
\begin{remark}
The above set of equations can be simplified by using the following commutation relations:
\begin{equation}
[\widetilde X_{i,j},\widetilde Y_{i,j}]=-2i\widetilde H_{i,j}, \quad [\widetilde X_{i,j},\widetilde X_{j,k}]=X_{i,k}, \quad [\widetilde Y_{i,j},\widetilde Y_{j,k}]=\widetilde Y_{i,k} \mathrm{.}    
\end{equation} Therefore, it is enough to solve
\begin{equation}\label{sud-eqs}
\widetilde X_{i,i+1}\kpsi=0,\quad \widetilde Y_{i,i+1}\kpsi=0,\quad {\mathrm{for\ all\ }}1\leq i\leq d-1.
\end{equation}
The above simplification is a manifestation of a more general fact which will be described in Section \ref{sec:final}. 
\end{remark}
The dimension of the solution space of \eqref{sud-eqs} is the multiplicity of the trivial representation of $\mathrm{SU}(d)$ in the product $\pi^{\otimes N}$. The resulting dimensions for small $d$ and $N$ are shown in Table \ref{tab:natural}.

\begin{table}[H]
\centering
\begin{tabular}{c|cccccccccccc}
\multirow{2}{*}{$d$} & \multicolumn{11}{c}{$N$} \\
&1 & 2 & 3 & 4 & 5 & 6 & 7 & 8 & 9 & 10 & 11 & 12\\
\hline 
\hline 
2 & 1 & 1 & 0 & 2 & 0 & 5 & 0 & 14 & 0 & 42 & 0  & 132\\
\hline 
3 & 1 & 0 & 1 & 0 & 0 & 5 & 0 & 0 & 42 & 0 & 0 & 462 \\
\hline
4 & 1 & 0 & 0 & 1 & 0 & 0 & 0 & 14 & 0 & 0  & 0 & 462 \\
\hline
5 & 1 & 0 & 0 & 0 & 1 & 0 & 0 & 0 & 0 & 42 & 0 & 0 \\
\hline
6 & 1 & 0 & 0 & 0 & 0 & 1 & 0 & 0 & 0 & 0 & 0 & 132
\end{tabular}
\caption{Multiplicities of the trivial irrep of $\mathrm{SU}(d)$ in $\left(\CC^d\right)^{\otimes N}$ for small $d$ and $N$.}
\label{tab:natural}
\end{table}
There are two key features of states with complete diagonal LU-symmetry that can be readily seen from Table \ref{tab:natural}.
\begin{itemize}
\item States with complete diagonal LU-symmetry exist if and only if $N=kd$ for some integer $k$. This can be seen as a consequence of Littlewood-Richardson rules introduced in Section \ref{sec:main}.
\item The dimension of the space of states with complete diagonal LU-symmetry increases with $N$. In fact, it is given by multidimensional Catalan numbers, as explained in Subsection \ref{sub:sud}.
\end{itemize}
If $N=kd$, it is straightforward to write down an example of the corresponding LME state with complete diagonal LU-symmetry. It has the form of the tensor product
\begin{equation}
\kpsi=\ket{\Psi_{0}}\otimes \ket{\Psi_1}\otimes\hdots\otimes\ket{\Psi_{k-1}} \mathrm{,}    
\end{equation}
where state $\ket{\Psi_i}$ is the completely antisymmetric state with respect to permutations of qudits $di+1$ through $d(i+1)$.  Using $\widehat{\mc{A}}$ we can write down the state $\kpsi$ as
\begin{align}
\begin{split}
\kpsi=\left(\widehat{\mc{A}}\left(\ket{1}_1\otimes\ket{2}_2\otimes\hdots\otimes\ket{d}_d\right)\right)\otimes \left(\widehat{\mc{A}}\left(\ket{1}_{d+1}\otimes\hdots\otimes\ket{d}_{2d}\right)\right)\otimes\hdots \\
\hdots\otimes\left(\widehat{\mc{A}}\left(\ket{1}_{(k-1)d+1}\otimes\hdots\otimes\ket{d}_{kd}\right)\right).
\end{split}\label{state-complete}
\end{align}

\subsection{A system of traps with bosons}\label{sub:bosons}
\label{ex:bosons}
Consider a quantum system consisting of three distinguishable traps enumerated by $i=1,2,3$, each containing two bosons that can occupy three modes denoted by creation operators $a^\dagger_i,\ b^\dagger_i,\ c^\dagger_i$ (Fig.\ref{traps}). The Hilbert space of a single trap is
  \begin{figure}[h]
\centering
\includegraphics[width=.7\textwidth]{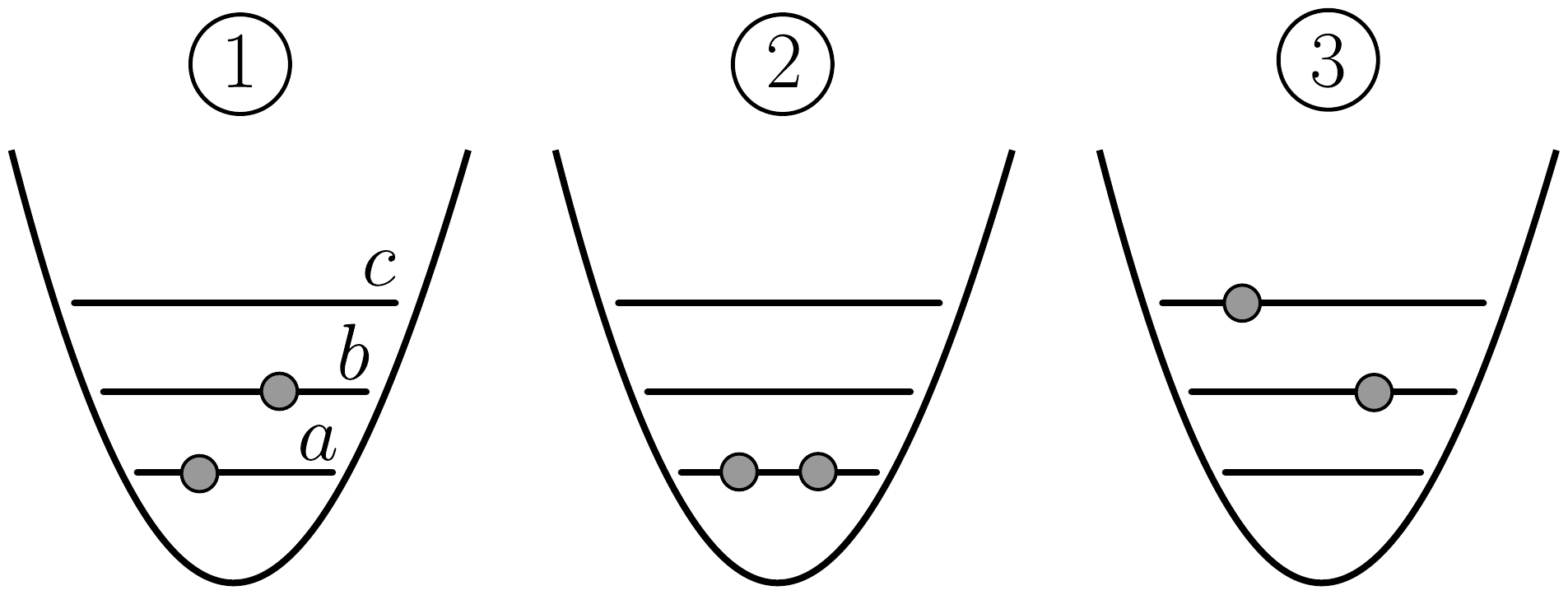}
\caption{Three traps with bosons.}
\label{traps}
\end{figure}
\begin{equation}
    \mH_1=S^2\left(\CC^3\right)=\left\langle\ket{n_a,n_b,n_c}|\ n_a+n_b+n_c=2\right\rangle\simeq \CC^6 \mathrm{,}
\end{equation}

which will be treated as one qudit with six degrees of freedom. The total Hilbert space consists of three such qudits
\begin{equation}\mH=\mH_1\otimes\mH_1\otimes\mH_1\mathrm{.}\end{equation}
The set of all LOCC operations is represented as triples of $6\times 6$ matrices $A_1,A_2,A_3$ that act on $\kpsi\in\mH$ as $\kpsi\mapsto A_1\otimes A_2\otimes A_3 \kpsi$. Physically, realising such operations requires introducing some interactions between bosons within a single trap to obtain two-body operators of the form $\alpha_i^\dagger\beta_i^\dagger\gamma_i\delta_i$ where $\{\alpha,\beta,\gamma,\delta\}\subset\{a,b,c\}$ and $i=1,2,3$.

We are interested in finding an LME state with the following symmetries. Symmetries are the local diagonal mode special unitary operations, i.e. $\kpsi\mapsto U\otimes U\otimes U\kpsi$, where a single $3\times 3$ unitary matrix $[U_{\alpha\beta}]$ of determinant one is represented on $i$th trap as 
\begin{equation}U=\sum_{\{\alpha,\beta\}\subset\{a,b,c\}} U_{\alpha\beta}\alpha_i^\dagger\beta_i \mathrm{,}\end{equation}
which is effectively a $6\times 6$ unitary matrix when written in the basis $\{\ket{n_a,n_b,n_c}\}$. This corresponds to the situation where $\pi$ from equation \eqref{diag-h-sym} is the representation of $\mathrm{SU}(3)$ corresponding to the Young diagram that consists of one row with two boxes, $\lambda=(2)$.

A state $\kpsi\in \mH$ is symmetric with respect to local diagonal mode operators if for all operators $U$ we have
\begin{equation} U\otimes U\otimes U\kpsi=\kpsi \mathrm{.}\end{equation}
This in turn happens when $\kpsi$ is annihilated by all generators of such operations. This boils down to the following set of six linear equations.
\begin{gather}\label{bosons:eqs}
\left(\sum_{i=1}^3\alpha^\dagger_i\beta_i\right)\kpsi=0,\quad\{\alpha,\beta\}\subset\{a,b,c\}.
\end{gather}
The solution is a state which is symmetric with respect to permutations of traps. Using $\widehat{\mc{S}}$ the solution can be written as
\begin{gather}
\kpsi=\widehat{\mc{S}}\Big{(}\ket{200}_1\ket{020}_2\ket{002}_3+\frac{1}{\sqrt{2}}\ket{110}_1\ket{101}_2\ket{011}_3+\\-\frac{1}{2}\left(\ket{200}_1\ket{011}_2\ket{011}_3+\ket{020}_1\ket{101}_2\ket{101}_3+\ket{002}_1\ket{110}_2\ket{110}_3\right)\Big{)}.
\end{gather}
For so defined state, the norm reads $\sqrt{\bk{\Psi}{\Psi}}=3\sqrt{2}$. State $\kpsi$ is a state which is locally maximally entangled with respect to the three traps, i.e. $\rho_i={\rm diag}(\frac{1}{6},\hdots,\frac{1}{6})$ for $i=1,2,3$, where $\rho_i$ is the single-qudit reduced density matrix coming from the reduction of $\rho=\kb{\Psi}{\Psi}$ with respect to the remaining two qudits (traps). Note that $\rho_i$ should not be confused with the reduced density matrix of a single boson. The fact that $\kpsi$ is LME can be checked directly or simply deduced from the fact that $\kpsi$ has a diagonal $\mathrm{SU}(3)$-symmetry.  

Notice that, since each trap is effectively a qudit with $d=6$, and $\mathrm{SU}(3)$ has an irrep in $\mathbb{C}^6$, from Corollary \ref{cor:lme}, $\mathcal{H}$ had to contain some LME state with diagonal $\mathrm{SU}(3)$-symmetry.

One may check the multiplicity of a trivial representation in $\mathcal{H}$ as a carrier space of local diagonal mode special unitary operations. A straightforward calculation involving Young diagrams shows that the multiplicity is 1.

The above physical setting can be generalised to a system of $m$ distinguishable traps with each trap containing $n$ bosons occupying $m$ modes. In that case, a single trap realises a qudit with $d=\binom{m+n-1}{n}$. This scenario realises representation $E^\lambda$ of $\mathrm{SU}(m)$ with a single-row diagram $\lambda=(n)$ as a diagonal LU-symmetry of the corresponding state. The explicit form of the state can be found by solving 
\begin{gather}\label{bosons:eqs-general}
\left(\sum_{i=1}^m\alpha^\dagger_i\beta_i\right)\kpsi=0,\quad\{\alpha,\beta\}\subset\left\{a^{(1)},\hdots,a^{(m)}\right\},
\end{gather}
where $a^{(k)}_i$ is the annihilation operator of $k$th mode in $i$th trap.

\subsection{Spinful fermions and beyond}\label{sub:fermions}
In this subsection, we introduce the most general construction capturing LME states with a diagonal LU-symmetry given by an arbitrary representation of $\mathrm{SU}(m)$. We start with the single-particle Hilbert space which describes a particle with $s$ internal degrees of freedom (s-charge) that can occupy $m$ modes in a single trap
\begin{equation}
 \mH_1=\CC^m\otimes \CC^s \mathrm{.}  
\end{equation}
In particular, if $s=2$, then we think of a single electron (spin-$1/2$ fermion) occupying $m$ orbitals. Next, we confine $N$ such fermions in one trap. The total Hilbert space for the system of $N$ fermions in one trap is given by the $N$th exterior power
\begin{equation}
\mH_N=\Lambda^{N} \left(\CC^m\otimes \CC^s\right).
\end{equation}
Hilbert space $\mH_N$ is equipped with the natural action of $\mathrm{SU}(m)\times \mathrm{SU}(s)$ that changes the basis of orbitals and internal degrees of freedom of each particle simultaneously. Let us next consider an abstract (possibly interacting) hamiltonian of the above system of trapped $N$ fermions and assume that it commutes with the total $s$-charge operator. In the representation-theoretic language, we require the hamiltonian to commute with the Casimir operator (the sum of squared elements of an orthonormal basis \cite{hall}) of the corresponding Lie algebra $\mathfrak{su}(s)$. For instance, if the considered fermions were ordinary electrons ($s=2$), the above symmetry would be the well-known total spin conservation symmetry. If this is the case, Hilbert space $\mH_N$ decomposes into sectors of fixed total $s$-charge each of which is isomorphic to the tensor product of an irrep of $\mathrm{SU}(m)$ with an irrep of $\mathrm{SU}(s)$ \cite{Klyachko}
\begin{equation}
\mH_N\cong\bigoplus_{\lambda:\ |\lambda|=N} \mathcal{H}_{m}^{\lambda} \otimes \mathcal{H}_{s}^{\lambda^T} \mathrm{,}   
\end{equation}
where $\mathcal{H}_{m}^{\lambda}$ is the irreducible representation of $\mathrm{SU}(m)$ corresponding to Young diagram $\lambda$ and $\mathcal{H}_{s}^{\lambda^T}$ is the irreducible representation of $\mathrm{SU}(s)$ corresponding to $\lambda^T$. The sum is over all Young diagrams with $N$ boxes and at most $m-1$ rows and at most $s$ columns. Again, for spinful electrons we require $\lambda$ to be a two-column diagram and $\mathcal{H}_{s}^{\lambda^T}$ is the representation of $\mathrm{SU}(2)$ with total spin equal to half of the difference of columns' lengths. Finally, we would like to superselect the particular value of the $s$-charge. This can be done in general by requiring the hamiltonian to commute with appropriate additional generators of $\mathrm{SU}(s)$. In the spin case, we desire the total conservation of the $z$-component of spin angular momentum $S_z$. In this way, we are able to superselect the sector of $\mH_N$ that consists of states of the form $\kpsi=\ket{\Phi}\otimes \ket{i}$, $\ket{\Phi}\in \mathcal{H}_{m}^{\lambda}$, for fixed $i\in\{1,\hdots,s\}$. This space is isomorphic to $\mathcal{H}_{m}^{\lambda}$ and carries the natural action 
\begin{equation}\label{rep-superselected}
\pi:\ \mathrm{SU}(m)\to \mathrm{SU}\left(\mathcal{H}_{m}^{\lambda}\right)
\end{equation}
stemming from the basis change of modes within a single trap. 

By considering a system of $m$ distinguishable traps, where the Hilbert space for every single trap is the above superselected $\mathcal{H}_{m}^{\lambda}$, we obtain an abstract physical setting that contains LME states with diagonal $\mathrm{SU}(m)$-symmetry given by the $m$-fold product of representation $\pi$ from \eqref{rep-superselected}. Finally, let us remark that, similarly to the case of traps with bosons described in Subsection \ref{sub:bosons}, every trap is treated as a qudit with local dimension $d=\dim\mathcal{H}_{m}^{\lambda}$. Hence, performing SLOCC operations on a full qudit requires interactions between fermions within one trap.

\begin{example}[The doublet space of $3$ electrons with $m=3$]
As an example, consider a system of $m=3$ traps, where the Hilbert space of a single trap is given by the doublet space of $3$ electrons occupying $3$ modes with total $S_z=+1/2$. Such a space comes from the following decomposition
\begin{equation}
\Lambda^{N} \left(\CC^3\otimes \CC^2\right)\cong\mH_3^{(2,1)}\otimes \CC^2\quad\oplus\quad \mH_3^{(1,1,1)}\otimes S^2\left(\CC^2\right) \mathrm{,}
\end{equation}
where $\mH_3^{(2,1)}\otimes \CC^2$ is the doublet component corresponding to the total spin of $1/2$ and $\mH_3^{(1,1,1)}\otimes S^2\left(\CC^2\right)$ is the quadruplet component corresponding to the total spin of $3/2$. Furthermore, by superselecting the total $S_z=+1/2$ within the doublet space, we obtain the Hilbert space of a single trap which is isomorphic to the irrep of $\mathrm{SU}(3)$ given by diagram $\lambda=(2,1)$. Space $\mH_3^{(2,1)}$ is of dimension $8$ and for completeness, we write down its (non-orthonormalised) basis as a subspace of $\Lambda^{N} \left(\CC^3\otimes \CC^2\right)$
\begin{gather}
\ket{1\uparrow}\wedge\ket{1\downarrow}\wedge\ket{2\uparrow},\quad \ket{1\uparrow}\wedge\ket{1\downarrow}\wedge\ket{3\uparrow},\quad \ket{2\uparrow}\wedge\ket{2\downarrow}\wedge\ket{1\uparrow}, \nonumber \\
\ket{2\uparrow}\wedge\ket{2\downarrow}\wedge\ket{3\uparrow},\quad \ket{3\uparrow}\wedge\ket{3\downarrow}\wedge\ket{1\uparrow},\quad \ket{3\uparrow}\wedge\ket{3\downarrow}\wedge\ket{2\uparrow}, \\
\ket{1\uparrow}\wedge\ket{3\downarrow}\wedge\ket{2\uparrow}-\ket{2\uparrow}\wedge\ket{1\downarrow}\wedge\ket{3\uparrow},\quad 
\ket{1\uparrow}\wedge\ket{3\downarrow}\wedge\ket{2\uparrow}+\ket{1\uparrow}\wedge\ket{2\downarrow}\wedge\ket{3\uparrow}.\nonumber
\end{gather}
Finally, equations for the desired LME state distributed across three traps $\kpsi\in\mH_3^{(2,1)}\otimes\mH_3^{(2,1)}\otimes\mH_3^{(2,1)}$ read
\begin{gather}\label{bosons:eqs-fermions}
\left(\sum_{i=1}^3\alpha^\dagger_i\beta_i\right)\kpsi=0,\quad\{\alpha,\beta\}\subset\{a^{(1)},a^{(2)},a^{(3)}\},
\end{gather}
where $a^{(k)}_i$ is the annihilation operator of $k$th mode in $i$th trap.

\end{example}

\subsection{Existence of strictly semistable states}\label{sub:git}
\label{sec:ssex}
According to the Kempf-Ness theorem from geometric invariant theory, a SLOCC class contains an LME state iff it is closed (in the standard complex topology) \cite{KeNe79}. We call states within a SLOCC class containing an LME state \emph{polystable} states. \emph{Strictly semistable}  are states whose SLOCC classes contain LME states only in the closure.  It turns out that strictly semistable states exist iff there exists an LME state with more symmetries (in the sense of stabiliser dimension) than generic (i.e. there are at least two LME states with different stabiliser dimensions; see Theorem 1 and Observation 3 from \cite{SlHe20} and Fig. 2 therein which explains the relationship between the symmetries of LME states and the structure of SLOCC classes). It is well-known that Hilbert spaces of systems of two qudits and three qubits contain open and dense SLOCC orbits going through LME states, so such systems do not have strictly semistable states. In the remaining cases the generic stabiliser is trivial or discrete \cite{SaWa18}, so zero-dimensional. On the other hand such systems contain $GHZ$ states that have continuous symmetries (and are LME). Thus, applying the criterion from \cite{SlHe20}, they contain strictly semistable states. Such strictly semistable states can be asymptotically transformed to the $GHZ$ state. We show that in some cases one may deduce the existence of the special classes of strictly semistable states, namely the ones that can be asymptotically transformed to an LME state with a $\mathrm{SU}(N)$-symmetry ($N$ being number of qudits). Specifically, combining Theorem \ref{th:main} with Corollary \ref{cor:lme} and results from \cite{SlHe20}, we obtain the following lemma.
\begin{lemma}
\label{lem:ssex}
Consider $\mathcal{H}_N=(\CC^{d})^{\otimes N}$, where $N \geq 4$ or ($N=3$ and $d \geq 3$). Suppose that there exists a non-trivial irreducible continuous representation $\pi:\ \mathrm{SU}(N)\to \mathrm{SU}(d)$. Then $\mathcal{H}_N$ contains an LME state $\ket{\Phi}$ with diagonal $\mathrm{SU}(N)$-symmetry with respect to $\pi^{\otimes N}$, as defined in \eqref{diag-h-sym}. Moreover $\mathcal{H}_N$ contains strictly semistable states that can be asymptotically transformed to $\ket{\Phi}$.
\end{lemma}
\begin{proof}
Recall that $G_{\ket{\Psi}} $ denotes the SLOCC stabiliser of $\kpsi$ while $K_{\ket{\Psi}} $ denotes the LU stabiliser of $\kpsi$.  Because LME states in $\mathcal{H}_N$ exist, there exist also polystable states \cite{SlHe20}. A generic polystable state in $\mathcal{H}_N$, say $\ket{\Psi}$, is known to have no continuous symmetries with respect to the action of $G$ \cite{SaWa18} . Let $\ket{\Psi_c}$ by an LME state in the SLOCC class of $\ket{\Psi}$. We have $\mathrm{dim} G_{\ket{\Psi_c}} = 0$. From Theorem \ref{th:main}, $\mathcal{H}_N$ contains a copy of the trivial representation of $\mathrm{SU}(N)$. Hence, from Corollary \ref{cor:lme}, $\mathcal{H}_N$ contains an LME state $\ket{\Phi}$ with diagonal $\mathrm{SU}(N)$-symmetry with respect to $\pi^{\otimes N}$. It remains to check that $\mathrm{dim}\, K_{\ket{\Phi}} \geq 1$. Indeed, in this case since $K < G$ also $\mathrm{dim}\, G_{\ket{\Phi}} \geq 1$ so $\ket{\Phi}$ and $\ket{\Psi_c}$ are two LME states with different $G$-stabiliser dimensions. Thus, $\ket{\Phi}$ is a state to which some strictly semistable states can be asymptotically transformed \cite{SlHe20}. Let us denote $H=\mathrm{SU}(N)$. Note that $H_{\ket{\Phi}}=H$ so $\pi(H) < K_{\ket{\Phi}}$. Since $K_{\ket{\Phi}}$ is closed, it contains the closure $\overline{\pi(H)}$ which is an embedded connected Lie subgroup of $\mathrm{SU}(d)$. Moreover, $\mathrm{dim}\,\overline{\pi(H)} \geq 1$ since otherwise $\overline{\pi(H)}$ would be discrete and connected, so a singleton. Thus, $\pi$ would be trivial.
\end{proof}

Our work also links to another facet of geometric invariant theory which is the problem of constructing SLOCC-invariant polynomials. Namely, homogeneous SLOCC-invariant polynomials of degree $d$ with respect to an irreducible representation $E^\lambda$ of $\mathrm{SU}(m)$ can be viewed as trivial components in the decomposition of $S^d\left(E^\lambda\right)$ into irreps. In the spirit of our Theorem \ref{th:main}, one can ask about the smallest $d$ for which $S^d\left(E^\lambda\right)$ contains the trivial irrep of $\mathrm{SU}(m)$ with non-zero multiplicity. For a way to derive lower bounds for such $d$, see \cite{Tsanov}.

\section[xxx]{Multiplying representations of $\mathrm{SU}(m)$}\label{sec:main}
\label{sec:mult}
In this section we recall some facts about multiplication of representations of $\mathrm{SU}(m)$ and formulate our main result - Theorem \ref{th:main}.

\begin{fact}\label{trivial}
$E^\lambda$ is the trivial representation of $\mathrm{SU}(m)$ iff $\lambda$ is a rectangular diagram with $m$ rows, i.e. $\lambda=(\lambda_1)^m$.
\end{fact}
\begin{fact}\label{power}
The $N$th tensor power of the natural representation contains all representations with Young diagrams of length $N$ \cite{Fu96}, i.e. $|\lambda| \coloneqq \sum_i\lambda_i=N$,
\begin{equation}
E^{\otimes N}=\bigoplus_{\lambda:\ |\lambda|=N} c_{\lambda} E^\lambda.
\end{equation}
The coefficients $c_{\lambda}$ can be calculated using the Littlewood-Richardson rule (see e.g. \cite{StTh01, FuHa91, Fu96}). When specified for $\mathrm{SU}(m)$, we have $c_\lambda=0$ iff $\lambda$ has more than $m$ rows. 
\end{fact}
The above Facts \ref{trivial} and \ref{power} allow us to formulate a necessary condition for the $N$th tensor power of a given representation $E^\lambda$ to contain a trivial representation of $\mathrm{SU}(m)$. Namely, representation $E^\lambda$ appears as an irreducible component of $E^{\otimes |\lambda|}$. Hence, if the trivial representation appears as an irreducible component of $\left(E^\lambda\right)^{\otimes N}$, then it appears as an irreducible component of $E^{\otimes |\lambda| N}$. Furthermore, product $E^{\otimes |\lambda| N}$ contains trivial representation $E^\tau$ with diagram $\tau=(\tau_1)^m$ iff $|\lambda| N=m\tau_1$. Hence, we have proved the following lemma.
\begin{lemma}\label{lemma:necessary}
If the $N$th tensor power of $E^\lambda$, an irreducible representation of $\mathrm{SU}(m)$, contains a copy of the trivial representation, then $|\lambda| N$ is an integer multiple of $m$.
\end{lemma}
The simplest way to satisfy the above necessary condition is to put $N=m$. In the remaining part of this section we show that this is, in fact, a sufficient condition (see Theorem \ref{th:main}).

Let us take a closer look at the multiplication rules for Young diagrams. In general, when multiplying two irreducible representations we get
\begin{equation}
E^\lambda\otimes E^\eta=\bigoplus_{\nu}c^{\lambda\eta}_\nu E^\nu,
\end{equation}
where $c^{\lambda\eta}_\nu$ is the Littlewood-Richardson (also known as Clebsch-Gordan) coefficient. Representations $E^\nu$ that appear in the above product with non-zero coefficients can be determined by the following set of rules.
\begin{fact}\label{multiplication-rules}[Littlewood-Richardson rule]
Here we describe how to find the Littlewood-Richardson coefficients by diagram expansions, according to the Littlewood-Richardson rule \cite{StTh01, FuHa91, Fu96}. This particular version of the rule can be found in \cite{Kir16}. We start with two irreducible representations of $\mathrm{SU}(m)$ whose corresponding diagrams are $\lambda$ and $\eta$. In order to find the irreducible representations (and their multiplicities) that appear in the product $E^\lambda\otimes E^\eta$, we draw the two diagrams next to each other and fill the second diagram, $\eta$, with numbers so that boxes of the $k$th row contain only integer $k$. To obtain diagram $\nu$ that corresponds to an irreducible component of $E^\lambda\otimes E^\eta$, we expand the diagram of $\lambda$ by appending all boxes of diagram $\eta$ to diagram $\lambda$ according to the following rules. The boxes can be appended only to the right or to the bottom of $\lambda$. Moreover, the following conditions need to be satisfied.
\begin{enumerate}
\item The consecutive rows of the resulting diagram have non-increasing lengths, i.e. $\nu_i\geq \nu_{i+1}$. This means they are Young diagrams.
\item In any column of $\nu$, there are no two boxes with the same label.
\item Let $\#_r(m)$ denote the number of boxes of $\nu$ with label $m$ in rows 1 up to $r$ (numerated from top to bottom). Then, for every row $r$ if $n<m$ then $\#_r(m) \leq \#_r(n)$. This is \textit{row counting condition}.
\item Let $\#_c(m)$ denote the number of boxes of $\nu$ with label $m$ in columns 1 up to $c$ (numerated from right to left). Then, for every column $c$ if $n<m$ then $\#_c(m) \leq \#_c(n)$. This is \textit{column counting condition}.
\item  Diagram $\nu$ has at most $m$ rows.
\end{enumerate}
The Littlewood-Richardson coefficient of a given diagram will be the number of its occurrences among all possible diagrams obtained via above procedure.
\end{fact}

It is worth noting that there is another insightful yet equivalent way one can deal with the irreps of $\mathrm{SU}(m)$, namely via Gelfand-Tsetlin pattern calculus (see the original paper \cite{GeTs50} and the reprinted version in \cite{GeMi63}). Using this approach, a numerical algorithm for the explicit calculations of $\mathrm{SU}(m)$ Clebsch-Gordan coefficients has been given in \cite{ArHu11} (see also \cite{BiLo68}).

In the following Examples \ref{ex:tableau2}, \ref{ex:tableau3} and \ref{ex:tableau4} we denote the rectangles with symbols of the form $\lambda_n$ or $\lambda_n \times m$. A symbol $\lambda_n$ means that a rectangle is formed by a row of $\lambda_n$ boxes with some fixed label. A symbol $\lambda_n \times m$ additionally specifies that this label is $m$.

\begin{example}
\label{ex:tableau2}
The square of any irreducible representation of $\mathrm{SU}(2)$ contains a copy of the trivial representation. The proof is shown on the picture below, where we show that $E^\lambda\otimes E^\lambda\supset E^{(|\lambda|,|\lambda|)}$ for any single-row diagram $\lambda$.
\begin{figure}[H]
\centering 
\includegraphics[width=.7\linewidth]{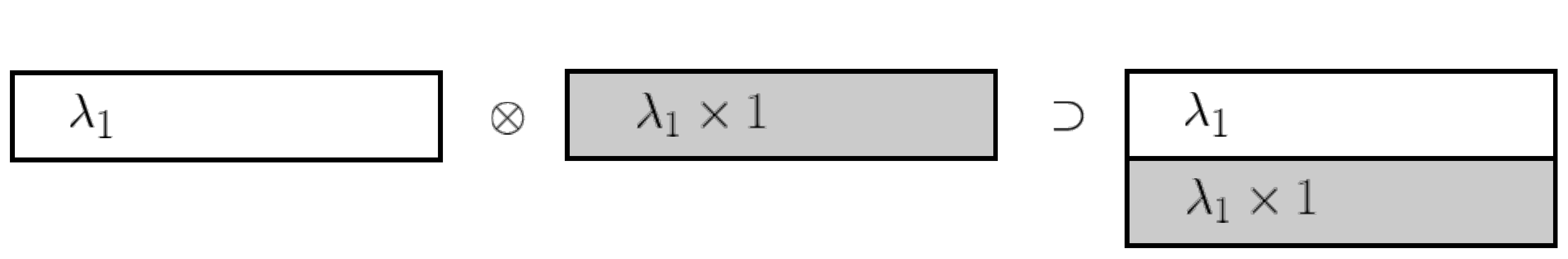}  
  \label{fig:ex1}
\end{figure}
\end{example}

\begin{example}
\label{ex:tableau3}
Let $E^\lambda$ be an irreducible representation of $\mathrm{SU}(3)$. Then, $\left(E^\lambda\right)^{\otimes 3}$ contains a copy of the trivial representation. The Young diagram of an irreducible representation of $\mathrm{SU}(3)$ has at most $2$ rows, i.e. $\lambda=(\lambda_1,\lambda_2)$. First, let us show that $E^\lambda\otimes E^\lambda\supset E^{(\lambda_1+\lambda_2,\lambda_1,\lambda_2)}$.
\begin{figure}[H]
\centering 
\includegraphics[width=.8\linewidth]{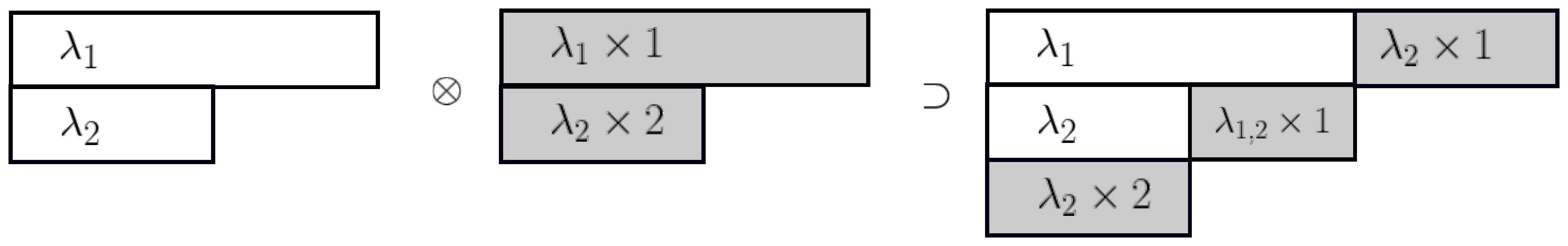}  
  \label{fig:ex2}
\end{figure}

Finally, we show that $E^{(\lambda_1+\lambda_2,\lambda_1,\lambda_2)}\otimes E^\lambda$ contains the irreducible representation with rectangular diagram $(|\lambda|,|\lambda|,|\lambda|)$ by appending $\lambda_2$ $1$'s to the penultimate row and $\lambda_1-\lambda_2$ $1$'s and $\lambda_2$ $2$'s to the bottom row.
\begin{figure}[H]
\centering 
\includegraphics[width=.8\linewidth]{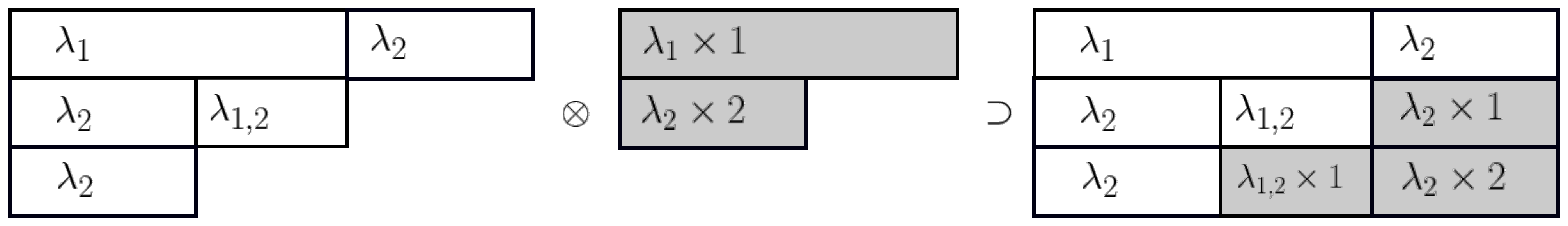}  
  \label{fig:ex2.1}
\end{figure}
\end{example}

\begin{example}
\label{ex:tableau4}
 Let $E^\lambda$ be an irreducible representation of $\mathrm{SU}(4)$. Then, $\left(E^\lambda\right)^{\otimes 4}$ contains a copy of the trivial representation. The Young diagram of an irreducible representation of $\mathrm{SU}(4)$ has at most $3$ rows, i.e. $\lambda=(\lambda_1,\lambda_2,\lambda_3)$. We proceed similarly as in Examples \ref{ex:tableau2} and \ref{ex:tableau3}.
\begin{figure}[H]
\centering 
\includegraphics[width=\textwidth]{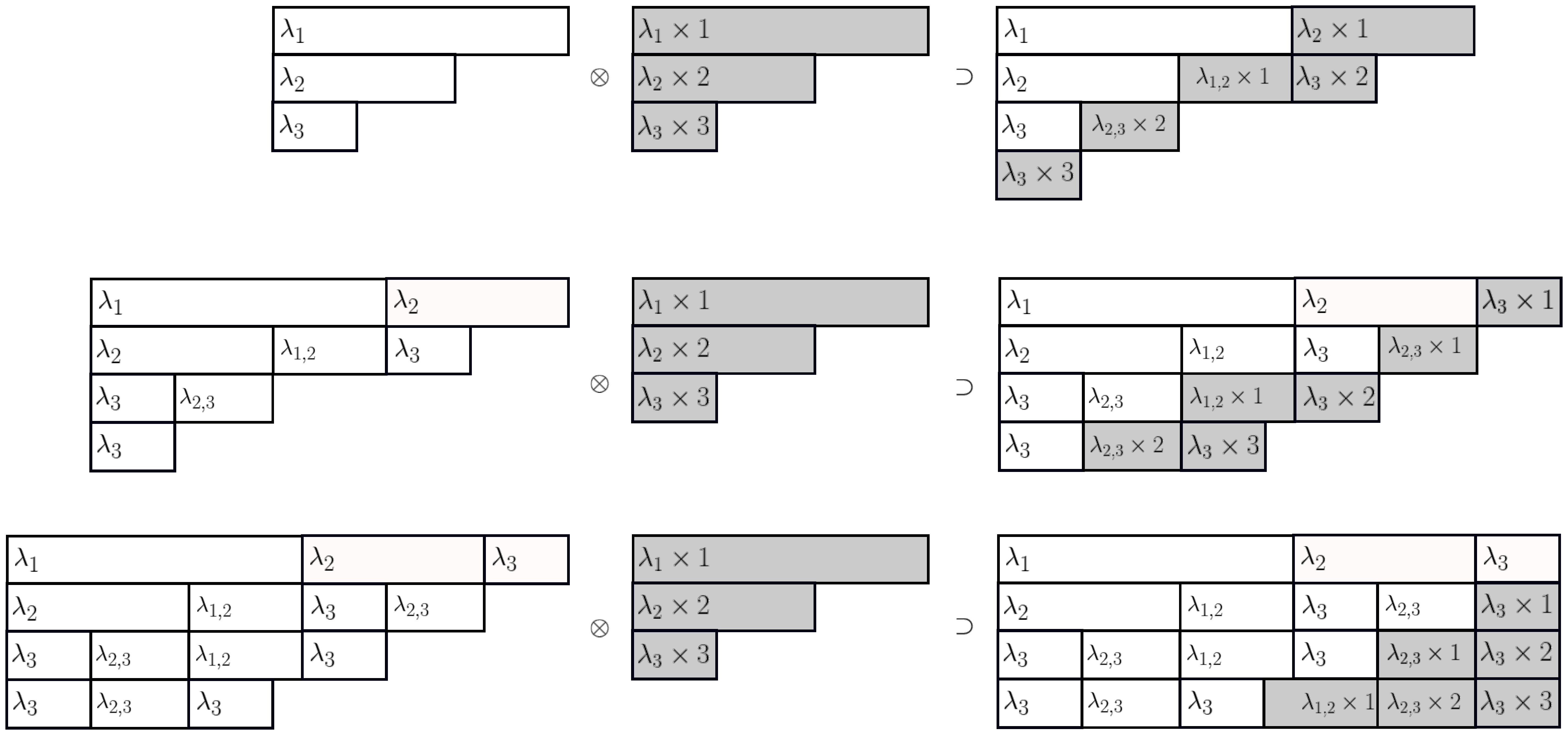}  
  \label{fig:ex3}
\end{figure}
\end{example}

In general, it is hard to prove a theorem which is valid for a Young diagram of any shape. However, we will show that in fact Examples \ref{ex:tableau2}, \ref{ex:tableau3} and \ref{ex:tableau4} can be generalised to any dimension $N$, and Young diagram $\lambda$, i.e. the following theorem holds.

\begin{theorem}
\label{th:main}
Let $E^\lambda$ be an irreducible representation of $\mathrm{SU}(N)$. Then, $\left(E^\lambda\right)^{\otimes N}$ contains a copy of the trivial representation.
\end{theorem}

\begin{proof}
The proof is divided into two main parts.
\begin{itemize}
    \item Part I: Ansatz for the solution - we provide an ansatz for the solution, i.e. similarly as in Examples \ref{ex:tableau2}, \ref{ex:tableau3} and \ref{ex:tableau4} we give an explicit construction of the sequence of diagram expansions leading to a rectangular diagram with $N$ rows (see Fact \ref{trivial}).
    \item Part II: Verification of correctness - we verify our ansatz, i.e. we check that the proposed sequence of diagram expansions is valid and gives a rectangular diagram with $N$ rows and shape $(|\lambda|,|\lambda|,\ldots,|\lambda|)$. 
\end{itemize}
 A roadmap of the proof is given in Figure \ref{fig:roadmap} from \nameref{app}. Below we cover Part I and in Part II we provide a set of five conditions (1-5) which have to be checked to conclude the proof. These checks can be found in \nameref{app}.\\
\textbf{Proof of Part I}\\
 We proceed as follows. In every step we first add boxes with label 1, then with label 2 etc. up to label $N-1$. Moreover, for every label we add the boxes to the consecutive rows i.e. we start with some top row $i$ and append a certain number of boxes to $q \geq 1$ rows from $i$th to $i+q-1$th row. In every step we need to specify the position, label and number of appended boxes. The idea is to append boxes in such a way that the exact number of boxes with a given label $p$ added to a given row is uniquely determined by $p$ and the total number of rows $q$ to which we append boxes with a given label.

In order to keep track of the number and type of boxes appended in every step, we define symbols

\begin{equation}
    T_{p,q}^r =
    \begin{cases*}
      \lambda_{p+r-1}-\lambda_{p+r} & if $1 \leq r < q$,\\
      \lambda_{p+q-1} & if $1 \leq r = q$,\\
      0       & otherwise.
    \end{cases*}\quad\mathrm{,}
  \end{equation}
which can be understood as entries of columns of numbers

  \begin{equation}
  \label{eq:tel}
  T_{p,q}=
  \begin{pmatrix}
    \lambda_{p+q-1}\\
    \lambda_{p+q-2}-\lambda_{p+q-1}\\
    \vdots\\
    \lambda_{p+1}-\lambda_{p+2}\\
    \lambda_{p}-\lambda_{p+1}
  \end{pmatrix}\mathrm{,}
  \end{equation}  
we call \textit{telescopes}. The name telescope is motivated by the fact that partial sums of entries of a telescope is a telescopic sum (see Section \ref{sec:cond3}). A symbol $T_{p,q}^r$ denotes the $r$th entry (counted from bottom) of a telescope $T_{p,q}$. We call $p$ the label and $q$ the length of a telescope (which is the number of its entries). For example $T_{2,3}$ is a telescope with label 2 and length 3. We will represent a telescope $T_{p,q}$ graphically as a column of $q$ boxes with label $p$, one box for each entry of a telescope. We call such boxes \textit{virtual}. For example, 

\begin{equation}
T_{2,3}=
  \begin{pmatrix}
    \lambda_{4}\\
    \lambda_{3}-\lambda_{4}\\
    \lambda_{2}-\lambda_{3}
  \end{pmatrix} \equiv \;
  \ytableausetup{nosmalltableaux,mathmode,boxsize=1em}
  \begin{ytableau}
       2  \\
       2    \\
       2 \\
       \none
\end{ytableau} \;\mathrm{,}    
\end{equation}
can be understood as three rows of boxes with labels 2. The $r$th virtual box (counted from bottom) of a column corresponds to a row of $T_{p,q}^r$ boxes with label $p$.

We will encode the way we expand diagrams using virtual boxes. Each column of virtual boxes corresponds to some telescope $T_{p,q}$. For each step $k$, we construct a pattern which corresponds to a sequence of telescopes, one for each label $m$. We append the boxes to the diagram label by label with $m$ increasing. Each virtual box in the column (i.e. entry of a telescope) corresponds to a particular number of boxes appended to some (yet not specified) row $l$ of an expanding diagram. Consecutive virtual boxes correspond to consecutive columns of the diagram. An exact number of boxes $m$ appended to the row $l$ is determined by a number of virtual boxes in a column, the label $m$ of the virtual boxes and the position of a virtual box in the column - according to (\ref{eq:tel}).

For example, the patterns encoding expansions form Examples \ref{ex:tableau3} and \ref{ex:tableau4} are

\begin{figure}[H]
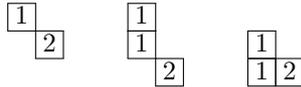

\centering 
\ytableausetup{nosmalltableaux,mathmode,boxsize=1em}
\begin{ytableau}
     1 & \none & \none  \\
    \none & 2 & \none \\
    \none & \none & \none
\end{ytableau}
\quad
\begin{ytableau}
        1 & \none & \none  \\
        1 & \none  & \none \\
          \none & 2 & \none  
\end{ytableau}
\quad
\begin{ytableau}
        \none & \none & \none  \\
        1 & \none & \none \\
         1 & 2 & \none  
\end{ytableau}
\label{fig:N=3}
\caption{The pattern of virtual boxes encoding the diagram expansions from Example \ref{ex:tableau3}.}
\end{figure}
and 
\begin{figure}[H]
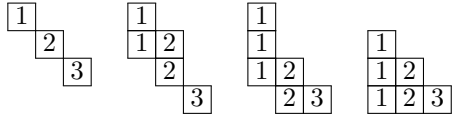

\centering 
\ytableausetup{nosmalltableaux,mathmode,boxsize=1em}
\begin{ytableau}
     1 & \none & \none  \\
    \none & 2 & \none \\
    \none & \none & 3 \\
    \none & \none & \none
\end{ytableau}
\quad
\begin{ytableau}
    1 & \none & \none  \\
    1 & 2 & \none \\
    \none & 2 & \none \\
    \none & \none & 3 
\end{ytableau}
\quad
\begin{ytableau}
    1 & \none & \none  \\
    1 & \none & \none \\
    1 & 2 & \none \\
    \none & 2 & 3 
\end{ytableau}
\quad
\begin{ytableau}
    \none & \none & \none  \\
    1 & \none & \none \\
    1 & 2 & \none \\
    1 & 2 & 3 
\end{ytableau}
\label{fig:N=4}
\caption{The pattern of virtual boxes encoding the diagram expansions from Example \ref{ex:tableau4}.}
\end{figure}

What remains to uniquely define the sequence of diagram expansions for any $N$ is to specify, for every step $k$ and every label $m$, a row of a diagram with the top virtual of a column, which we denote by $\text{beg}(m,k)$, and number of virtual boxes in a column, i.e. the length of a telescope, denoted by $\text{len}(m,k)$ (or equivalently - a row with the bottom virtual box $\text{end}(m,k)$). Trying to generalise the patterns for $N=3$ and $N=4$, we propose the following

 \begin{equation}
 \label{eq:beg}
    \text{beg}(m,k)=
    \begin{cases*}
     m & if $1 \leq k \leq N-m$,\\
     m+1 & if $N-m < k \leq N$.
    \end{cases*} \quad\mathrm{,}
  \end{equation}
  
  \begin{equation}
  \label{eq:len}
    \text{len}(m,k)=
    \begin{cases*}
     k & if $1 \leq k \leq N-m$,\\
     N-m & if $N-m < k \leq N$.\\
    \end{cases*} \quad\mathrm{,}
  \end{equation}
so that, since $\text{end}(m,k)=\text{beg}(m,k)+\text{len}(m,k)-1$, we get
  
  \begin{equation}
  \label{eq:end}
    \text{end}(m,k)=
    \begin{cases*}
     m+k-1 & if $1 \leq k \leq N-m$,\\
     N & if $N-m < k \leq N$.
    \end{cases*} \quad\mathrm{.}
  \end{equation}
 
  For example, in case $N=6$ we obtain the patterns presented below.
  
  \begin{figure}[H]
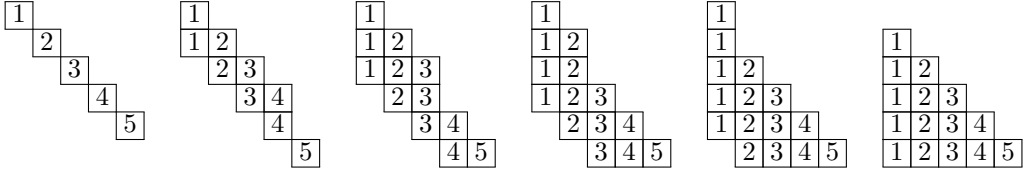

\centering 
\ytableausetup{nosmalltableaux,mathmode,boxsize=1em}
\begin{ytableau}
       1 & \none & \none & \none & \none \\
         \none & 2 & \none & \none & \none  \\
         \none & \none & 3 & \none & \none \\
         \none & \none & \none & 4 & \none  \\
         \none & \none & \none & \none & 5  \\
         \none & \none & \none & \none & \none  
\end{ytableau}
\quad
\begin{ytableau}
       1 & \none & \none & \none & \none \\
       1 & 2 & \none & \none & \none  \\
         \none & 2 & 3 & \none & \none \\
         \none & \none & 3 & 4 & \none  \\
         \none & \none & \none & 4 & \none  \\
         \none & \none & \none & \none & 5  
\end{ytableau}
\quad
\begin{ytableau}
       1 & \none & \none & \none & \none \\
       1 & 2 & \none & \none & \none  \\
       1 & 2 & 3 & \none & \none \\
         \none & 2 & 3 & \none & \none  \\
         \none & \none & 3 & 4 & \none  \\
         \none & \none & \none & 4 & 5  
\end{ytableau}
\quad
\begin{ytableau}
       1 & \none & \none & \none & \none \\
       1 & 2 & \none & \none & \none  \\
       1 & 2 & \none & \none & \none \\
       1 & 2 & 3 & \none & \none  \\
         \none & 2 & 3 & 4 & \none  \\
         \none & \none & 3 & 4 & 5  
\end{ytableau}
\quad
\begin{ytableau}
       1 & \none & \none & \none & \none \\
       1 & \none & \none & \none & \none  \\
       1 & 2 & \none & \none & \none \\
       1 & 2 & 3 & \none & \none  \\
       1 & 2 & 3 & 4 & \none  \\
       \none & 2 & 3 & 4 & 5  
\end{ytableau}
\quad
\begin{ytableau}
       \none & \none & \none & \none & \none \\
       1 & \none & \none & \none & \none  \\
       1 & 2 & \none & \none & \none \\
       1 & 2 & 3 & \none & \none  \\
       1 & 2 & 3 & 4 & \none  \\
       1 & 2 & 3 & 4 & 5  
\end{ytableau}
\label{fig:N=6}
\caption{The proposed patterns of appending virtual boxes corresponding to a sequence of diagram expansions for $N=6$.}
\end{figure}
\noindent
\textbf{Proof of Part II}\\
According to Fact \ref{multiplication-rules} we need to check the following conditions for every step $k \in \{1,2,\ldots,N\}$.
\begin{enumerate}
\item The result of appending boxes in step $k$ gives a Young diagram.
\item Among boxes appended in step $k$, there are no two boxes with the same number in any column.
\item Boxes appended in step $k$ satisfy row counting condition.
\item Boxes appended in step $k$ satisfy column counting condition.
\end{enumerate}
Moreover there is additional Condition 5 - the final diagram need to have a rectangular shape. Since the final diagram consists of $N|\lambda|$ boxes, Condition 5 follows immediately from Condition 1. Conditions 1-4 are verified in the \nameref{app}. This ends the proof.
\end{proof}

\section{Other approaches and final remarks}\label{sec:final}
Remarkably, our problem of multiplying many Young diagrams may be reformulated so that it becomes equivalent to a problem of multiplication of just two respectively larger diagrams \cite{Fu00}. Let us next revisit main points of this construction whose proof can be found in \cite{Fu00}. Our original problem asks when representation $E^\beta$ of $\mathrm{SU}(d)$ appears in the $N$-fold tensor product $E^\lambda\otimes E^\lambda\otimes\hdots\otimes E^\lambda$. For now, we allow $N\geq d$ and, as usual, we want $\beta$ to give the trivial representation, i.e. we require $\beta$  to be the square diagram with $|\lambda|$ columns and $N$ rows.

\begin{lemma}\label{lemma:fulton}
The multiplicity of $E^\beta$ in $\left(E^\lambda\right)^{\otimes N}$ is equal to the Littlewood-Richardson coefficient $c_{\alpha\beta}^\gamma$ where Young diagrams $\alpha$ and $\gamma$ are constructed from $\lambda$ as follows. The skew diagram $\gamma/\alpha$ (the complement of $\alpha$ as a subdiagram of $\gamma$) is constructed by arranging $n$ copies of $\lambda$ so that the top right corner of the $i$th copy touches the bottom left corner of the $(i+1)$th copy.
\end{lemma}

At the first sight, it may appear that using the above lemma would simplify our proof of Theorem \ref{th:main}. However, the technical difficulty of checking the Littlewood-Richardson rules in the multiplication of such large diagrams remains the same. Nevertheless, Lemma \ref{lemma:fulton} will be very useful in our considerations concerning the complete LU symmetry described in Subsection \ref{sub:sud}.

For general (semi)simple compact Lie groups some techniques analogous to Young diagrams have been developed \cite{koike,Klimyk,g2}. However, their application is not as straightforward as it was in the case of Young diagrams. In this section, we describe the most straightforward way of finding states with diagonal $G$-symmetry that relies on solving a large set of linear equations derived from generators of $G$. We test this algorithm for $G=\mathrm{SO}(d)$ in subsection \ref{sub:sod}. The algorithm is as follows. The complexification of the Lie algebra of $G$, $\mkg^\CC \coloneqq \mkg\oplus i\mkg$, has the Cartan decomposition \cite{FuHa91}
\begin{equation}
\mkg^\CC=\mk{h}\oplus\bigoplus_{\alpha\in\Delta_+}\left(\mkg^\CC_{\alpha}\oplus\mkg^\CC_{-\alpha}\right) \mathrm{,}   
\end{equation}
where $\mk{h}$ is the maximal commutative subalgebra, $\Delta_+$ is the set of positive roots of $\mkg^\CC$ and $\mkg^\CC_{\alpha}=\CC e_\alpha$ and $\mkg^\CC_{-\alpha}=\CC e_{-\alpha}$ with $e_{\alpha}$ and $e_{-\alpha}$ called the positive and negative root operators respectively. For instance, when $G=\mathrm{SU}(d)$, the positive root operators are $\kb{i}{j}$ with $i<j$ while the negative root operators are $\kb{i}{j}$ with $i>j$. For $G=\mathrm{SO}(d)$, we specify them in Subsection \ref{sub:sod}. The key fact we are using is that given a unitary representation $\tilde \pi:\ G\to U(D)$, the highest weights in the decomposition of $\pi$ into irreps are annihilated by all positive root operators \cite{FuHa91}. However, we are interested in finding the trivial representation whose highest weight is by definition also the lowest weight, i.e. it is annihilated by all negative root operators. This gives us a set of $2|\Delta_+|$ equations. In fact, commutation relations allow us to consider only the simple roots $\Pi\subset \Delta_+$ which form a basis of $\Delta_+$. The number of simple roots defines the rank of group $G$, $r_G \coloneqq |\Pi|$. For $G=\mathrm{SU}(d)$, the root operators associated with simple roots are $\kb{i}{i+1}$ for $i=1,\hdots,d-1$. This means in particular that in equation \eqref{bosons:eqs} from Subsection \ref{sub:bosons} it would be enough just to take $(\alpha,\beta)\in\{(a,b),(b,a),(b,c),(c,b)\}$ which gives a set of $4$ equations instead of $6$ equations.

More formally, we are looking at representation $\tilde\pi:\ G\to U(d^N)$ defined as $\tilde \pi=\pi\otimes\hdots\otimes\pi$ with $\pi:\ G\to U(d)$ being an irrep of $G$. Algebra $\mkg$ is represented via the derived representation
\begin{equation}
\partial\tilde\pi \coloneqq \partial\pi\otimes\bone\otimes\hdots\otimes\bone+\bone\otimes\partial\pi\otimes\bone\otimes\hdots\otimes\bone+\bone\otimes\hdots\otimes\bone\otimes\partial\pi \mathrm{,}    
\end{equation}
where $\partial\pi:\ \mkg\to \mku(d)$ is the derived representation of $\pi$. State $\kpsi\in \left(\CC^d\right)^{\otimes N}$ spans the trivial representation of $G$ of and only if
\begin{equation}\label{simple-general}
\partial\tilde\pi(e_\alpha)\kpsi=0{\mathrm{\ and\ }}\partial\tilde\pi(e_{-\alpha})\kpsi=0{\mathrm{\ for\ all\ }}\alpha\in \Pi.
\end{equation}
The above set of equations \eqref{simple-general} gives us $2r_G$ linear equations for $\kpsi$ being an element of $d^N$-dimensional Hilbert space. Hence, the complexity of the problem scales polynomially with the rank of $G$ and exponentially with the number of particles.

\subsection[xxx]{Complete diagonal LU-symmetry and diagonal $\mathrm{SU}(2)$-symmetry}\label{sub:sud}
As promised in Subsection \ref{sub:qudits}, here we derive multiplicities of the trivial representation of $\mathrm{SU}(d)$ that appear in Table \ref{tab:natural} using Lemma \ref{lemma:fulton}.
\begin{figure}[H]
\centering
\includegraphics[width=0.4\linewidth]{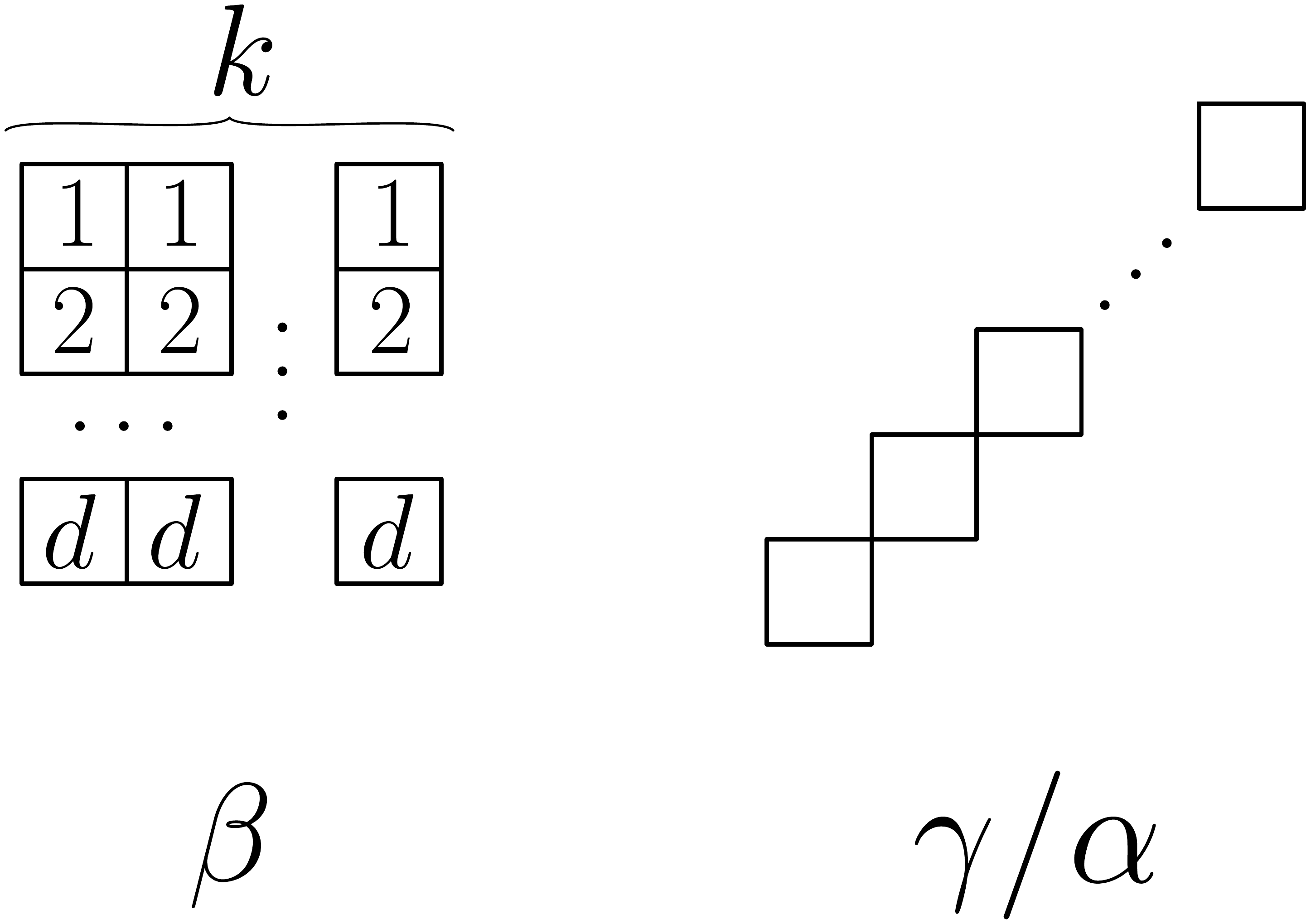}
  \caption{Diagrams $\beta$ and $\gamma/\alpha$ from Lemma \ref{lemma:fulton} for $\lambda=(1)$ and $N=kd$; $E^\lambda$ being the natural representation of $\mathrm{SU}(d)$.}
  \label{fig:fulton-sud}
\end{figure}
Figure \ref{fig:fulton-sud} shows that the multiplicity of the trivial representation $E^\beta$ of $\mathrm{SU}(d)$ in the product $\left(E^\lambda\right)^{\otimes kd}$, $\lambda=(1)$, is the same as the number of ways one can fill diagram $\gamma/\alpha$ with labels from tableau $\beta$ while satisfying Littlewood-Richardson rules. Diagram $\gamma/\alpha$ has a particularly simple form for which the row counting condition and column counting condition are equivalent to the following condition. Arrange labels of the skew tableau $\gamma/\alpha$ in a sequence $(a_1,a_2,\hdots,a_{kd})$ where $a_i$ is the label of the box in the $i$th row with row $1$ being the top row. Then, for any $j\in\{1,\hdots,kd\}$ we require that
\begin{equation}\label{sud-rule}
\#\{i\leq j|\ a_i=n\}\leq\#\{i\leq j|\ a_i=n-1\}\quad{\mathrm{for\ all\ }}n\in\{1,\hdots,d\}.
\end{equation}
Hence, the multiplicity of $E^\beta$ in $\left(E^\lambda\right)^{\otimes kd}$ for $\lambda=(1)$ is the number of sequences $(a_1,a_2,\hdots,a_{kd})$, $1\leq a_i\leq d$ that satisfy conditions \eqref{sud-rule} and
\begin{equation}
\#\{1\leq i\leq kd|\ a_i=n\}=k\quad{\mathrm{for\ all\ }}n\in\{1,\hdots,d\} \mathrm{.}    
\end{equation}
Such sequences are in a bijection with certain combinatorial objects called generalised Dyck paths whose number is given by $d$-dimensional Catalan numbers \cite{OEIS,GP13}
\begin{equation}
c_k^{(d)}=(kd)!\prod_{i=0}^{k-1}\frac{i!}{(d+i)!} \mathrm{.}    
\end{equation}

As a final part of this subsection, we briefly remark that the problem of finding states of diagonal $\mathrm{SU}(2)$-symmetry is the classical problem of spin composition. More specifically, one views a qudit, $\CC^d$, as the configuration space of a particle with spin $J$ where $d=2J+1$, i.e. $\mH_1=\CC^d\cong S^{d-1}\left(\CC^2\right)$. In order to find states of $N$ such qudits with diagonal $\mathrm{SU}(2)$-symmetry, we decompose the $N$-fold product $\mH_1\otimes\hdots\otimes\mH_1$ into irreducible $\mathrm{SU}(2)$-components and ask about the multiplicity of the trivial representation which is the singlet (spin $0$) space. The multiplicities have been calculated in \cite{spin} and they can be expressed by certain combinations of hypergeometric functions. In particular, for $d$-even the singlet space appears with nonzero multiplicity only when $N$ is even.

\subsection[xxx]{States with diagonal $\mathrm{SO}(d)$-symmetry}\label{sub:sod}
Here, we specify the general algorithm \eqref{simple-general} to $G=\mathrm{SO}(d)$. We consider $N$-fold tensor products of the natural unitary representation of $\mathrm{SO}(d)$ on $\CC^d$. Generators are real $d\times d$ real antisymmetric matrices, i.e. $A^T=-A$.  In contrast to $\mk{su}(d)^\CC$, simple root operators (and their negatives) of $\mk{so}(d)^\CC$ have slightly more complicated forms whose derivation can be found in \cite{Knapp}. Simple root operators are different for even and odd $d$. There is a set of $\floor{\frac{d}{2}}$ operators that is common for $d$ both even and odd. It consists of block-diagonal operators containing only one nonzero block of size $4\times 4$. In the basis $\ket{s}$, $s\in\{0,\hdots,d-1\}$, their precise forms read
\begin{gather}
E_{j}^{+}=\kb{2j}{2j+2}+i\kb{2j}{2j+3}-i\kb{2j+1}{2j+2}+\kb{2j+1}{2j+3} \\ -\kb{2j+2}{2j}+i\kb{2j+2}{2j+2}-i\kb{2j+3}{2j}-\kb{2j+3}{2j+1}, \nonumber \\
E_{j}^{-}=\kb{2j}{2j+2}-i\kb{2j}{2j+3}+i\kb{2j+1}{2j+2}+\kb{2j+1}{2j+3} \\ -\kb{2j+2}{2j}-i\kb{2j+2}{2j+2}+i\kb{2j+3}{2j}-\kb{2j+3}{2j+1}\nonumber,
\end{gather}
where $j\in\{0,\hdots,\floor{\frac{d}{2}}-2\}$. For $d$ even, we add to the above set one more simple root operator (and its negative) which reads
\begin{gather}
E_{even}^{+}=\kb{d-4}{d-2}-i\kb{d-4}{d-1}-i\kb{d-3}{d-2}-\kb{d-3}{d-1} \\ -\kb{d-2}{d-4}+i\kb{d-2}{d-3}+i\kb{d-1}{d-4}+\kb{d-1}{d-3}\nonumber, \\
E_{even}^{-}=\kb{d-4}{d-2}+i\kb{d-4}{d-1}+i\kb{d-3}{d-2}-\kb{d-3}{d-1} \\ -\kb{d-2}{d-4}-i\kb{d-2}{d-3}-i\kb{d-1}{d-4}+\kb{d-1}{d-3}\nonumber.
\end{gather}
In the case when $d$ is odd, we add the following operators
\begin{gather}
E_{odd}^{+}=\kb{d-3}{d-1}-i\kb{d-2}{d-1}-\kb{d-1}{d-3}+i\kb{d-1}{d-2}, \\
E_{odd}^{-}=\kb{d-3}{d-1}+i\kb{d-2}{d-1}-\kb{d-1}{d-3}-i\kb{d-1}{d-2}. 
\end{gather}
According to the general algorithm \eqref{simple-general}, in order to find a state $\kpsi\in \left(\CC^d\right)^N$ with the above diagonal $\mathrm{SO}(d)$-symmetry, we need to solve the following set of equations
\begin{equation}\label{sod-1}
\left(E_j^{+/-}\otimes\bone\otimes\hdots\otimes\bone+\hdots+\bone\otimes\hdots\otimes\bone\otimes E_j^{+/-}\right)\kpsi=0
\end{equation}
for $j\in\{0,\hdots,\floor{\frac{d}{2}}-2\}$. On top of that, we require
\begin{equation}\label{sod-2}
\left(E_{even/odd}^{+/-}\otimes\bone\otimes\hdots\otimes\bone+\hdots+\bone\otimes\hdots\otimes\bone\otimes E_{even/odd}^{+/-}\right)\kpsi=0
\end{equation} 
for $d$ even and $d$ odd respectively. The dimension of the solution space (the multiplicity of the trivial representation of $\mathrm{SO}(d)$) for small $N$ and $d$ has been calculated in Table \ref{tab:sod}. For calculations we used a multicore server with 120 GB of RAM. As the size of equations grows exponentially with $N$, the main limiting resource is the size of RAM. As one can see from Table \ref{tab:sod}, the method of directly solving equations \eqref{sod-1} and \eqref{sod-2} is inefficient and allows one to find LME states only for small $N$ and $d$. However, if one is interested in finding only the dimension of the LME space, one can use more efficient methods for computing multiplicities of irreps. Such methods have been implemented in the Lie-ART 2.0 package \cite{lieART} that we used here.

\begin{table}[H]
\centering
\begin{tabular}{c|cccccccccc}
\multirow{2}{*}{$d$} & \multicolumn{10}{c}{$N$} \\
& 1 & 2 & 3 & 4 & 5 & 6 & 7 & 8 & 9 & 10\\
\hline 
\hline
3 & 1 & 1 & 1 & 3 & 6 & 15 & 36 & 91 & {\bf{232}} & {\bf603} \\
\hline
4 & 1 & 1 & 0 & 4 & 0 & 25 & 0 & {\bf{196}} & {\it{0}} & {\bf{1764}} \\
\hline
5 & 1 & 1 & 0 & 3 & 1 & 15 & {\bf{15}} & {\bf{105}} & {\bf{190}} & {\bf{945}} \\
\hline
6 & 1 & 1 & 0 & 3 & 0 & {\bf{16}} & {\it{0}} & {\bf{126}} & {\it{0}} & {\bf{1296}} \\
7 & 1 & 1 & 0 & 3 & 0 & {\bf{15}} & {\bf{1}} & {\bf{105}} & {\bf{28}} & {\bf{945}} \\
\end{tabular}
\caption{Multiplicities of the trivial irrep of $\mathrm{SO}(d)$ in $\left(\CC^d\right)^{\otimes N}$. Numbers in regular font come from solving linear equations \eqref{sod-1} and \eqref{sod-2}. Zeros in italics come from an application of Lemma \ref{lemma:su2}. Numbers in bold were computed using the Lie-ART 2.0 package \cite{lieART}.}
\label{tab:sod}
\end{table}

Lemma \ref{lemma:su2} explains the existence of zeros in Table \ref{tab:sod} above for even $d$.
\begin{lemma}\label{lemma:su2}
Denote by $\pi_d$ the symmetric representation of $\mathrm{SU}(2)$ on $\CC^d\cong S^{(d-1)}\left(\CC^2\right)$ and by $\partial\pi_d:\ \mk{su}(2)\to\mk{su}(d)$ the corresponding derived representation of the $\mk{su}(2)$ algebra. For a given compact semisimple Lie group $G$ and its irrep $\pi_{G,d}:\ G\to \mathrm{SU}(d)$ assume that $\im\,\partial\pi_d\subset \im\,\partial\pi_{G,d}$. If in $\mH=\left(\CC^{d}\right)^{\otimes N}$ there are no LME states with diagonal $\mathrm{SU}(2)$-symmetry given by $\pi_d$, then there are no LME states with diagonal $G$-symmetry. In particular, there are no LME states with diagonal $\mathrm{SU}(d)$- or $\mathrm{SO}(d)$-symmetry given by natural representations of these groups.
\end{lemma}
\begin{proof}
Because  $G$ is semisimple, there is a copy of the trivial representation in the product $\pi_ {G,d}\otimes\hdots\otimes\pi_{G,d}$ iff there exists a vector in $\mH$ that is annihilated by all operators from $\im\,\partial\pi_{G,k}$. Because $\im\,\partial\pi_d\subset \im\,\partial\pi_{G,d}$, any vector annihilated by all operators from $\im\,\partial\pi_{G,d}$ is also annihilated by all operators from $\im\,\partial\pi_k$, i.e. has a diagonal $\mathrm{SU}(2)$-symmetry. The claim for $\mathrm{SU}(d)$ and $\mathrm{SO}(d)$ can be verified by a straightforward calculation in standard bases of $\mk{su}(d)$ and $\mk{so}(d)$.
\end{proof}

\begin{remark}
One can easily find systems containing LME states with some diagonal $G$-symmetry but having no LME states with $\mathrm{SU}(2)$ symmetry. For instance, consider $\mH=\CC^8\otimes\CC^8\otimes\CC^8$ and take the $8$-dimensional representation of $\mathrm{SU}(3)$ given by diagram $\lambda=(2,1)$. By Theorem \ref{th:main}, $\left(E^\lambda\right)^{\otimes 3}$ contains a copy of the trivial representation of $\mathrm{SU}(3)$. However, the third tensor power of the $8$-dimensional representation of $\mathrm{SU}(2)$ described by diagram $\mu=(3)$, $E^\mu=S^7\left(\CC^2\right)$, does not satisfy the necessary condition given in Lemma \ref{lemma:necessary}. Hence, $\left(E^\mu\right)^{\otimes 3}$ does not contain any trivial $\mathrm{SU}(2)$ components.
\end{remark}

To end this subsection, we write down exemplary linearly independent solutions of \eqref{sod-1} and \eqref{sod-2} for $N=4$ and any $d$. We conjecture that this is in fact the full solution space for any $d\neq 4$.
\begin{gather}
\ket{\Psi_1}=\sum_{i,j=1}^d\ket{i,j,j,i},\quad \ket{\Psi_2}=\sum_{i,j=1, i\neq j}^d(\ket{i,j,i,j}-\ket{i,j,j,i}), \\
\ket{\Psi_3}=\sum_{i,j=1, i\neq j}^d(\ket{i,i,j,j}-\ket{i,j,j,i}).
\end{gather}
Furthermore, if $d=N$, we have one state which has in fact the complete diagonal LU-symmetry and is of the form \eqref{state-complete}.

\section{Summary}\label{summary}

 In this paper we proved a theorem in representation theory and showed how it can be used to design critical states with particular large local symmetries, called diagonal $\mathrm{SU}(m)$-symmetries. Our method may be a source of operationally useful quantum states that can be realised in various systems of distinguishable qudits and traps with bosons or fermions occupying a finite number of modes. Moreover, the criterion given in the paper can be used to identify interesting classes of strictly-semistable states - namely the ones which asymptotically have large local symmetries. One should keep in mind that in the systems with LME states, the dimension of a set of LME states (which contains all states with diagonal $H$-symmetries) grows exponentially with $N$, just as the dimension of the Hilbert space of a system. The proof of the main theorem is lengthy and very technical since it involves quite complex combinatorial objects connected with the Littlewood-Richardson rule. However, up to our best knowledge, there is no simple and general criterion allowing one to determine the existence of LME states with diagonal $\mathrm{SU}(m)$-symmetries. In the future, it would be good to consider systems with nonhomogenous local dimensions. We also conjecture that if the manifold of LME states in a multiqudit system $\mH \coloneqq \CC^{d_1}\otimes\hdots\otimes \CC^{d_N}$, up to local unitary equivalence, has dimension at least one, $\mathrm{dim}(\mathrm{LME}/K) \geq 1$, then $\mathcal{H}$ contains strictly semistable states. One can also think about better ways of finding explicit forms of LME states with diagonal $H$-symmetry. For example, in case of compact $H$ it is clear that the set of all such states is the image $\mathrm{Im}\mathcal{P} \subset \mathcal{H}$ of the projector
 \begin{equation}
    \mathcal{P}=\int_H \pi_1(h) \otimes \ldots \otimes \pi_N (h) d\mu(h) \mathrm{,}
 \end{equation}
 where $\mu$ is the normalised Haar measure on $H$.
 %
 Finally, it would be good to think about concrete forms of Kraus operators realising SEP operations for states with diagonal $G$-symmetry.
 
 \section*{Acknowledgments}
 T.M. acknowledges the support of the Foundation for Polish Science (FNP), START programme. O.S. acknowledges the support from the National Science Centre, Poland under the grant SONATA BIS: 2015/18/E/ST1/00200, A.S. acknowledges the support from the Foundation for Polish Science TEAM-NET project (contract no. POIR.04.04.00-00-17C1/18-00).
 
 \newpage
\appendix
\section*{Appendix}
\label{app}
In this appendix we complement Part II of the proof of Theorem \ref{th:main}. by checking Conditions 1-4 from the proof.

\begin{figure}[H]
        \centering
        \begin{tikzpicture}
  \node[rotate=90] at (0,0) {
        \begin{tikzpicture}[node distance=2cm]

       \node (start) [structure] {Start};
        \node (ansatz) [structure,below =1cm of start] {Ansatz for the solution};
        \node (A1) [structure, below left =2.3cm of ansatz] {A1. Definition of $\Delta$ and $\delta(m)$};
        \node (A2) [structure, below of=A1] {A2. Preliminary calculations};
        \node (B1) [structure, below left= .2cm of A2] {B1. Calculations of $\Delta$ \color{red}{(1)}};
        \node (B2) [structure, below right= .2cm of A2] {B2. Calculations of $\delta(m)$};
        \node (B3) [structure, below right=0.2 cm of B1] {B3. Comparison of $\Delta$ and $\delta(m)$ \color{red}{(2)}};
        \node (C2) [structure, below right=1.3cm of B3] {C2. Condition 4 \color{red}{(4)}};
        \node (C1) [structure, right =0.5cm of B2] {C1. Condition 3 \color{red}{(3)}};
        \node (stop) [structure, below = 1 cm of C2] {Stop};

        \draw [arrow] (start) -- (ansatz);
        \draw [arrow] (ansatz) -- (A1);
        \draw [arrow] (ansatz) -- (C1);
        \draw [arrow] (A1) -- (A2);
        \draw [arrow] (A2) -- (B1);
        \draw [arrow] (A2) -- (B2);
        \draw [arrow] (B1) -- (B3);
        \draw [arrow] (B2) -- (B3);
        \draw [arrow] (B3) -- (C2);
        \draw [arrow] (C1) -- (C2);
        \draw [arrow] (C2) -- (stop);

        \node[draw, thick, dotted, rounded corners, inner xsep=1em, inner ysep=1em, fit=(ansatz)] (box) {};
        \node[fill=white] at (box.south) {Part I: Ansatz for the solution};
        \node[draw, thick, dotted, rounded corners, inner xsep=1em, inner ysep=1em, fit=(A1) (A2) (B1) (B3) (C1) (C2)] (box) {};
        \node[fill=white] at (box.south) {Part II: Verification of correctness};
        \end{tikzpicture}
        };
        \end{tikzpicture}
        \caption{A roadmap of the proof of Theorem \ref{th:main}. The arrows show logical dependencies. Red numbers in brackets indicate which condition is concluded in a given subsection.} \label{fig:roadmap}
    \end{figure}
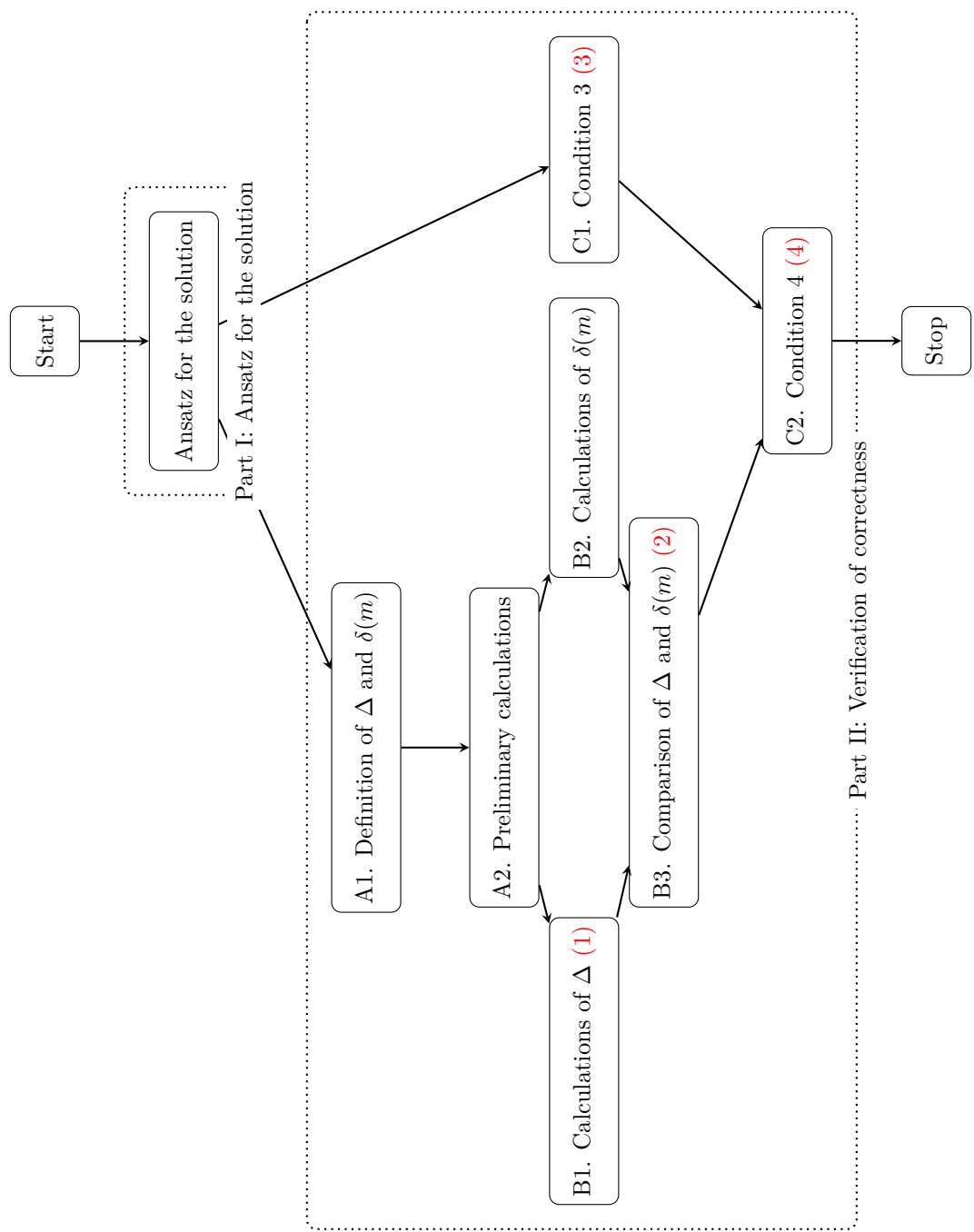

\section{Auxiliary objects for Conditions 1 and 2}
In this Section we introduce auxiliary objects - $\Delta$ and $\delta(m)$ and perform some preliminary calculations needed to calculate them in Section \ref{sec:app_b}. This enables us to check Conditions 1 and 2 since they can be reduced to inequalities (\ref{ineq1} and \ref{ineq2}) involving $\Delta$ and $\delta(m)$.
\subsection[xxx]{Definition of $\Delta$ and $\delta(m)$}

From the ansatz given in Part I of the proof of Theorem \ref{th:main}, it is easy to see that the number of boxes with label $n$ appended at step $k$ to row $i$ is given by
  \begin{equation}
  \label{eq:n}
    (\#n)_i^k=
    \begin{cases*}
      T_{n,k}^{n+k-i} & if $1 \leq k \leq N-n $,\\
      T_{n,N-n}^{N-i+1} & if $N-n < k \leq N $,\\
      0       & otherwise .
    \end{cases*} \quad\mathrm{.}
  \end{equation}
   Let us define

\begin{equation}
\label{eq:Del}
   \Delta_i^k =\sum_{n=1}^{i+1} \left( \sum_{s=1}^{k-1} (\#n)_{i}^s - \sum_{s=1}^{k-1} (\#n)_{i+1}^s \right) \text{,}
\end{equation}
and
\begin{equation}
\label{eq:del}
   \delta_i^k (m)=\sum_{n=1}^m (\#n)_{i+1}^k - \sum_{n=1}^{m-1} (\#n)_{i}^k \mathrm{.}
\end{equation}

Numbers $\Delta_i^k$ for a fixed $k$ describe the shape of the tableau obtained after all steps up to $k-1$ -- these are the lengths of steps of black "stairs" in Fig.~\ref{fig:deltas}. Precisely, $\Delta_i^k$ is a difference between the total number of boxes (length) of $i$th and $(i+1)$th row after $k-1$ steps. In order to find the interpretation of $\delta_i^k(m)$, one can left-align all the boxes appended in step $k$. Then, $\delta_i^k(m)$ is the distance (number of boxes) between the rightmost box with label $m-1$ in $i$th row and the rightmost box with label $m$ in $i+1$th row. One can think of a matrix $\Delta$ with entries $\Delta_i^k$ and matrix $\delta(m)$ with entries $\delta_i^k(m)$.

\begin{figure}[H]
\centering
\includegraphics[width=0.6\linewidth]{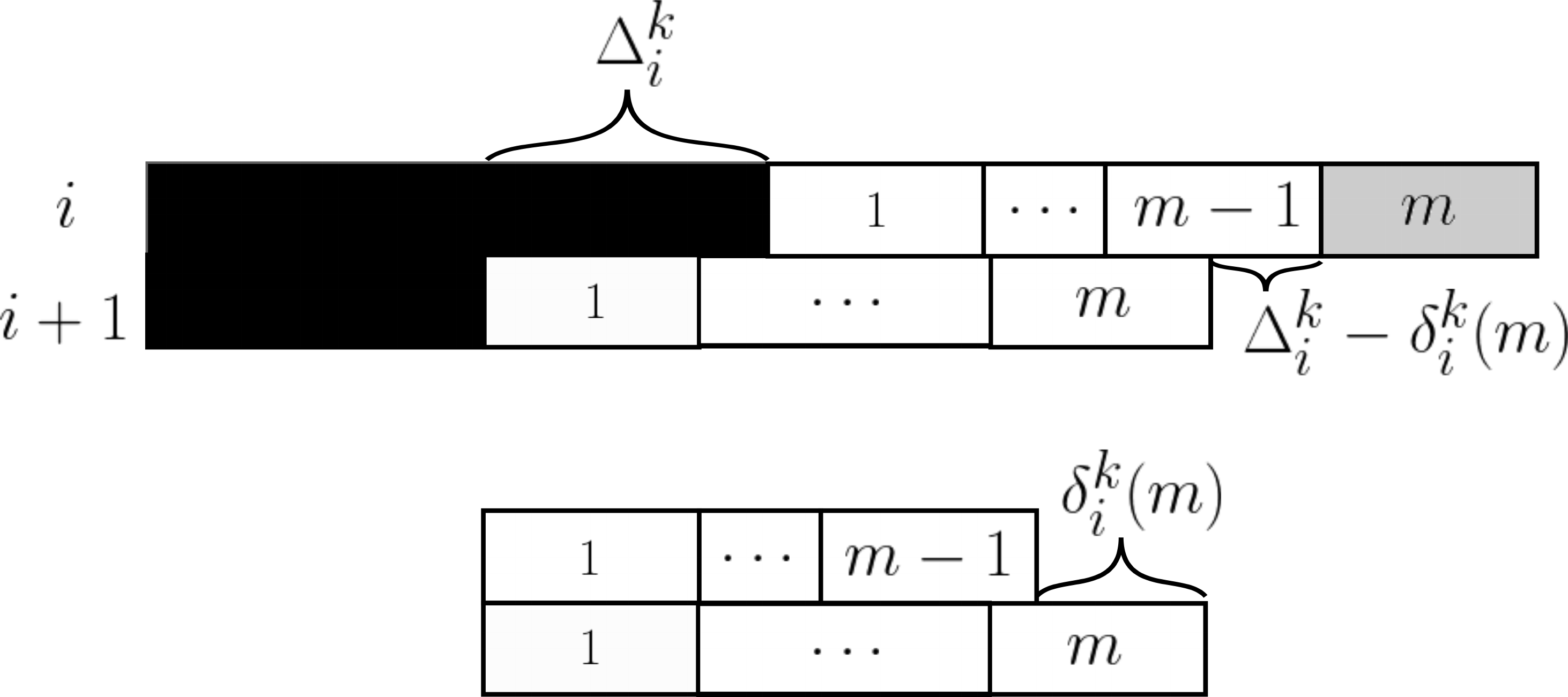}
  \caption{The meaning of  $\Delta_i^k$ and  $\delta_i^k (m)$. Rectangles represent groups of boxes with the same label. The black block represents the shape of the tableau obtained after all previous steps (step $k-1$ being the last one). The grey block is a block for which $\Delta_i^k \geq \delta_i^k (m)$ tests Condition 2.}
  \label{fig:deltas}
\end{figure}
\noindent Let us next reformulate conditions given in Fact \ref{multiplication-rules} using $\Delta_i^k$ and $\delta_i^k (m)$. Condition 1 is equivalent with
\begin{equation}
\label{ineq1}
\Delta_i^k  \geq 0\mathrm{,}
\end{equation}
and Condition 2 (see also Fig. \ref{fig:deltas}) is implied by
\begin{equation}
\label{ineq2}
   \Delta_i^k  \geq \delta_{i}^k(m)\mathrm{,}
\end{equation}
for every  $$i \in \{2,\ldots,N-1\}, \quad k \in \{1,2,\ldots,N\}, m \in \{1,2,\ldots,N-1\} \mathrm{.}$$
We introduce the symbol:
\begin{equation}
\lambda_{p,q} \coloneqq \lambda_p-\lambda_q \textrm{.}
\end{equation}
For example, in case $N=6$ we obtain (see also Fig. \ref{fig:nlayers}):

\begin{table}[H]
\centering
\begin{tabular}{c|cccccc}
\multirow{2}{*}{$i$} & \multicolumn{5}{c}{$k$} \\
& 2 & 3 & 4 & 5 & 6  \\
\hline 
\hline
1 & $\lambda_{1,2}$ & $\lambda_{2,3}$ & $\lambda_{3,4}$ & $\lambda_{4,5}$ & $\lambda_{5}$ \\
\hline
2 & $\lambda_{2,3}$ & $\lambda_{1,2}+\lambda_{3,4}$ & $\lambda_{2,3}+\lambda_{4,5}$ & $\lambda_{3,4} + \lambda_{5}$ & $\lambda_{4,5}$ \\
\hline
3 & $\lambda_{3,4}$ & $\lambda_{2,3}+\lambda_{4,5}$ & $\lambda_{1,2}+\lambda_{3,4} + \lambda_{5}$ & $\lambda_{2,3} + \lambda_{4,5}$ & $\lambda_{3,4}$ \\
\hline
4 & $\lambda_{4,5}$ & $\lambda_{3,4}+\lambda_{5}$ & $\lambda_{2,3}+\lambda_{4,5}$ & $\lambda_{1,2} + \lambda_{3,4}$ & $\lambda_{2,3}$ \\
\hline
5 & $\lambda_{5}$ & $\lambda_{4,5}$ & $\lambda_{3,4}$ & $\lambda_{2,3}$ & $\lambda_{1,2}$ \\
\hline
\end{tabular}
\caption{Entries of $\Delta^{k}_{i}$ for $N=6$.}
\end{table}

\begin{table}[H]
\centering
\begin{tabular}{c|ccccccc}
\multirow{2}{*}{$i$} & \multicolumn{5}{c}{$k$} \\
& 1 & 2 & 3 & 4 & 5 & 6  \\
\hline 
\hline
1 & - & $\lambda_{1,2}$ & $\lambda_{2,3}$ & $\lambda_{3,4}$ & $\lambda_{4,5}$ & $\lambda_5$ \\
\hline
2 & - &- & $\lambda_{1,2}$ & $\lambda_{2,3}$ & $\lambda_{3,4}$ & $\lambda_{4,5}$ \\
\hline
3 &- & - & - & $\lambda_{1,2}$ & $\lambda_{2,3}$ & $\lambda_{3,4}$ \\
\hline
4 &- & - & - & - & $\lambda_{1,2}$ & $\lambda_{2,3}$ \\
\hline
5 & - &- & - & - & - & $\lambda_{1,2}$ \\
\hline
\end{tabular}
\caption{Entries of $\delta^{k}_{i}(1)$ for $N=6$.}
\end{table}

\begin{table}[H]
\centering
\begin{tabular}{c|ccccccc}
\multirow{2}{*}{$i$} & \multicolumn{5}{c}{$k$} \\
& 1 & 2 & 3 & 4 & 5 & 6  \\
\hline 
\hline
1 & $-\lambda_{1,2}$ & $\lambda_{1,2}-\lambda_{2,3}$ & $\lambda_{2,3}-\lambda_{3,4}$ & $\lambda_{3,4}-\lambda_{4,5}$ & $\lambda_{4,5}-\lambda_5$ & $\lambda_5$ \\
\hline
2 & - & $\lambda_{2,3}-\lambda_{1,2}$ & $\lambda_{1,2}+\lambda_{3,4}-\lambda_{2,3}$ & $\lambda_{2,3}+\lambda_{4,5}-\lambda_{3,4}$ & $\lambda_{3,4}+\lambda_5-\lambda_{4,5}$ & $\lambda_{4,5}$ \\
\hline
3 & - &- & $\lambda_{2,3}-\lambda_{1,2}$ & $\lambda_{1,2}+\lambda_{3,4}-\lambda_{2,3}$ & $\lambda_{2,3}+\lambda_{4,5}-\lambda_{3,4}$ & $\lambda_{3,4}$ \\
\hline
4 & - &- & - & $\lambda_{2,3}-\lambda_{1,2}$ & $\lambda_{1,2}+\lambda_{3,4}-\lambda_{2,3}$ & $\lambda_{2,3}$ \\
\hline
5 & - &- & - & - & $\lambda_{2,3}-\lambda_{1,2}$ & $\lambda_{1,2}$ \\
\hline
\end{tabular}
\caption{Entries of $\delta^{k}_{i}(2)$ for $N=6$.}
\end{table}

\begin{table}[H]
\centering
\resizebox{\textwidth}{!}{\begin{tabular}{c|ccccccc}
\multirow{2}{*}{$i$} & \multicolumn{5}{c}{$k$} \\
& 1 & 2 & 3 & 4 & 5 & 6  \\
\hline 
\hline
1 & $-\lambda_{1,2}$ & $\lambda_{1,2}-\lambda_{2,3}$ & $\lambda_{2,3}-\lambda_{3,4}$ & $\lambda_{3,4}-\lambda_{4,5}$ & $\lambda_{4,5}-\lambda_5$ & $\lambda_5$ \\
\hline
2 & $-\lambda_{2,3}$ & $\lambda_{2,3}-(\lambda_{1,2}+\lambda_{3,4})$ & $\lambda_{1,2}+\lambda_{3,4}-(\lambda_{2,3}+\lambda_{4,5})$ & $\lambda_{2,3}+\lambda_{4,5}-(\lambda_{3,4}+\lambda_5)$ & $\lambda_{3,4}+\lambda_{5}-\lambda_{4,5}$ & $\lambda_{4,5}$ \\
\hline
3 & - & $\lambda_{3,4}-\lambda_{2,3}$ & $\lambda_{2,3}+\lambda_{4,5}-(\lambda_{1,2}+\lambda_{3,4})$ & $\lambda_{1,2}+\lambda_{3,4}+\lambda_5-(\lambda_{2,3}+\lambda_{4,5})$ & $\lambda_{2,3}+\lambda_{4,5}-\lambda_{3,4}$ & $\lambda_{3,4}$ \\
\hline
4 & - &- & $\lambda_{3,4}-\lambda_{2,3}$ & $\lambda_{2,3}+\lambda_{4,5}-(\lambda_{1,2}+\lambda_{3,4})$ & $\lambda_{1,2}+\lambda_{3,4}-\lambda_{2,3}$ & $\lambda_{2,3}$ \\
\hline
5 & - &- & - & $\lambda_{3,4}-\lambda_{2,3}$ & $\lambda_{2,3}-\lambda_{1,2}$ & $\lambda_{1,2}$ \\
\hline
\end{tabular}}
\caption{Entries of $\delta^{k}_{i}(3)$ for $N=6$.}
\end{table}

\begin{table}[H]
\centering
\resizebox{\textwidth}{!}{\begin{tabular}{c|ccccccc}
\multirow{2}{*}{$i$} & \multicolumn{5}{c}{$k$} \\
& 1 & 2 & 3 & 4 & 5 & 6  \\
\hline 
\hline
1 & $-\lambda_{1,2}$ & $\lambda_{1,2}-\lambda_{2,3}$ & $\lambda_{2,3}-\lambda_{3,4}$ & $\lambda_{3,4}-\lambda_{4,5}$ & $\lambda_{4,5}-\lambda_5$ & $\lambda_5$ \\
\hline
2 & $-\lambda_{2,3}$ & $\lambda_{2,3}-(\lambda_{1,2}+\lambda_{3,4})$ & $\lambda_{1,2}+\lambda_{3,4}-(\lambda_{2,3}+\lambda_{4,5})$ & $\lambda_{2,3}+\lambda_{4,5}-(\lambda_{3,4}+\lambda_5)$ & $\lambda_{3,4}+\lambda_{5}-\lambda_{4,5}$ & $\lambda_{4,5}$ \\
\hline
3 & $-\lambda_{3,4}$ & $\lambda_{3,4}-(\lambda_{2,3}+\lambda_{4,5})$ & $\lambda_{2,3}+\lambda_{4,5}-(\lambda_{1,2}+\lambda_{3,4}+\lambda_5)$ & $\lambda_{1,2}+\lambda_{3,4}+\lambda_5-(\lambda_{2,3}+\lambda_{4,5})$ & $\lambda_{2,3}+\lambda_{4,5}-\lambda_{3,4}$ & $\lambda_{3,4}$ \\
\hline
4 & - &$\lambda_{4,5}-\lambda_{3,4}$ & $\lambda_{3,4}+\lambda_5-(\lambda_{2,3}+\lambda_{4,5})$ & $\lambda_{2,3}+\lambda_{4,5}-(\lambda_{1,2}+\lambda_{3,4})$ & $\lambda_{1,2}+\lambda_{3,4}-\lambda_{2,3}$ & $\lambda_{2,3}$ \\
\hline
5 & - & - & $\lambda_{4,5}-\lambda_{3,4}$ & $\lambda_{3,4}-\lambda_{2,3}$ & $\lambda_{2,3}-\lambda_{1,2}$ & $\lambda_{1,2}$ \\
\hline
\end{tabular}}
\caption{Entries of $\delta^{k}_{i}(4)$ for $N=6$.}
\end{table}

\begin{table}[H]
\centering
\resizebox{\textwidth}{!}{\begin{tabular}{c|ccccccc}
\multirow{2}{*}{$i$} & \multicolumn{5}{c}{$k$} \\
& 1 & 2 & 3 & 4 & 5 & 6  \\
\hline 
\hline
1 & $-\lambda_{1,2}$ & $\lambda_{1,2}-\lambda_{2,3}$ & $\lambda_{2,3}-\lambda_{3,4}$ & $\lambda_{3,4}-\lambda_{4,5}$ & $\lambda_{4,5}-\lambda_5$ & $\lambda_5$ \\
\hline
2 & $-\lambda_{2,3}$ & $\lambda_{2,3}-(\lambda_{1,2}+\lambda_{3,4})$ & $\lambda_{1,2}+\lambda_{3,4}-(\lambda_{2,3}+\lambda_{4,5})$ & $\lambda_{2,3}+\lambda_{4,5}-(\lambda_{3,4}+\lambda_5)$ & $\lambda_{3,4}+\lambda_{5}-\lambda_{4,5}$ & $\lambda_{4,5}$ \\
\hline
3 & $-\lambda_{3,4}$ & $\lambda_{3,4}-(\lambda_{2,3}+\lambda_{4,5})$ & $\lambda_{2,3}+\lambda_{4,5}-(\lambda_{1,2}+\lambda_{3,4}+\lambda_5)$ & $\lambda_{1,2}+\lambda_{3,4}+\lambda_5-(\lambda_{2,3}+\lambda_{4,5})$ & $\lambda_{2,3}+\lambda_{4,5}-\lambda_{3,4}$ & $\lambda_{3,4}$ \\
\hline
4 & $-\lambda_{4,5}$ & $\lambda_{4,5}-(\lambda_{3,4}+\lambda_5)$ & $\lambda_{3,4}+\lambda_5-(\lambda_{2,3}+\lambda_{4,5})$ & $\lambda_{2,3}+\lambda_{4,5}-(\lambda_{1,2}+\lambda_{3,4})$ & $\lambda_{1,2}+\lambda_{3,4}-\lambda_{2,3}$ & $\lambda_{2,3}$ \\
\hline
5 & $-$ & $\lambda_5-\lambda_{4,5}$ & $\lambda_{4,5}-\lambda_{3,4}$ & $\lambda_{3,4}-\lambda_{2,3}$ & $\lambda_{2,3}-\lambda_{1,2}$ & $\lambda_{1,2}$ \\
\hline
\end{tabular}}
\caption{Entries of $\delta^{k}_{i}(5)$ for $N=6$.}
\end{table}

\subsection[xxx]{Preliminary calculations}

The possible values of $n$, $k$ and $i$ in (\ref{eq:n}) form a 3-dimensional lattice we denote by $\Lambda=\{1,2, \ldots,N\}^{\times 3}$. By an $x$-layer we mean a subset of $\Lambda$ with fixed $x$, where $x$ is $n$, $k$ or $i$. By $\#$ we denote the function given by $\#(n,k,i)=(\#n)^k_i$. In order to derive explicit formula for $\#$ we combine (\ref{eq:tel}) and  (\ref{eq:n}).  We get four possible nontrivial formulas, denoted $A1,\,A2,\,B1,\,B2$, each one valid on a different domain of $\Lambda$, which define $\#$ piecewise. The domains are given by
\begin{align}
\mathrm{Dom}(A1)&=\{(n,k,i) \in \Lambda \; |  1+i-n\leq k \leq N-n,\; n < i \} \nonumber \mathrm{,}\\
\mathrm{Dom}(A2)&=\{(n,k,i) \in \Lambda \; | \; 1+i-n\leq k \leq N-n,\; n = i  \} \mathrm{,}\\ 
\mathrm{Dom}(B1)&=\{(n,k,i) \in \Lambda \; | \; N-n < k ,\; n +1 < i   \}  \nonumber \mathrm{,} \\
\mathrm{Dom}(B2)&=\{(n,k,i) \in \Lambda \; | \; N-n < k ,\;  n +1 = i   \}  \nonumber \mathrm{.}
\end{align}
Let us denote
\begin{equation}
 S=\mathrm{Dom}(A1) \cupdot \mathrm{Dom}(A2) \cupdot \mathrm{Dom}(B1) \cupdot \mathrm{Dom}(B2)\mathrm{,}
\end{equation}
and define $\#: \Lambda \rightarrow \mathbb{N}$ by
\begin{align}
\label{eq:formulas}
(\#n)_i^k=
\begin{cases}
 A1(n,k,i) = \lambda_{2n+k-i-1}-\lambda_{2n+k-i},\; &\mathrm{if} \; (n,k,i) \in \mathrm{Dom}(A1) \mathrm{,}\\
A2(n,k,i)=\lambda_{n+k-1},\; &\mathrm{if} \; (n,k,i) \in \mathrm{Dom}(A2) \mathrm{,} \\
 B1(n,k,i)=\lambda_{n+N-i} - \lambda_{n+N-i+1},\; &\mathrm{if} \; (n,k,i) \in \mathrm{Dom}(B1) \mathrm{,} \\
 B2(n,k,i)=\lambda_{N-1},\; &\mathrm{if} \; (n,k,i) \in \mathrm{Dom}(B2) \mathrm{,} \\
 0 ,\; &\mathrm{if} \; (n,k,i) \not\in S \mathrm{.}
\end{cases}
\end{align}

In Section \ref{sec:app_b} we use the following properties of (\ref{eq:formulas})
\begin{align}
&A1(n,k,i)=A1(n,k \pm 1,i \pm 1)=A1(n \pm 1,k \mp 2,i)=A1(n \pm 1,k, i \pm 2) \mathrm{,}\\
&A2(n,k,i)=A2(n \pm 1,k \mp 1,i)=A2(n,k,i \pm 1) \mathrm{,}\\
&B1(n,k,i)=B1(n \pm 1,k, i \pm 1)=B1(n,k\pm 1,i) \mathrm{,}\\
&B2(n,k,i)=\mathrm{const}(N) \mathrm{,}
\end{align}
whenever right-hand side is defined. 
For example, in case $N=6$ we get the following figures depicting the partition of $S$ into domains.

\begin{figure}[H]
\centering
\includegraphics[width=.8\linewidth]{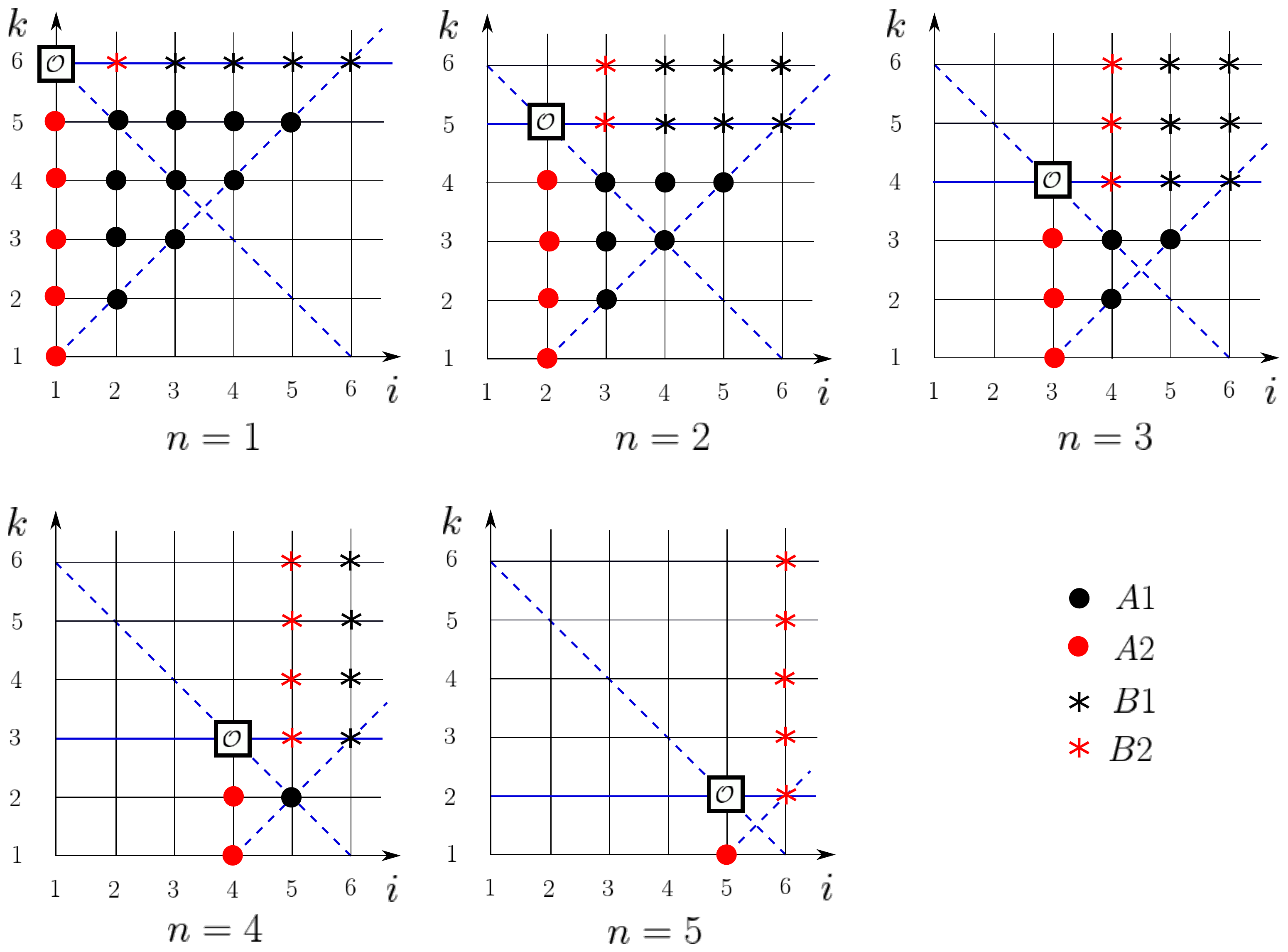}
  \caption{Decomposition of $S$ into domains for $N=6$ in every $n$-layer. The square with symbol $\mathcal{O}$ denotes the moving distinguished point (one may think of a slider) which will be used later to manage calculations.}
\label{fig:nlayers}
\end{figure}

\begin{figure}[H]
\centering
\includegraphics[width=0.9\linewidth]{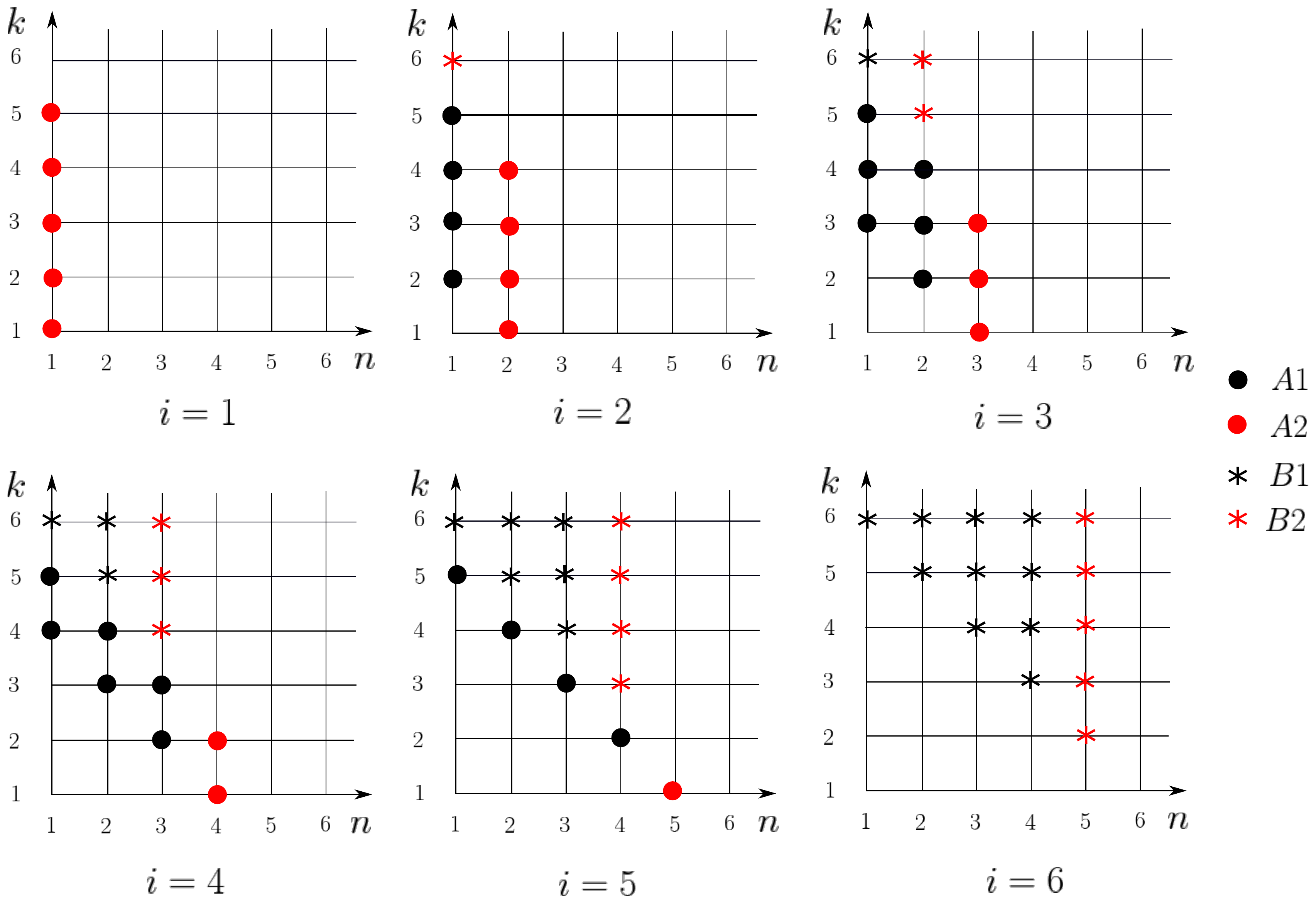}
  \caption{Decomposition of $S$ into domains for $N=6$ in every $i$-layer.}
\label{fig:ilayers}
\end{figure}

\begin{figure}[H]
\centering
\includegraphics[width=0.9\linewidth]{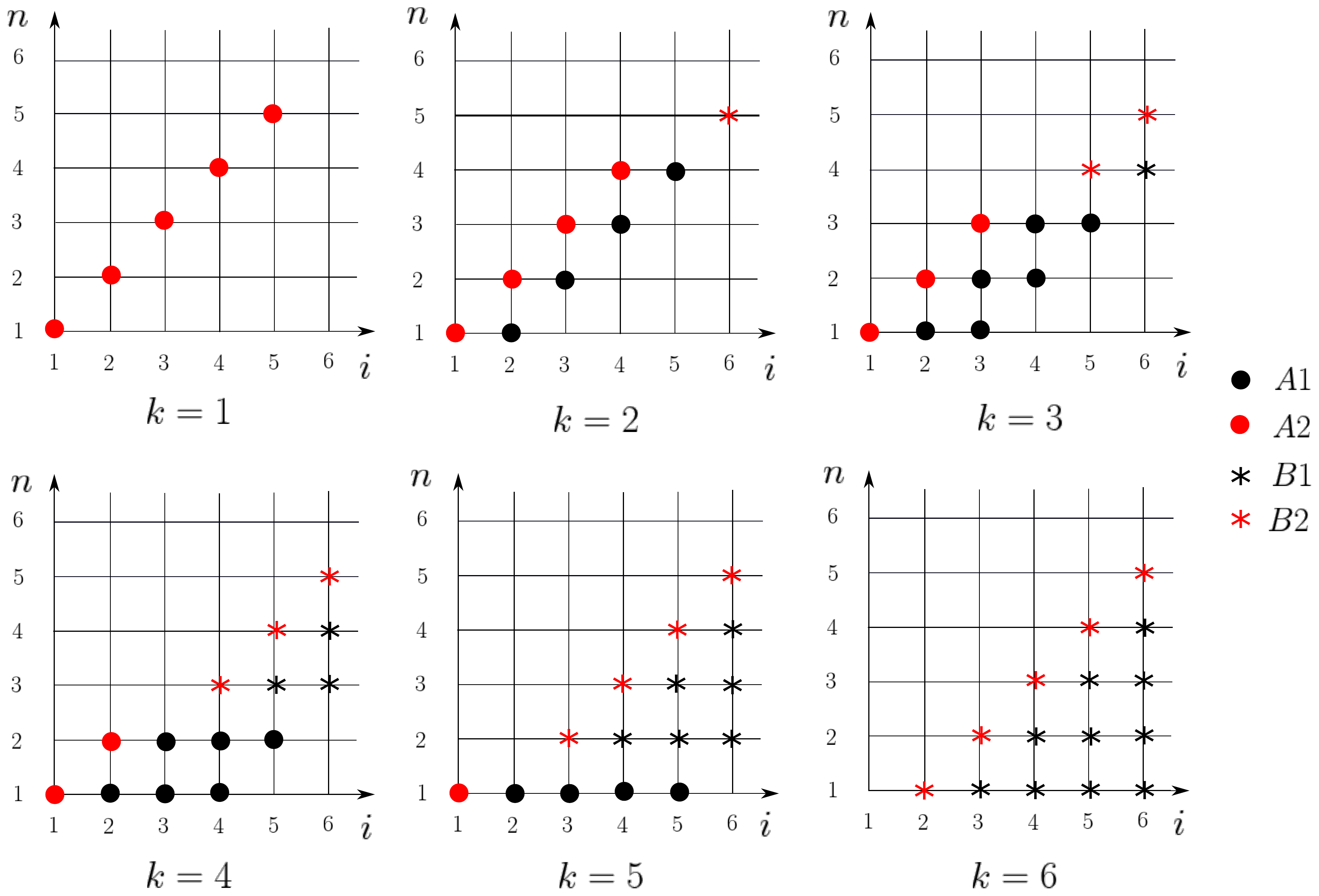}
  \caption{Decomposition of $S$ into domains for $N=6$ in every $k$-layer.}
\label{fig:klayers} 
\end{figure}
\noindent
\textbf{Preliminary calculations of $\Delta$}\\
From now on we will stick to decomposition of $S$ into $n$-layers. To make calculations more traceable, we introduce the moving distinguished point (slider) denoted by $\mathcal{O}$. Its $ik$-coordinates (for fixed $n$-layer) are $\mathcal{O}=(n,\,N-n+1)$. Fix a column of interest $j$. We introduce the parameter $\ell=j-n$, i.e. the distance from $\mathcal{O}$. We have 5 cases.
\begin{align}
\label{eq:casescomp}
\circled{0},\; &\mathrm{if} \; \ell < 0 \nonumber  \mathrm{,} \\
\circled{1},\; &\mathrm{if} \; \ell = 0 \nonumber \mathrm{,}\\
\circled{2},\; &\mathrm{if} \; \ell = 1 \mathrm{,}\\
\circled{3},\; &\mathrm{if} \; N-n > \ell > 1 \nonumber \mathrm{,} \\
\circled{4},\; &\mathrm{if} \;\ell = N - n \nonumber \mathrm{.} \\\nonumber 
\end{align}
Let us define the symbol
\begin{equation}
\sum(s,j) \coloneqq \sum_{k=1}^{s-1} (\#n)_j^k \mathrm{,}
\end{equation}
for the sum of elements from $j$th column which will contribute to $\Delta^s_j$ (this corresponds to the non-faded points in $j$th column on Fig. \ref{fig:decompS} (b)). The value of $n$ should be clear from the context.
Notice that
\begin{equation}
\Delta^k_i = \sum_{n=1}^{i+1} \left( \sum(s,i) - \sum(s,i+1) \right) \mathrm{.}
\end{equation}
We also introduce the symbol

\begin{equation}
\circled{X}_i^n \coloneqq X(n,k,i) \mathrm{,}
\end{equation}
where $X \in \{\mathrm{A1},\mathrm{A2},\mathrm{B1},\mathrm{B2}\}$ is one of (\ref{eq:formulas})
and $k$ is the running parameter (the symbols will be used under summation over $k$).

In order to avoid splitting equations into cases, with every condition $\mathcal{C}$ on $n$, $k$ and $i$ we associate the subset $S(\mathcal{C}) \subset \Lambda$ of all triples $(n,k,i)$ that satisfy $\mathcal{C}$. We define the characteristic function $\chi(\mathcal{C}): \Lambda \rightarrow \{0,1\}$ as the characteristic function of a set $S(\mathcal{C})$, i.e.
\begin{equation}
\chi_\mathcal{C}(n,k,i)=
\begin{cases} 1,\; \mathrm{if} \; (n,k,i) \in S(\mathcal{C})\mathrm{,}\\
0, \; \mathrm{otherwise} \mathrm{.}
\end{cases} 
\end{equation}
Hence, considering (\ref{eq:casescomp}), we can write
\begin{align}
\label{eq:casescompchiF}
&\circled{0}\;\sum(s,j)=0 \mathrm{,}\\
&\circled{1}\;\sum(s,j)=\sum_{k=1}^{M(n)} \circled{A2}_{i=j}^n \mathrm{,}\\
&\circled{2}\;\sum(s,j)= \sum_{k=\ell+1}^{M(n)} \circled{A1}_{i=j}^n \chi_{s \geq j-n+2}\;\chi_{j < N} +
\sum_{k=N-n+1}^{s-1} \circled{B2}_{i=j}^n \chi_{s \geq N-n+2} \mathrm{,}\\
&\circled{3}\;\sum(s,j)= \sum_{k=\ell+1}^{M(n)} \circled{A1}_{i=j}^n \chi_{s \geq j-n+2}\;\chi_{j < N}+
\sum_{k=N-n+1}^{s-1} \circled{B1}_{i=j}^n \chi_{s \geq N-n+2} \mathrm{,}\\
\label{eq:casescompchiL}
&\circled{4}\;\sum(s,j)= \sum_{k=N-n+1}^{s-1} \circled{B1}_{i=j}^n \chi_{s \geq N-n+2} +
\sum_{k=N-n+1}^{s-1} \circled{B2}_{i=j}^n \chi_{n=N-1} \; \chi_{s \geq 3} \mathrm{.}
\end{align}
By $\circled{a}^n_i$, where $a \in \{1,2,3,4\}$, we denote the formula for $\sum (s,j)$ in case $\circled{a}$ given by (\ref{eq:casescompchiF}) up to (\ref{eq:casescompchiL}). The value of $s$ should be clear from the context.

We say that the $n$-layer is \textit{non-degenerate} if it realises all 5 cases (\ref{eq:casescomp}) in a non-degenerate way. Otherwise the layer is called \textit{degenerate}. Non-degeneracy is equivalent with $n<N-2$ and $N \geq 4$. From now on we assume that $N \geq 5$ to avoid some degenerate cases; For the proof in cases $N=2,3,4$ we refer to direct constructions given in Examples \ref{ex:tableau2}, \ref{ex:tableau3} and \ref{ex:tableau4}, which can be easily checked to be valid.

\begin{figure}[H]
    \centering
    \subfloat[]{{\includegraphics[width=0.47\textwidth]{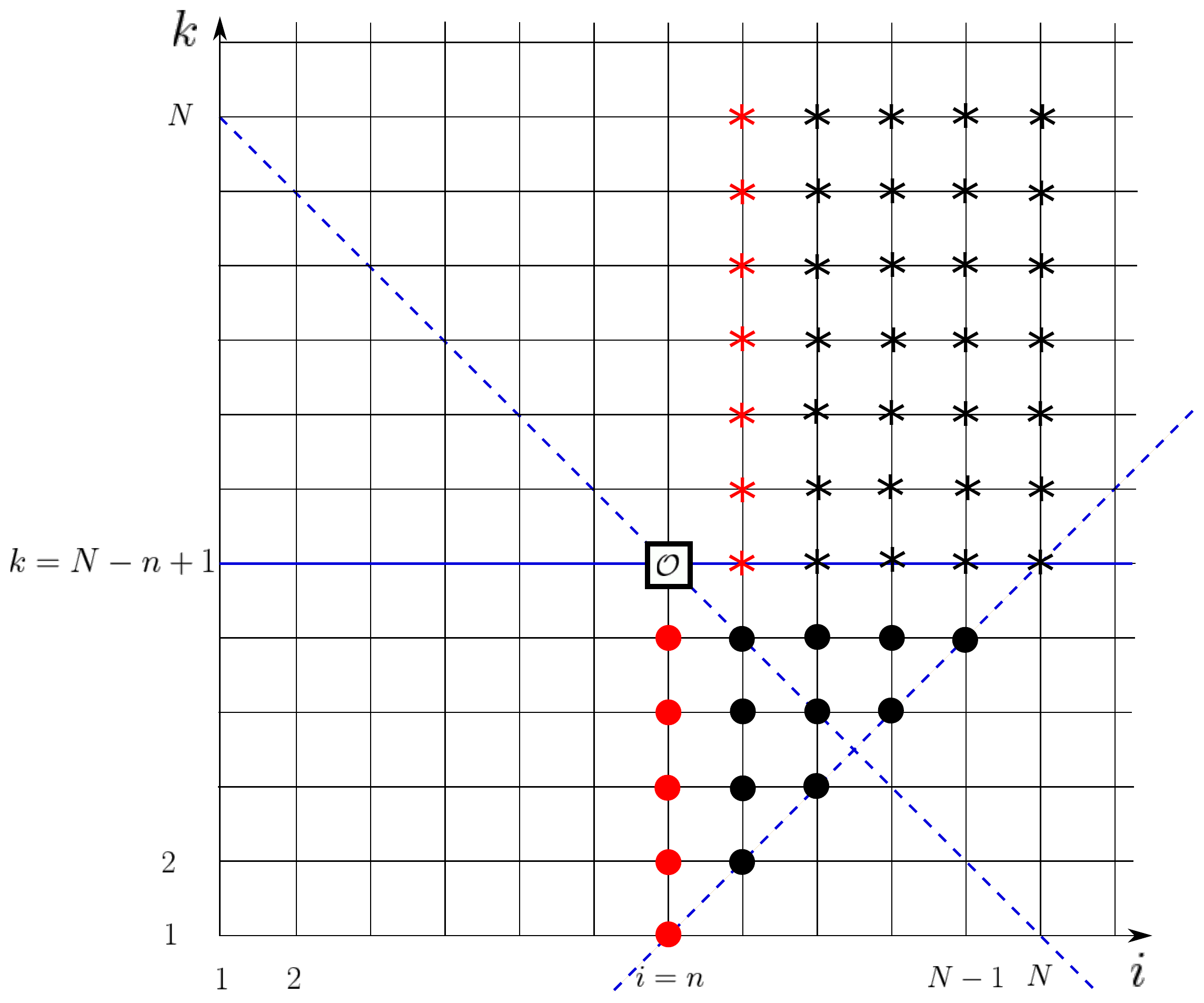} }}%
    \subfloat[]{{\includegraphics[width=0.53\textwidth]{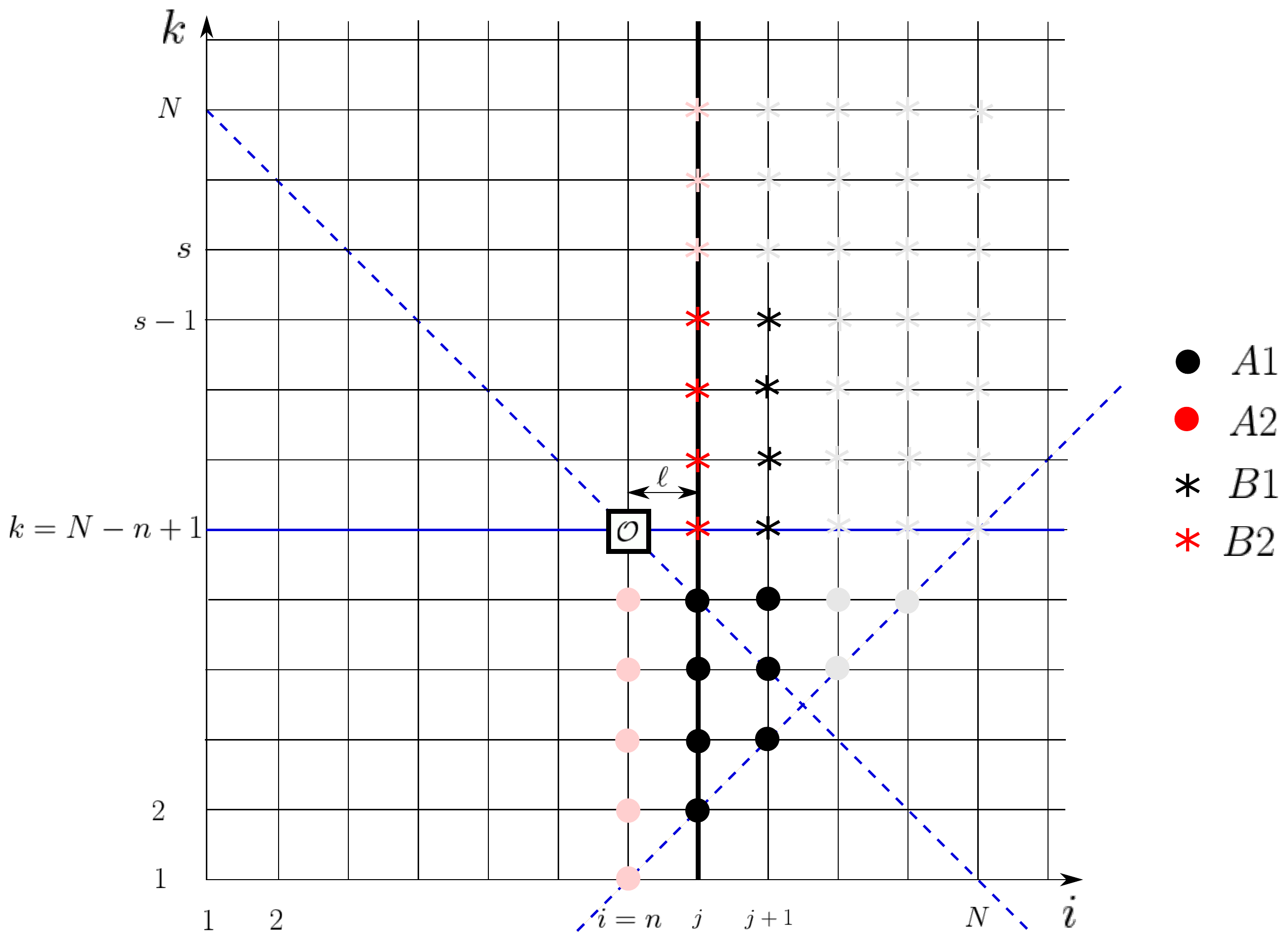} }}%
    \caption{(a) - decomposition of $S$ into domains for generic $n$; (b) - the contribution (non-faded points) to $\Delta_j^s$ from the $n$th layer in case $\ell=1$. Positive contribution is from column $i=j$ (thick) and negative from $i=j+1$.}
    \label{fig:decompS}%
\end{figure}

\begin{figure}[H]
    \centering
    \subfloat[]{{\includegraphics[width=0.43\textwidth]{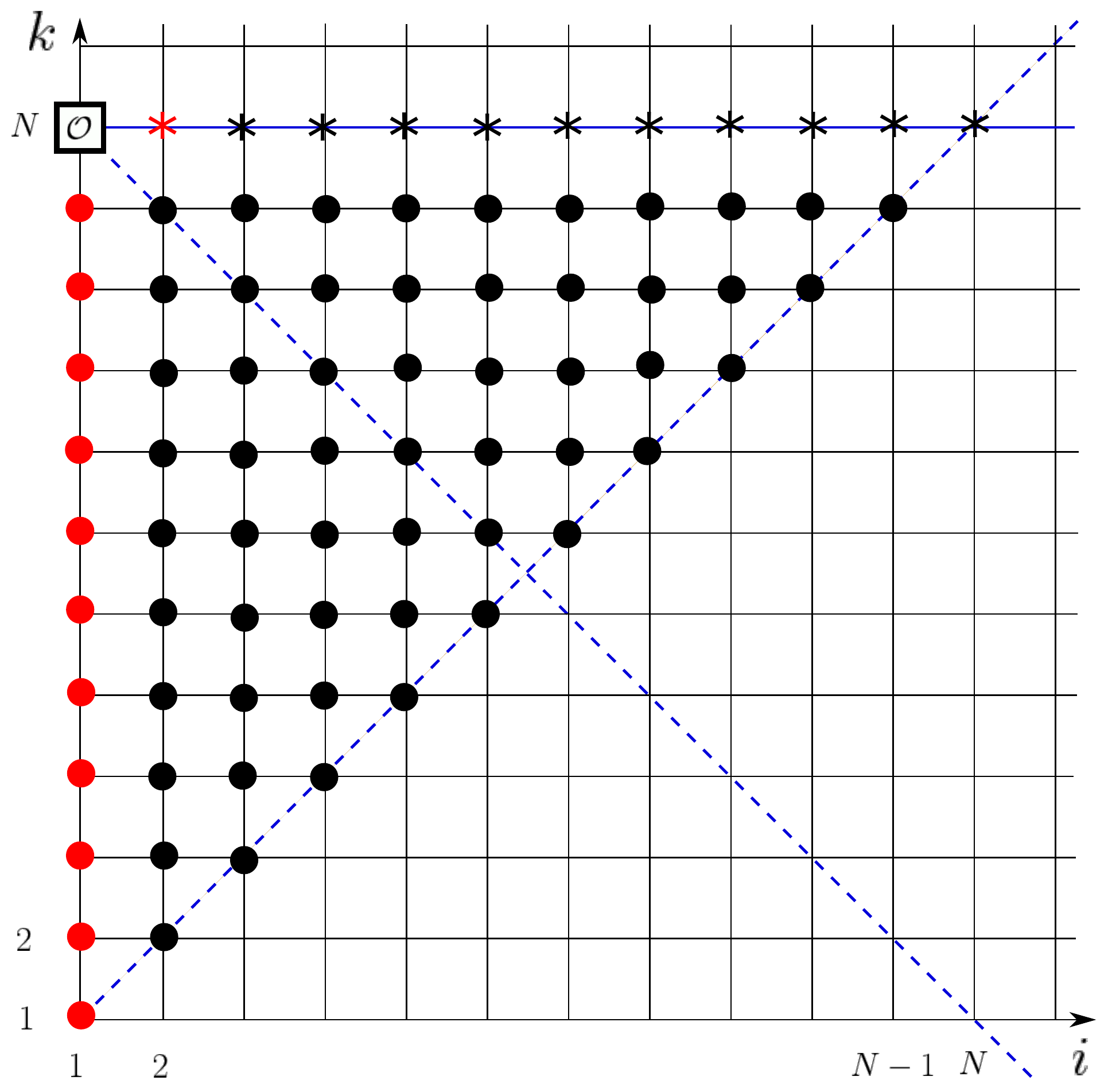} }}%
    \subfloat[]{{\includegraphics[width=0.57\textwidth]{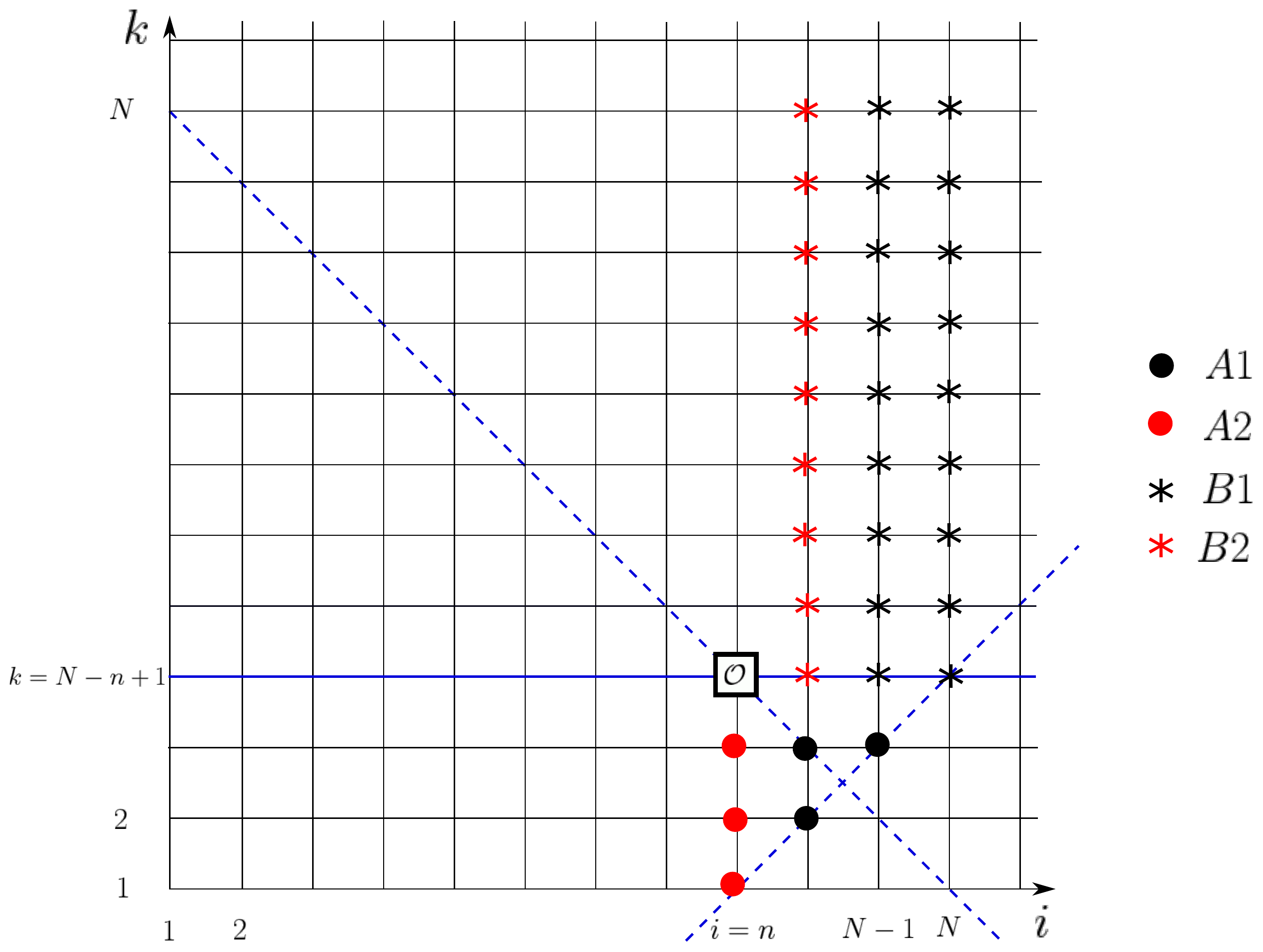} }}%
    \caption{Decomposition of $S$ into domains: (a) - for $n=1$ (smallest $n$) ; (b) - for $n=N-3$ (last non-degenerate case).}
    \label{fig:decompnr}%
\end{figure}

\begin{figure}[H]
    \centering
    \subfloat[]{{\includegraphics[width=0.48\textwidth]{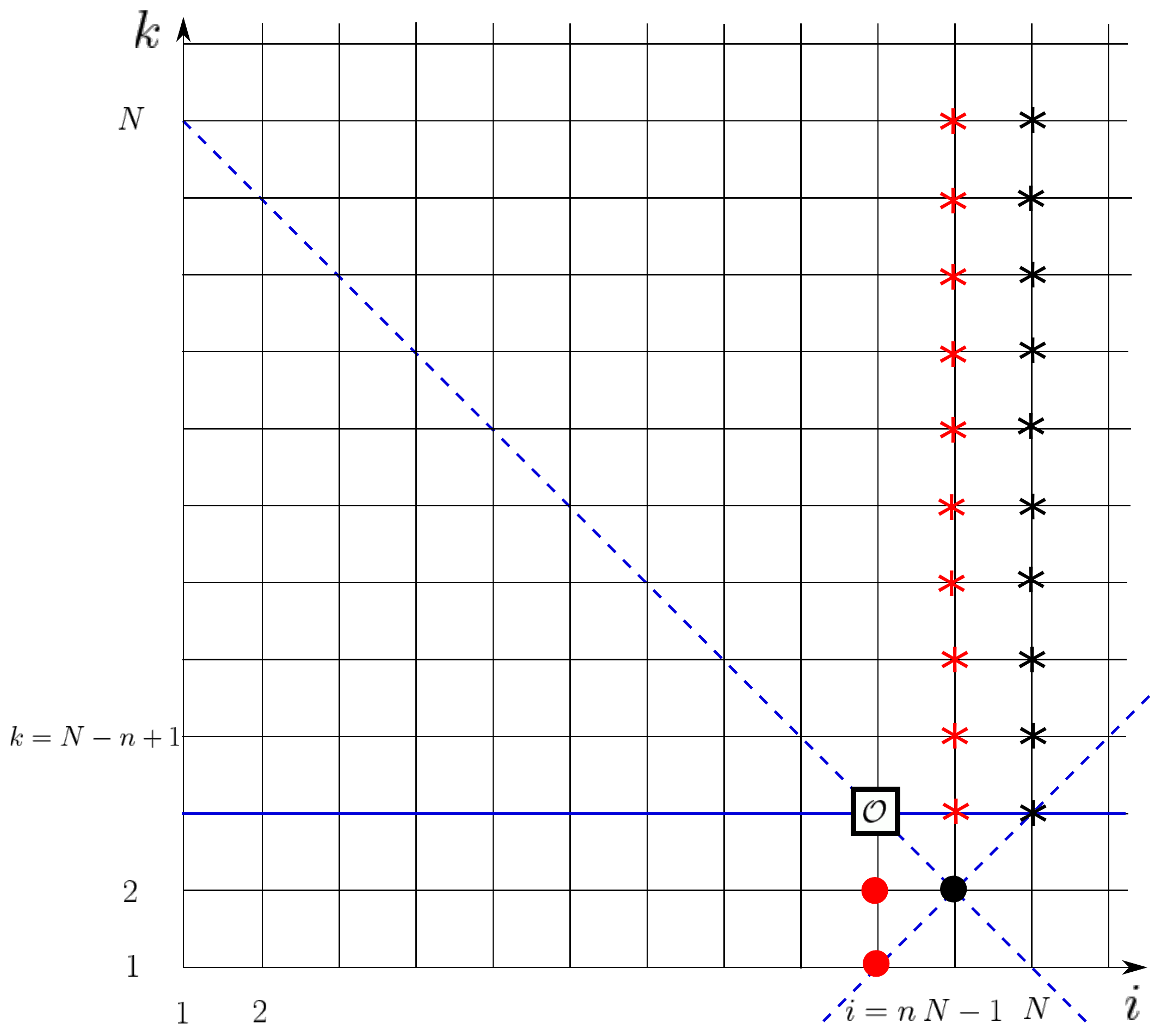} }}%
    \subfloat[]{{\includegraphics[width=0.52\textwidth]{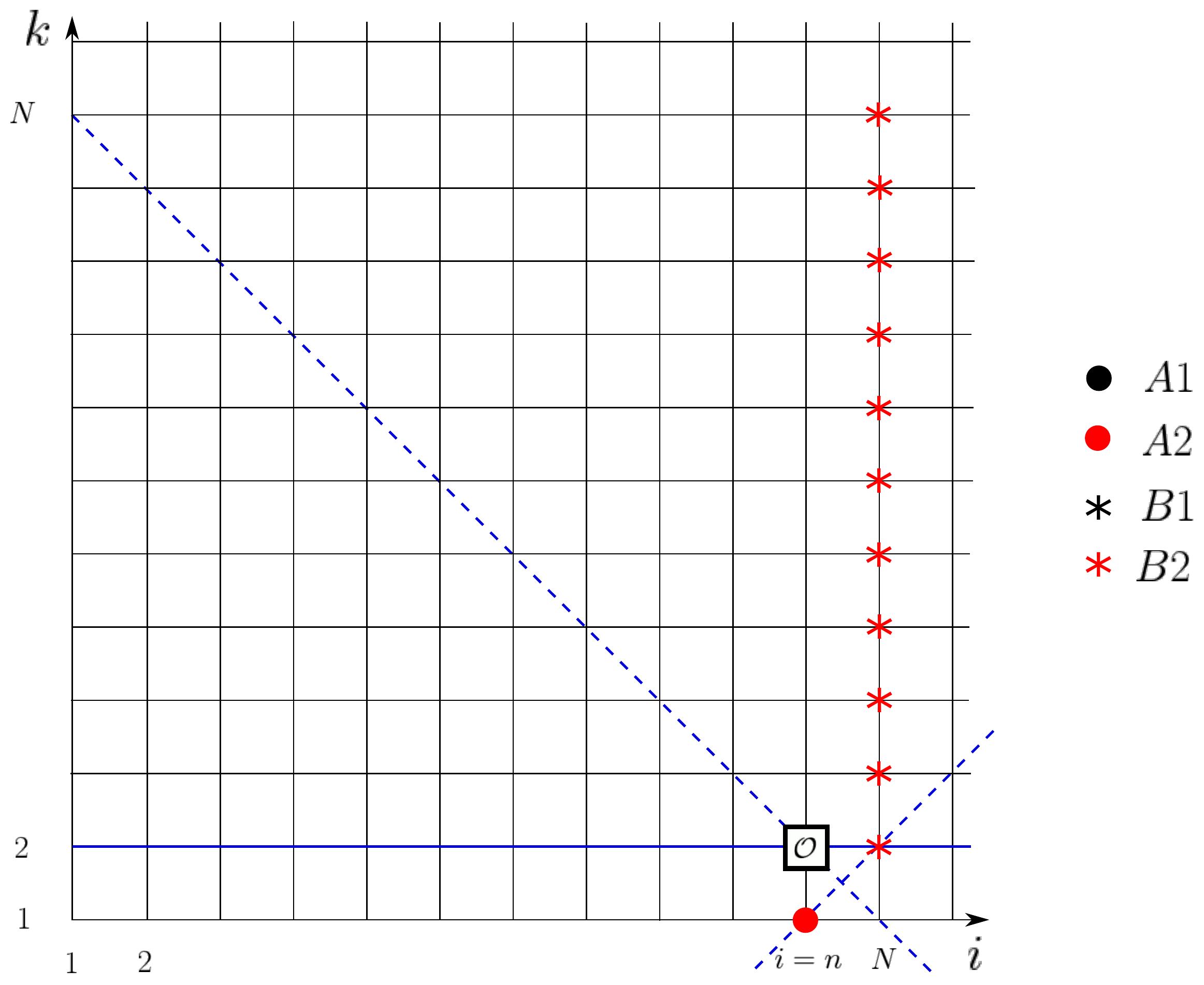} }}%
    \caption{Decomposition of $S$ into domains: (a) - for $n=N-2$ (degenerate case) ; (b) - for $n=N-1$ (largest $n$; degenerate case).}
    \label{fig:fig:decompdeg}%
\end{figure}
\noindent
\textbf{Preliminary calculations of $\delta(m)$}\\
Using similar reasoning as in case of $\Delta$ (see (\ref{eq:casescompchiF})-(\ref{eq:casescompchiL}) and Fig. \ref{fig:decompS}, we have two cases (since $k$ is fixed here, it is instructive to look also at Figure \ref{fig:klayers}).

\begin{align}
\label{eq:casecomp2}
\circled{I},\; &\mathrm{if} \; m < N-k+1 \mathrm{,} \nonumber\\
\circled{II},\; &\mathrm{if} \; m \geq N-k+1 \mathrm{.}\\\nonumber
\end{align}
Considering (\ref{eq:casecomp2}), we obtain

\begin{align}
&\circled{I}\;\sum(s,j,m)=\sum\limits_{n=1}^{m} \left( \circled{A1} \chi_{s+n-1 \geq j \geq n} + \circled{A2}\chi_{n=j} \right) \mathrm{,}\\
&\circled{II}\;\sum(s,j,m)=\sum\limits_{n=1}^{N-s} \left(\circled{A1} \chi_{s+n-1 \geq j \geq n} + \circled{A2}\chi_{n=j} \right) \\
+& \sum\limits_{n=N-s+1}^m \left(\circled{B1}\chi_{j > n+1} + \circled{B2}\chi_{j=n+1}\right) \mathrm{,} \nonumber
\end{align}
where
\begin{equation}
\sum(s,j,m) \coloneqq \sum_{n=1}^{m} (\#n)_j^s\mathrm{,}
\end{equation}
so that
\begin{equation}
\delta_j^s (m)= \sum(s,j+1,m) - \sum(s,j,m-1) \mathrm{.}
\end{equation}

Below we evaluate some telescopic sums which appear frequently in Section \ref{sec:app_b}.
\begin{align}
&\sum_{n=p}^q (\lambda_{n+b} - \lambda_{n+b+1})=\lambda_{p+b}-\lambda_{q+b+1}  \mathrm{,}\\
&\sum_{n=p}^q (n-p+1)(\lambda_{n+b}-\lambda_{n+b+1})=\sum_{n=p}^q \lambda_{n+b} - (q+p-1)\lambda_{q+b+1}  \mathrm{.}
\end{align}

\section{Calculations for Conditions 1 and 2}
\label{sec:app_b}
\subsection[xxx]{Calculations of $\Delta$ - Condition 1}
It is convenient to consider the following cases:
\begin{itemize}
\item Case I: $j=1$,
\begin{itemize}
\item Subcase I.a: $s=2$,
\item Subcase I.b: $N \geq s>2$,
\end{itemize}
\item Case II: $j=2$,
\begin{itemize}
\item Subcase II.a: $s=2$,
\item Subcase II.b :$s=3$,
\item Subcase II.c :$N-1>s>3$,
\item Subcase II.d :$s=N-1$,
\item Subcase II.e :$s=N$,
\end{itemize}
\item Case III: $N-1>j>2$,
\begin{itemize}
\item Subcase III.a: $s=2$,
\item Subcase III.b :$s=3$,
\item Subcase III.c: $N-j+2>s>3$,
\begin{itemize}
\item Subsubcase III.c.1: $j-s+2 \geq 1$,
\item Subsubcase III.c.2: $j-s+2 < 1$,
\end{itemize}
\item Subcase III.d: $s=N-j+2$,
\begin{itemize}
\item Subsubcase III.d.1: $j-s+2 \geq 1$,
\item Subsubcase III.d.2: $j-s+2 < 1$,
\end{itemize}
\item Subcase III.e: $s=N-j+3$,
\begin{itemize}
\item Subsubcase III.e.1: $j-s+2 \geq 1$,
\item Subsubcase III.e.2: $j-s+2 < 1$,
\end{itemize}
\item Subcase III.f: $ N \geq s > N-j+3$,
\begin{itemize}
\item Subsubcase III.f.1: $j-s+2 \geq 1$,
\item Subsubcase III.f.2: $j-s+2 < 1$,
\end{itemize}
\end{itemize}
\item Case IV: $j=N-1$,
\begin{itemize}
\item Subcase IV.a: $s=2$,
\item Subcase IV.b: $s=3$,
\item Subcase IV.c: $s=4$,
\item Subcase IV.d: $N \geq s>4$.
\end{itemize}
\end{itemize}
We define the functions
\begin{align}
\label{eq:sumfunF}
&M(n) \coloneqq \mathrm{min}\{N-n,s-1\}  \mathrm{,}\\
&M(n)_{-} \coloneqq \mathrm{min}\{N-n-1,s-1\}  \mathrm{,}\\
&\tilde{M}_{+} \coloneqq \mathrm{min}\{N-j+1,s-1\}  \mathrm{,}\\
&\tilde{M} \coloneqq \mathrm{min}\{N-j,s-1\}  \mathrm{,}\\
&\tilde{M}_{-}- \coloneqq \mathrm{min}\{N-j-1,s-1\}  \mathrm{,}\\
&mx_- \coloneqq \mathrm{max}\{j-s+1,1\}  \mathrm{,}\\
&mx \coloneqq \mathrm{max}\{j-s+2,1\}  \mathrm{,}\\
\label{eq:sumfunL}
&mx_+ \coloneqq \mathrm{max}\{j-s+3,1\}\mathrm{.}\\\nonumber  
\end{align}
We did not include the arguments $N,j$ and $s$ in the definitions \ref{eq:sumfunF} up to \ref{eq:sumfunL} since they will always be constant. Using already the fact that $j<N$ we can simplify $\chi_{j<N}=1$. \newline

\noindent \textbf{Case I ($j=1$)}
\begin{flalign}
& \Delta_1^s=\circled{1}^{n=1}_{i=1} + \cancel{\circled{0}^{n=2}_{i=1}} - \left( \circled{2}^{n=1}_{i=2} + \circled{1}^{n=2}_{i=2} \right)\\  
& = \sum_{k=1}^{M(1)}\circled{A2}_{i=1}^{n=1} -  \cancel{\sum_{k=N-n+1}^{s-1}\circled{B2}_{i=j+1}^{n=1} \chi_{s \geq N+1}} - \sum_{k=2}^{M(1)} \circled{A1}_{i=j+1}^{n=1} - \sum_{k=1}^{M(2)} \circled{A2}_{i=j+1}^{n=2}\nonumber\\ 
& =\sum_{k=1}^{M(1)} \lambda_k - \sum_{k=2}^{M(1)} \lambda_{k-1}+\sum_{k=2}^{M(1)} \lambda_{k}
 - \sum_{k=1}^{M(2)} \lambda_{k+1}\nonumber
\end{flalign}
\textbf{Subcase I.a ($s=2$)}
\begin{align}
\Delta_2^1=\lambda_1 - \lambda_2 = \lambda_{1,2}
\end{align}
\textbf{Subcase I.b ($N \geq s > 2$)}\\
\begin{align}
\Delta_s^1=\lambda_{s-1} - \lambda_s =\lambda_{s-1,s}
\end{align}
\textbf{Case II ($j=2$)}
\begin{flalign}
\Delta_2^s=&\circled{2}^{n=1}_{i=2} + \circled{1}^{n=2}_{i=2} +\cancel{\circled{0}^{n=3}_{i=2}} - \left(\circled{3}^{n=1}_{i=3}+ \circled{2}^{n=2}_{i=3} + \circled{1}^{n=3}_{i=3} \right)\\
=&\sum_{k=2}^{M(1)} \circled{A1}_{i=2}^{n=1} \chi_{s>2} + \sum_{k=N-n+1}^{s-1} \cancel{\circled{B2}_{i=2}^{n=1} \chi_{s > N}} + \sum_{k=1}^{M(2)} \circled{A2}_{i=2}^{n=2} \nonumber\\
 - &\sum_{k=3}^{M(1)}\circled{A1}_{i=3}^{n=1} \chi_{s > 3} - \sum_{k=N-n+1}^{s-1} \cancel{\circled{B2}_{i=3}^{n=1} \chi_{s > N}} - \sum_{k=2}^{M(2)} \circled{A1}_{i=3}^{n=2} \chi_{s > 2} \nonumber\\
 - & \sum_{k=N-1}^{s-1} \circled{B2}_{i=3}^{n=2} \chi_{ s \geq N} - \sum_{k=1}^{M(3)} \circled{A2}_{i=3}^{n=3}= \sum_{k=2}^{M(1)} (\lambda_{k-1}-\lambda_{k}) \chi_{s>2} + \sum_{k=1}^{M(2)} \lambda_{k+1} \nonumber\\
 -&\sum_{k=3}^{M(1)} (\lambda_{k-1} - \lambda_{k-1}) \chi_{s > 3} - \sum_{k=2}^{M(2)} (\lambda_{k}-\lambda_{k+1}) \chi_{s > 2} - (s-N+1)\lambda_{N-1} \chi_{ s \geq N} \\\nonumber
 - & \sum_{k=1}^{M(3)} \lambda_{k+2}\nonumber
\end{flalign}
\textbf{Subcase II.a ($s=2$)}
\begin{equation}
\Delta_2^2= \lambda_2 - \lambda_3 =\lambda_{2,3}
\end{equation}
\textbf{Subcase II.b ($s=3$)}
\begin{equation}
\Delta_3^2= \lambda_1 -\lambda_2 +\lambda_3 - \lambda_4=\lambda_{1,2} + \lambda_{3,4}
\end{equation}
\textbf{Subcase II.c ($N-1>s>3$)}
\begin{equation}
\Delta_s^2=\lambda_{s-2}+\lambda_s - \lambda_{s-1} -\lambda_{s+1}=\lambda_{s-1,s} + \lambda_{s,s+1}
\end{equation}
\textbf{Subcase II.d ($s=N-1$)}
\begin{equation}
\Delta_{N-1}^2=\lambda_{s-2}+\lambda_{s} - \lambda_{s-1} =\lambda_{s-2,s-1} +\lambda_s
\end{equation}
\textbf{Subcase II.e ($s=N$)}
\begin{equation}
\Delta_N^2=\lambda_{N-2}-\lambda_{N-1}=\lambda_{N-2,N-1}
\end{equation}
\textbf{Case III ($N-1>j>2$); generic}\\
Let us define
\begin{align}
\circled{R1} \coloneqq & \left(\sum_{n=j-s+2}^{j-2} \sum_{k=j-n+1}^{M(n)} \circled{A1}^n_{i=j} - \sum_{n=j-s+3}^{j-1}\sum_{k=j-n+1}^{M(n)-1} \circled{A1}^n_{i=j} \right)\chi_{s>3}\\
=& \left(\sum_{n=j-s+2}^{j-2}(\lambda_n - \lambda_{2n-j+M}) - \sum_{n=j-s+3}^{j-1}(\lambda_n - \lambda_{2n-j+M-1}) \right) \chi_{s>3}\nonumber\\
=& \left(\lambda_{j-s+2} - \lambda_{j-1}+\sum_{n=j-s+2}^{j-2}(\lambda_{2n-j+M(n)_{-}+1}-\lambda_{2n-j+M(n)}) \right) \chi_{s>3}\nonumber\\
\circled{R1}^{'}=&\sum_{n=1}^{j-2} \left( \sum_{k=j-n+1}^{M(n)} \circled{A1}^n_{i=j} - \sum_{k=j-n+1}^{M(n)-1} \circled{A1}^n_{i=j} \right)\chi_{s>3}\\
=& \left(\sum_{n=1}^{j-2} A1(n,M(n),j) - \sum_{k=j-n+2}^{\tilde{M}_+} A1(n,k,j+1) \right) \chi_{s>3}\nonumber\\
=& \left(\sum_{n=1}^{j-2}(\lambda_{2n-j+M(n)-1}-\lambda_{2n-j+M(n)}) - \lambda_{j-1} + \lambda_{j+\tilde{M}_+-3} \right)\chi_{s>3} \nonumber\\
\circled{R2} \coloneqq &\sum_{k=N-n+1}^{s-1} \circled{B1}_{i=j}^n \chi_{s \geq N-n+2} -\sum_{k=N-n+1}^{s-1} \circled{B1}_{i=j+1}^n \chi_{s \geq N-n+2} \\
=&\sum_{n=N-s+2}^{j-2}(s-N+n-1)\left(\lambda_{n+N-j}-\lambda_{n+N-j+1} \right)\chi_{s \geq N-j+4} \nonumber\\
-& \sum_{n=N-s+1}^{j-2} \left(\lambda_{n+N-j}-\lambda_{n+N-j+1}\right) \chi_{s \geq N-j+3}\nonumber\\
\circled{R3} \coloneqq &\left(\sum_{k=1}^{\tilde{M}_+-1} \circled{A1}_{i=j+1}^{n=j}-\sum_{k=2}^{\tilde{M}} \circled{A1}_{i=j+1}^{n=j}\right) \chi_{s > 2}=\\
=&\left[A1(j,1,j+1)  + A1(j,s-1,j+1) \chi_{N-j+1 \geq s}\right]\chi_{s > 2}\nonumber\\
= &\left[ \lambda_{i-1}-\lambda_i - (\lambda_{i+s-3}-\lambda_{i+s-2})\chi_{N-j+1 \geq s} \right] \chi_{s > 2}\nonumber\\
\circled{R4} \coloneqq &\sum_{k=N-j+2}^{s-1} \circled{B2}_{i=j}^{n=j-1} \chi_{s \geq  N-j+3}-\sum_{k=N-j+1}^{s-1} \circled{B2}_{i=j+1}^{n=j} \chi_{s \geq  N-j+2}\\
=&\left[(s-N+j-2)\chi_{s \geq N-j+3}-(s-N+j-1)\chi_{s \geq N-j+2}\right]\lambda_{N-1}\nonumber\\
\circled{R5} \coloneqq &\sum_{k=1}^{\tilde{M}} \circled{A2}_{i=j}^{n=j} - \sum_{k=1}^{\tilde{M}_{-}} \circled{A2}_{i=j+1}^{n=j+1}\\
=&\lambda_j - (\lambda_{j+s-1}) \chi_{s \leq N-j}\mathrm{.}\nonumber
\end{align}

Then
\begin{align}
\Delta_j^s= &\sum_{n=1}^{j-2} \circled{3}^n_{i=j} + \circled{2}^{n=j-1}_{i=j} + \circled{1}^{n=j}_{i=j} + \cancel{\circled{0}^{n=j+1}_{i=j}}\\
& - \left( \sum_{n=1}^{j-1} \circled{3}^n_{i=j+1} + \circled{2}^{n=j}_{i=j+1}\circled{1}^{n=j+1}_{i=j+1} \right)\nonumber\\
=&\sum_{n=1}^{j-2} \left( \sum_{k=j-n+1}^{M(n)}  \circled{A1}_{i=j}^n \chi_{s \geq j-n+2} + \sum_{k=N-n+1}^{s-1} \circled{B1}_{i=j}^n \chi_{s \geq N-n+2} \right) \nonumber\\
+ &\sum_{k=2}^{\tilde{M}} \circled{A1}_{i=j}^{n=j-1} \chi_{s > 2} + \sum_{k=N-j+2}^{s-1} \circled{B2}_{i=j}^{n=j-1} \chi_{s \geq  N-j+3}+ \sum_{k=1}^{\tilde{M}} \circled{A2}_{i=j}^{n=j} \nonumber\\
-&\sum_{n=1}^{j-1} \left( \sum_{k=j-n+2}^{M(n)}  \circled{A1}_{i=j+1}^n \chi_{s \geq j-n+3}- \sum_{k=N-n+1}^{s-1} \circled{B1}_{i=j+1}^n \chi_{s \geq N-n+2} \right) \nonumber\\
-& \sum_{k=2}^{\tilde{M}} \circled{A1}_{i=j+1}^{n=j} \chi_{s > 2} -\sum_{k=N-j+1}^{s-1} \circled{B2}_{i=j+1}^{n=j} \chi_{s \geq  N-j+2} - \sum_{k=1}^{\tilde{M}_{-}} \circled{A2}_{i=j+1}^{n=j+1}
\nonumber\\
&= \circled{R1}\chi_{j-s+2 \geq 1} +\circled{R1}^{'}\chi_{j-s+2 < 1} + \circled{R2}+ \circled{R3} + \circled{R4} + \circled{R5}  \nonumber\end{align}

\begin{align}
&\sum_{n=1}^{j-2}  \sum_{k=j-n+1}^{M(n)}  \circled{A1}_{i=j}^n \chi_{s \geq j-n+2}-\sum_{n=1}^{j-1} \sum_{k=j-n+2}^{M(n)}  \circled{A1}_{i=j+1}^n \chi_{s \geq j-n+3}\\
=&\left(\sum_{n=mx}^{j-2}  \sum_{k=j-n+1}^{M(n)} \circled{A1}_{i=j}^n -\sum_{n=mx_{+}}^{j-1} \sum_{k=j-n+2}^{M(n)}  \circled{A1}_{i=j+1}^n \right) \chi_{s >3}\nonumber\\
=& \circled{R1}\chi_{j-s+2 \geq 1} +\circled{R1}^{'}\chi_{j-s+2 < 1}\mathrm{.}\nonumber
\end{align}
\textbf{Subcase III.a ($s=2$)}\\
\begin{align}
\Delta^2_{j}=\circled{R5}=\lambda_{j,j+1}
\end{align}
\textbf{Subcase III.b ($s=3$)}\\
\begin{align}
\Delta^3_j=\circled{R3}+\circled{R5}=\lambda_{j-1} -\lambda_{j}-\cancel{\lambda_{j}}+\lambda_{j+1}+\cancel{\lambda_j} -\lambda_{j+2}=\lambda_{j-1,j} + \lambda_{j+1,j+2}
\end{align}
\textbf{Subcase III.c: ($N-j+2>s>3$)}\\
\begin{align}
\Delta^s_j=&\circled{R1}\chi_{j-s+2 \geq 1} +\circled{R1}^{'}\chi_{j-s+2 < 1}+\circled{R3} + \circled{R5}\\
&\circled{R3} + \circled{R5}=\lambda_{j-1}-\cancel{\lambda_j}-\lambda_{j+s-3}-\lambda_{j+s-2}+\cancel{\lambda_j}-\lambda_{j+s-1}\chi_{s \neq N-j+1}\nonumber\\
=&\lambda_{j-1}-\lambda_{j+s-3}-\lambda_{j+s-2}-\lambda_{j+s-1}\chi_{s \neq N-j+1}\nonumber
\end{align}
\textbf{Subsubcase III.c.1: ($j-s+2 \geq 1$)}\\
\begin{align}
\Delta^s_j=&\circled{R1} +\circled{R3} + \circled{R5}=\lambda_{j-s+2} - \cancel{\lambda_{j-1}}
\sum_{n=j-s+2}^{j-2} (\lambda_{2n-j+M(n)_{-}+1}\\
-&\lambda_{2n-j+M(n)}) + \cancel{\lambda_{j-1}}-\lambda_{j+s-3}-\lambda_{j+s-2}-\lambda_{j+s-1}\chi_{s \neq N-j+1}\nonumber\\
=&\lambda_{j-s+2}-\lambda_{j-s+1}\chi_{s \neq N-j+1}
\end{align}
\textbf{Subsubcase III.c.2: ($j-s+2 < 1$)}\\
\begin{align}
\Delta^s_j=&\circled{R1}^{'}+\circled{R3} + \circled{R5}=\sum_{n=1}^{j-2}\left(\lambda_{2n-j+M(n)-1}-\lambda_{2n-j+M(n)} \right) \\
-& \cancel{\lambda_{j-1}} + \lambda_{j+\tilde{M}_+-3} + \cancel{\lambda_{j-1}}-\lambda_{j+s-3}+\lambda_{j+s-2}-\lambda_{j+s-1}\chi_{s \neq N-j+1}\nonumber\\
=&\sum_{n=1}^{j-2}(\lambda_{2n-j+M(n)-1}-\lambda_{2n-j+M(n)}) -\lambda_{j+\tilde{M}_+-3} -\lambda_{j+s-3}+\lambda_{j+s-2} \nonumber\\
-& \lambda_{j+s-1}\chi_{s \neq N-j+1}\nonumber
\end{align}
\textbf{Subcase III.d: ($s=N-j+2$)}\\
\begin{align}
\Delta^{N-j+2}_j=&\circled{R1}\chi_{j-s+2 \geq 1} +\circled{R1}^{'}\chi_{j-s+2 < 1}+\circled{R3}+\circled{R4} + \circled{R5}\\
&\circled{R3}+\circled{R4} + \circled{R5}=\lambda_{j-1}-\cancel{\lambda_j} -\lambda_{N-1}+\cancel{\lambda_{j}}=\lambda_{j-1}-\lambda_{N-1}\nonumber\nonumber
\end{align}
\textbf{Subsubcase III.d.1: ($j-s+2 \geq 1$)}\\
\begin{align}
\Delta^{N-j+2}_j=&\circled{R1}+\circled{R3} + \circled{R4} + \circled{R5}=\lambda_{j-s+2}-\cancel{\lambda_{j-1}} \\
+& \sum_{n=j-s+2}^{j-2} (\lambda_{2n-j+M(n)_{-}+1} - \lambda_{2n-j+M(n)}) + \cancel{\lambda_{j-1}} -\lambda_{N-1}=\lambda_{2j-N}\nonumber\\
+&\sum_{n=2j-N}^{j-2} (\lambda_{2(n-j)+N+2}-\lambda_{2(n-j)+N+1})-\lambda_{N-1}\nonumber
\end{align}
\textbf{Subsubcase III.d.2: ($j-s+2 < 1$)}\\
\begin{align}
\Delta^{N-j+2}_j=&\circled{R1}^{'}+\circled{R3} + \circled{R4} + \circled{R5}=\sum_{n=1}^{j-2}(\lambda_{2n-j+M(n)-1}-\lambda_{2n-j+M(n)})\\
-& \cancel{\lambda_{j-1}} + \lambda_{j+\tilde{M}_+-3} +\cancel{\lambda_{j-1}}-\lambda_{N-1}=\sum_{n=1}^{N-s}(\lambda_{2(n-j)+N}-\lambda_{2(n-j)+N+1})\nonumber\\ +&\lambda_{N-2,N-1}\nonumber
\end{align}
\textbf{Subcase III.e: ($s=N-j+3$)}\\
\begin{align}
\Delta^{N-j+3}_j=&\circled{R1}\chi_{j-s+2 \geq 1} +\circled{R1}^{'}\chi_{j-s+2 < 1}+\circled{R2} + \circled{R3}+\circled{R4} + \circled{R5}\\
&\circled{R2} + \circled{R3}+\circled{R4} + \circled{R5}=-\lambda_{N-2}+\cancel{\lambda_{N-1}}+\lambda_{j-1}-\cancel{\lambda_{N-1}}\nonumber\\
-&\lambda_{N-2}+\lambda_{j-1}\nonumber
\end{align}
\textbf{Subsubcase III.e.1: ($j-s+2 \geq 1$)}\\
\begin{align}
\Delta^{N-j+3}_j=&\circled{R1} +\circled{R2} + \circled{R3}+\circled{R4} + \circled{R5}=\lambda_{2j-N-1}-\cancel{\lambda_{j-1}}\\
+&\sum_{n=j-s+2}^{j-2}(\lambda_{2n-j+M(n)_{-}+1}-\lambda_{2n-j+M(n)})-\lambda_{N-2}+\cancel{\lambda_{j-1}}\nonumber\\
=&\lambda_{2j-N-1}+\sum_{n=j-s+2}^{j-2}(\lambda_{2n-j+M(n)_{-}+1}-\lambda_{2n-j+M(n)})-\lambda_{N-2}\nonumber
\end{align}
\textbf{Subsubcase III.e.2: ($j-s+2 < 1$)}\\
\begin{align}
\Delta^{N-j+3}_j=&\circled{R1}^{'} +\circled{R2} + \circled{R3}+\circled{R4} + \circled{R5}\\
=&\sum_{n=1}^{j-2}\lambda_{2n-j+M(n)-1}-\lambda_{2n-j+M(n)})-\cancel{\lambda_{j-1}}+\lambda_{j+\tilde{M}_+-3}-\lambda_{N-2}+\cancel{\lambda_{j-1}}\nonumber\\
=&\sum_{n=1}^{j-2}(\lambda_{2n-j+M(n)-1}-\lambda_{2n-j+M(n)})
\end{align}
\textbf{Subcase III.f: ($ N \geq s > N-j+3$)}\\
\begin{align}
\Delta^s_j=&\circled{R1}\chi_{j-s+2 \geq 1} +\circled{R1}^{'}\chi_{j-s+2 < 1}+\circled{R2} + \circled{R3}+\circled{R4} + \circled{R5}\\
=&\sum_{n=N-s+1}^{j-2}(\lambda_{n+N-j+1}-\lambda_{n+N-j})+\lambda_{j-1}
-\cancel{\lambda_{j}}-\lambda_{N-1}+\cancel{\lambda_{j}}\nonumber\\
=&-\lambda_{2N-s-j+1}+\cancel{\lambda_{N-1}}+\lambda_{j-1}-\cancel{\lambda_{N-1}}=-\lambda_{2N-s-j+1}+\lambda_{j-1}\nonumber
\end{align}
\textbf{Subsubcase III.f.1: ($j-s+2 \geq 1$)}\\
\begin{align}
\Delta^s_j=&\circled{R1}+\circled{R2} + \circled{R3}+\circled{R4} + \circled{R5}\\
=&\lambda_{j-s+2}-\cancel{\lambda_{j-1}}+\sum_{n=j-s+2}^{j-2}(\lambda_{2n-j+M(n)_{-}+1}-\lambda_{2n-j+M(n)})\nonumber\\
-&\lambda_{2N-s-j+1}+\cancel{\lambda_{j-1}}\nonumber\\
=&\lambda_{j-s+2}+\sum_{n=j-s+2}^{j-2}(\lambda_{2n-j+M(n)_{-}+1}-\lambda_{2n-j+M(n)})-\lambda_{2N-s-j+1}\nonumber
\end{align}
\textbf{Subsubcase III.f.2: ($j-s+2 < 1$)}\\
\begin{align}
\Delta^s_j=&\circled{R1}^{'}+\circled{R2} + \circled{R3}+\circled{R4} + \circled{R5}=\sum_{n=1}^{j-2}(\lambda_{2n-j+M(n)-1}-\lambda_{2n-j+M(n)})\\
-&\cancel{\lambda_{j-1}}+\lambda_{i+\tilde{M}_{+}-3}-\lambda_{2N-s-j+1}+\cancel{\lambda_{j-1}}=\sum_{n=1}^{j-2}(\lambda_{2n-j+M(n)-1}-\lambda_{2n-j+M(n)})\nonumber\\
+&\lambda_{j+\tilde{M}_{+}-3}-\lambda_{2N-s-j+1}\nonumber
\end{align}
\textbf{Case IV} ($j=N-1$)\\
\begin{align}
\Delta_{N-1}^s=&\sum_{n=1}^{N-3}\circled{3}^{n}_{i=N-1} + \circled{2}^{n=N-2}_{i=N-1} +\circled{1}^{n=N-1}_{i=N-1} + \cancel{\circled{0}^{n=N}_{i=N-1}} \\
-& \left( \sum_{n=1}^{N-1} \circled{4}^n_{i=N} + \cancel{\circled{4}^{n=N}_{i=N}} \right) \nonumber\\
=&\sum_{n=1}^{N-3} \left( \sum_{k=j-n+1}^{M(n)} \circled{A1}_{i=N-1}^n \chi_{s \geq N-n+1} + \sum_{k=N-n+1}^{s-1} \circled{B1}_{i=N-1}^n  \chi_{s \geq N-n+2} \right) \nonumber\\
+& \sum_{k=2}^{M(N-2)} \circled{A1}_{i=N-1}^{n=N-2} \chi_{s > 2} + \sum_{k=3}^{s-1} \circled{B2}_{i=N-1}^{n=N-2}  \chi_{s > 3} + \sum_{k=1}^{M(N-1)} \circled{A2}_{i=N-1}^{n=N-1} \nonumber\\
-& \sum_{n=1}^{N-1}\sum_{k=2}^{s-1} \circled{B2}_{i=N}^{n}  \chi_{n=N-1}\chi_{s>2}  - \sum_{n=1}^{N-1}\sum_{k=N-n+1}^{s-1} \circled{B1}_{i=N}^n  \chi_{n \neq N-1} \chi_{s \geq N-n+2}\nonumber\\
=&\sum_{n=1}^{N-3}  \sum_{k=j-n+1}^{M(n)} (\lambda_{2n+k-N}-\lambda_{2n+k-N+1}) \chi_{s \geq N-n+1} \nonumber\\
+&\sum_{n=1}^{N-3} \sum_{k=N-n+1}^{s-1} (\lambda_{n+1}-\lambda_{n+2})  \chi_{s \geq N-n+2} \nonumber\\
+& \sum_{k=2}^{M(N-2)} (\lambda_{N+k-4} - \lambda_{N+k-3}) \chi_{s > 2} + \sum_{k=3}^{s-1} \lambda_{N-1}  \chi_{s > 3} + \sum_{k=1}^{M(N-1)} \lambda_{N+k-2} \nonumber\\
-& \sum_{n=1}^{N-1}\sum_{k=2}^{s-1} \lambda_{N-1}  \chi_{n=N-1}\chi_{s>2}  - \sum_{n=1}^{N-1}\sum_{k=N-n+1}^{s-1}(\lambda_n - \lambda_{n+1})  \chi_{n \neq N-1} \chi_{s \geq N-n+2}\nonumber
\end{align}
\textbf{Subcase IV.a: ($s=2$)}\\
\begin{align}
\Delta_{N-1}^2=\sum_{k=1}^{M(N-1)} \lambda_{N+k-2}=\lambda_{N-1}
\end{align}
\textbf{Subcase IV.b: ($s=3$)}\\
\begin{align}
\Delta_{N-1}^3=&\sum_{k=2}^{M(N-2)} (\lambda_{N+k-4}-\lambda_{N+k-3})+\sum_{k=1}^{M(N-1)} \lambda_{N+k-2} - \sum_{n=1}^{N-1} \sum_{k=2}^{s-1} \lambda_{N-1}\\
=&\lambda_{N-2}- \cancel{\lambda_{N-1}}+\cancel{\lambda_{N-1}}-\lambda_{N-1}=\lambda_{N-2,N-1}\nonumber
\end{align}
\textbf{Subcase IV.c: ($s=4$)}\\
\begin{align}
\Delta_{N-1}^4=&\sum_{k=3}^{3} (\lambda_{2n+k-n-3}-\lambda_{2n+k-n-2})+\sum_{k=2}^{2}(\lambda_{N+k-4}-\lambda_{N+k-3}) \\
+& (s-3)\lambda_{N-1} +\lambda_{N-1} -\sum_{k=3}^3 (\lambda_{N-2}-\lambda_{N-1}) - 2\lambda_{N-1}\nonumber\\
=&\lambda_{N-3}-\lambda_{N-2} + \cancel{\lambda_{N-2}} - \cancel{\lambda_{N-1}}+\cancel{\lambda_{N-1}}
+ \cancel{\lambda_{N-1}}- \cancel{\lambda_{N-2}}+\cancel{\lambda_{N-1}} \nonumber\\
-& \cancel{2 \lambda(N-1)}=\lambda_{N-3,N-2}\nonumber
\end{align}
\textbf{Subcase IV.d: ($N \geq s>4$)}\\
\begin{align}
\Delta_j^s&=\sum_{n=N-s+1}^{N-3} (\lambda_n - \lambda_{n+1}) + \sum_{n=N-s+2}^{N-3} \sum_{k=N-n+1}^{s-1} (\lambda_{n+1}-\lambda_{n+2}) + \lambda_{N-2} \cancel{\lambda_{N-1}} \\
+&(s-3)\lambda_{N-1} + \cancel{\lambda_{N-1}} - \sum_{n=N-s+2}^{N-2} \sum_{k=N-n+1}^{s-1} (\lambda_n -\lambda_{n+1})- (s-2)\lambda_{N-1}\nonumber\\
=&\lambda_{N-s+1} - \cancel{\lambda_{N-2}} + \sum_{n=N-s+2}^{N-3} (s-N+n-1)(\lambda_{n+1}-\lambda_{n+2}) + \cancel{\lambda_{N-2}} - \lambda_{N-1} \nonumber\\
-& \sum_{n=N-s+2}^{N-2} (s-N+n-1)(\lambda_n-\lambda_{n+1})=\lambda_{N-s+1} -\lambda_{N-1} + \sum_{n=N-s+2}^{N-3} \lambda_{n+1} \nonumber\\
-& (s-6)\lambda_{N-1}-\sum_{n=N-s+2}^{N-2}\lambda_n + (s-5)\lambda_{N-1} = \lambda_{N-s+1} - \cancel{\lambda_{N-1}} + \cancel{\lambda_{N-1}} \nonumber\\
+& \sum_{n=N-s+3}^{N-2} \lambda_n - \sum_{n=N-s+2}^{N-2} \lambda_n\nonumber
=\lambda_{N-s+1}-\lambda_{N-s+2}=\lambda_{N-s+1,N-s+2}\nonumber
\end{align}
Clearly, $\Delta_j^s$ is non-negative in every case, so Condition 1 is satisfied.
\subsection[xxx]{Calculations of $\delta(m)$}

Let us move to the calculations of $\delta_j^s (m)$. 
It is convenient to introduce the following cases:
\begin{itemize}
\item Case I: $ N-s+1 > m \geq 1$,
\begin{itemize}
\item Subcase I.a: $j+1 > m \geq 1$,
\begin{itemize}
\item Subsubcase I.a.1($j \geq s$),
\item Subsubcase I.a.2($j < s$),
\end{itemize}
\item Subcase I.b :$m \geq j+1$,
\begin{itemize}
\item Subsubcase I.b.1($j \geq s$),
\item Subsubcase I.b.2($j < s$),
\end{itemize}
\end{itemize}
\item Case II: $m = N-s+1$,
\begin{itemize}
\item Subcase II.a :$j > m \geq 1$,
\begin{itemize}
\item Subsubcase II.a.1($j \geq s$),
\item Subsubcase II.a.2($j < s$),
\end{itemize}
\item Subcase II.b: $m=j$,
\begin{itemize}
\item Subsubcase II.b.1($j \geq s$),
\item Subsubcase II.b.2($j < s$),
\end{itemize}
\item Subcase II.c :$m \geq j+1$,
\begin{itemize}
\item Subsubcase II.c.1($j \geq s$),
\item Subsubcase II.c.2($j < s$),
\end{itemize}
\end{itemize}
\item Case III: $m = N-s+2$,
\begin{itemize}
\item Subcase III.a: $j > m \geq 1$,
\item Subcase III.b :$m \geq j$,
\end{itemize}
\item Case IV: $N > m > N-s+2$,
\begin{itemize}
\item Subcase IV.a: $j > m$,
\item Subcase IV.b :$m \geq j$,
\end{itemize}
\end{itemize}
\textbf{Case I}($N-s+1 > m \geq 1$)\\
\textbf{Subcase I.a}($j+1 > m \geq 1$)\\
\textbf{Subsubcase I.a.1}($j \geq s$)\\
\begin{align}
\delta_j^s(m)=\left(\sum_{n=j-s+2}^{m} \circled{A1}_{i=j+1}^n -\sum_{n=j-s+1}^{m-1} \circled{A1}_{i=j}^n \chi_{m \geq 2}\right)\chi_{m \geq j-s+2}
\end{align}
\textbf{Subsubcase I.a.2}($j < s$)\\
\begin{align}
\delta_j^s(m)=\sum_{n=1}^{m} \circled{A1}_{i=j+1}^n -\sum_{n=1}^{m-1} \circled{A1}_{i=j}^n \chi_{m \geq 2}
\end{align}
\textbf{Subcase I.b}($m \geq j+1$)\\
\textbf{Subsubcase I.b.1}($j \geq s$)\\
\begin{align}
\delta_j^s(m)=&\left(\sum_{n=j-s+2}^{j} \circled{A1}_{i=j+1}^n-\sum_{n=j-s+1}^{j-1} \circled{A1}_{i=j}^n\chi_{j \geq 2} \right)\chi_{s \geq 2} \\
+& \circled{A2}_{i=j+1}^{n=j+1} -\circled{A2}_{i=j}^{n=j} \nonumber
\end{align}
\textbf{Subsubcase I.b.2}($j < s$)\\
\begin{align}
\delta_j^s(m)=\sum_{n=1}^{j} \circled{A1}_{i=j+1}^n-\sum_{n=1}^{j-1} \circled{A1}_{i=j}^n\chi_{j \geq 2} + \circled{A2}_{i=j+1}^{n=j+1} -\circled{A2}_{i=j}^{n=j}
\end{align}
\textbf{Case II}($m = N-s+1$)\\
\textbf{Subcase II.a}($j > m$)\\
\textbf{Subsubcase II.a.1}($j \geq s$)\\
\begin{align}
\delta_j^s(m)= \left( \sum_{n=j-s+2}^{m-1} \circled{A1}_{i=j+1}^n \chi_{j < N-1} -\sum_{n=j-s+1}^{m-1} \circled{A1}_{i=j}^n \right) \chi_{s < N} + \circled{B1}_{i=j+1}^{n=m} 
\end{align}
\textbf{Subsubcase II.a.2}($j < s$)\\
\begin{align}
\delta_j^s(m)= \left( \sum_{n=1}^{m-1} \circled{A1}_{i=j+1}^n  -\sum_{n=1}^{m-1} \circled{A1}_{i=j}^n \right) \chi_{s < N} + \circled{B1}_{i=j+1}^{n=m} 
\end{align}
\textbf{Subcase II.b}($m=j$)\\
\textbf{Subsubcase II.b.1}($j \geq s$)\\
\begin{align}
\delta_j^s(m)=&\left(\sum_{n=j-s+2}^{j-1} \circled{A1}_{i=j+1}^n \chi_{s \geq 3}  -\sum_{n=j-s+1}^{j-1} \circled{A1}_{i=j}^n \chi_{s \geq 2}\right)\chi_{s<N} \\
+& \circled{B2}_{i=j+1}^{n=j} \nonumber
\end{align}
\textbf{Subsubcase II.b.2}($j < s$)\\
\begin{align}
\delta_j^s(m)=&\left(\sum_{n=1}^{j-1} \circled{A1}_{i=j+1}^n -\sum_{n=1}^{j-1} \circled{A1}_{i=j}^n\right)\chi_{s<N} \\
+& \circled{B2}_{i=j+1}^{n=j} \nonumber
\end{align}
\textbf{Subcase II.c}($m \geq j+1$)\\
\textbf{Subsubcase II.c.1}($j \geq s$)\\
\begin{align}
\delta_j^s(m)=\left(\sum_{n=j-s+2}^{j} \circled{A1}_{i=j+1}^n  -\sum_{n=j-s+1}^{j-1} \circled{A1}_{i=j}^n \chi_{j \geq 2} \right)\chi_{s \geq 2} - \circled{A2}_{i=j}^{n=j}
\end{align}
\textbf{Subsubcase II.c.2}($j < s$)\\
\begin{align}
\delta_j^s(m)=\sum_{n=1}^{j} \circled{A1}_{i=j+1}^n  -\sum_{n=1}^{j-1} \circled{A1}_{i=j}^n \chi_{j \geq 2} - \circled{A2}_{i=j}^{n=j}
\end{align}
\textbf{Case III}($m = N-s+2$)\\
\textbf{Subcase III.a}($m < j $)\\
\begin{align}
\delta_j^s(m)=& \left(\sum_{n=j-s+2}^{N-s} \circled{A1}_{i=j+1}^n \chi_{j< N-1}  -\sum_{n=j-s+1}^{N-s} \circled{A1}_{i=j}^n\right)\chi_{s<N} \\
+&  \sum_{n=N-s+1}^m \circled{B1}_{i=j+1}^{n} - \sum_{n=N-s+1}^{m-1}\circled{B1}_{i=j}^{n} \nonumber\\
=& \left(\sum_{n=j-s+2}^{N-s} \circled{A1}_{i=j+1}^n\chi_{j<N-1} -\sum_{n=j-s+1}^{N-s} \circled{A1}_{i=j}^n\right)\chi_{s<N} \nonumber\\
+&  \circled{B1}_{i=j+1}^{n=m-1} + \cancel{\circled{B1}_{i=j+1}^{n=m}}  - \cancel{\circled{B1}_{i=j}^{n=m-1}} \nonumber\\
=& \left(\sum_{n=j-s+2}^{N-s} \circled{A1}_{i=j+1}^n\chi_{j<N-1} -\sum_{n=j-s+1}^{N-s} \circled{A1}_{i=j}^n\right)\chi_{s<N} +  \circled{B1}_{i=j+1}^{n=N-s+1}\nonumber
\end{align}
\textbf{Subcase III.b}($m \geq j$)\\
\begin{align}
\delta_j^s(m)=& \left(\sum_{n=j-s+2}^{N-s} \circled{A1}_{i=j+1}^n \chi_{j<N-1} -\sum_{n=j-s+1}^{N-s} \circled{A1}_{i=j}^n\right)\chi_{s<N} \\
+&  \sum_{n=N-s+1}^{m-1} \circled{B1}_{i=j+1}^{n} + \cancel{\circled{B2}_{i=j+1}^{n=m}}
 - \cancel{\circled{B2}_{i=j}^{n=m-1}}\nonumber\\
 =& \left(\sum_{n=j-s+2}^{N-s} \circled{A1}_{i=j+1}^n \chi_{j<N-1} -\sum_{n=j-s+1}^{N-s} \circled{A1}_{i=j}^n\right)\chi_{s<N} + \circled{B1}_{i=j+1}^{n=N-s+1}\nonumber
\end{align}
\textbf{Case IV}($N > m > N-s+2$)\\
\textbf{Subcase IV.a}($m < j$)\\
\begin{align}
\delta_j^s(m)=& \left(\sum_{n=j-s+2}^{N-s} \circled{A1}_{i=j+1}^n \chi_{j<N-1}-\sum_{n=j-s+1}^{N-s} \circled{A1}_{i=j}^n \right)\chi_{s<N}\\
+&  \sum_{n=N-s+1}^{m} \circled{B1}_{i=j+1}^{n}  - \sum_{n=N-s+1}^{m-1}\circled{B1}_{i=j}^{n}\nonumber\\
=& \left(\sum_{n=j-s+2}^{N-s} \circled{A1}_{i=j+1}^n \chi_{j<N-1}-\sum_{n=j-s+1}^{N-s} \circled{A1}_{i=j}^n \right)\chi_{s<N} \nonumber\\
+& \circled{B1}_{i=j+1}^{n=N-s+1} + \cancel{\sum_{n=N-s+2}^{m} \circled{B1}_{i=j+1}^{n}} 
 - \cancel{\sum_{n=N-s+1}^{m-1}\circled{B1}_{i=j}^{n}}\nonumber\\
 =&\left(\sum_{n=j-s+2}^{N-s} \circled{A1}_{i=j+1}^n \chi_{j<N-1}  -\sum_{n=j-s+1}^{N-s} \circled{A1}_{i=j}^n \right)\chi_{s<N} + \circled{B1}_{i=j+1}^{n=N-s+1}\nonumber
\end{align}
\textbf{Subcase IV.b}($m \geq j$)\\
\begin{align}
\delta_j^s(m)=& \left(\sum_{n=j-s+2}^{N-s} \circled{A1}_{i=j+1}^n \chi_{j<N-1} -\sum_{n=j-s+1}^{N-s} \circled{A1}_{i=j}^n \right)\chi_{s<N} \\
+&  \sum_{n=N-s+1}^{m-1} \circled{B1}_{i=j+1}^{n}  + \cancel{\circled{B2}_{i=j+1}^{n=m}} \
- \sum_{n=N-s+1}^{m-2} \circled{B1}_{i=j}^n - \cancel{\circled{B2}_{i=j}^{n=m-1}}\nonumber\\
=& \left(\sum_{n=j-s+2}^{N-s} \circled{A1}_{i=j+1}^n\chi_{j<N-1} -\sum_{n=j-s+1}^{N-s} \circled{A1}_{i=j}^n\right)\chi_{s<N} \nonumber\\
+ &  \circled{B1}_{i=j+1}^{n=N-s+1} +  \cancel{\sum_{n=N-s+2}^{m-1} \circled{B1}_{i=j+1}^{n}}  - \cancel{\sum_{n=N-s+1}^{m-2} \circled{B1}_{i=j}^n}\nonumber\\
=& \left(\sum_{n=j-s+2}^{N-s} \circled{A1}_{i=j+1}^n\chi_{j<N-1} -\sum_{n=j-s+1}^{N-s} \circled{A1}_{i=j}^n\right)\chi_{s<N} 
+ \circled{B1}_{i=j+1}^{n=N-s+1}\nonumber
\end{align}
Cases III.a, III.b, IV.a and IV.b are, in fact, described by a single formula, e.g. the one from Case III.a. Hence we will merge them to a case called Case III ($N > m > N-s+1$).

\subsection[xxx]{Comparison of $\Delta$ and $\delta(m)$ - Condition 2}

By comparing the values of $\Delta^k_i$ and $\delta^k_i(m)$ in overlapping cases, one may check that indeed $\Delta^k_i \geq \delta^k_i(m)$ for all $i\in \{1,\ldots, N-1\}$, $k \in \{2,\ldots, N\}$ and $m \in \{1,\ldots, N-1\}$. Hence Condition 2 is satisfied.

\section{Conditions 3 and 4}
\subsection{Condition 3 (row counting)}
\label{sec:cond3}
Clearly, the partial sum of entries of a telescope $T_{p,q}$ of the form (\ref{eq:tel}) is a telescopic sum equal

\begin{equation}
\label{eq:topbot}
  \begin{pmatrix}
    \lambda_{p+q-1}\\
    \lambda_{p+q-2}\\
    \vdots\\
    \lambda_{p+1}\\
    \lambda_{p}
  \end{pmatrix} \mathrm{,}
  \end{equation}
from top to bottom and
\begin{equation}
  \begin{pmatrix}
    \lambda_{p}\\
    \lambda_{p}-\lambda_{p+q-1}\\
    \vdots\\
    \lambda_{p}-\lambda_{p+2}\\
    \lambda_{p}-\lambda_{p+1}
  \end{pmatrix} \mathrm{,}
  \end{equation}\noindent
from bottom to top. Note that the entries of (\ref{eq:topbot}) are determined, up to position $q$, only by the top entry $\lambda_{p+q-1}$, which in turn depends only on $p+q$. Let $\sigma_{n,k}$ be the difference between the rows in which top virtual boxes of columns, for two consecutive labels $n$ and $n+1$, appear at a given step $k$. Note that from (\ref{eq:beg}), for fixed $k$ and $m>n$, we have $\text{beg}(m,k) > \text{beg}(n,k)$. Thus, $\sigma_{n,k}$ is either $1$ or $2$ and $\sigma_{n,k}=2$ is obtained exactly in step $k=N-n$. Moreover, from (\ref{eq:len}), for fixed $k$ and $m>n$, we have $\text{len}(m,k) \geq \text{len}(n,k)$. Hence, the difference $\tau_{n,k}$ between lengths of telescopes for two consecutive labels $n$ and $n+1$, at given step $k$, is either $0$ or $1$. Thus, it suffices to consider three cases, depicted in Fig. \ref{fig:rc}.

\begin{figure}[H]
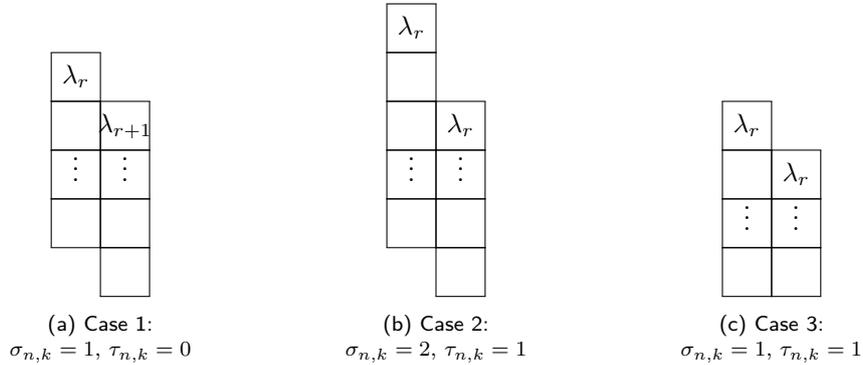

\centering
\ytableausetup{nosmalltableaux,mathmode,boxsize=1.8em}

 \begin{subfigure}[b]{0.27\textwidth} 
 \centering
    \begin{ytableau}[]
       \lambda_r  & \none  \\
         &  \lambda_{r+1}  \\
  \vdots & \vdots    \\ 
         &    \\  
   \none &     
   \end{ytableau}
    \caption{Case 1:\\$\sigma_{n,k}=1,\,\tau_{n,k}=0$}
    \label{fig:rc1}
  \end{subfigure}
   \begin{subfigure}[b]{0.27\textwidth} 
   \centering 
   \begin{ytableau}[]
       \lambda_r &  \none  \\
         &  \none  \\
         & \lambda_r  \\
 \vdots  &  \vdots  \\ 
         &    \\  
   \none &    
    \end{ytableau}
    \caption{Case 2:\\ $\sigma_{n,k}=2,\,\tau_{n,k}=1$}
    \label{fig:rc2}
  \end{subfigure}
  \begin{subfigure}[b]{0.27\textwidth} 
  \centering
    \begin{ytableau}[]
       \lambda_r & \none   \\
         & \lambda_r  \\
  \vdots & \vdots  \\ 
         &      
    \end{ytableau}

    \caption{Case 3:\\ $\sigma_{n,k}=1,\,\tau_{n,k}=1$}
    \label{fig:rc3}
  \end{subfigure}
  \caption{Possible relative positions of the columns of virtual boxes corresponding to telescopes (columns on the left -- with label $n$, columns on the right -- with label $n+1$). Symbols on the top indicate the relationship between the top entries of compared telescopes (i.e. top entries of (\ref{eq:topbot})).}
  \label{fig:rc}
\end{figure}
Since (\ref{eq:topbot}) is non-increasing from top to bottom, it is clear that Condition 3 is satisfied locally (i.e. for any two telescopes with consecutive labels, in any fixed step) in every case, so Condition 3 is always satisfied.

\subsection{Condition 4 (column counting)} 
We show that Condition 4 in fact follows from what we already have proved. Let $\tilde{\#}(m)^{k}_{l}$ denote the number of boxes with label $m$ appended in step $k$ up to the $l$th column (counted from right to left). Notice that the inequalities proved within Condition 2, $\Delta_i^k  \geq \delta_{i}^k(m)$, imply that, for a fixed step $k$, the blocks with the same label are distributed stair-like (block in row $i$ ends before or just as the block in row $i+1$ begins), as shown in Fig. \ref{fig:stairs}.
\begin{figure}[H]
\centering
\includegraphics[width=0.4\linewidth]{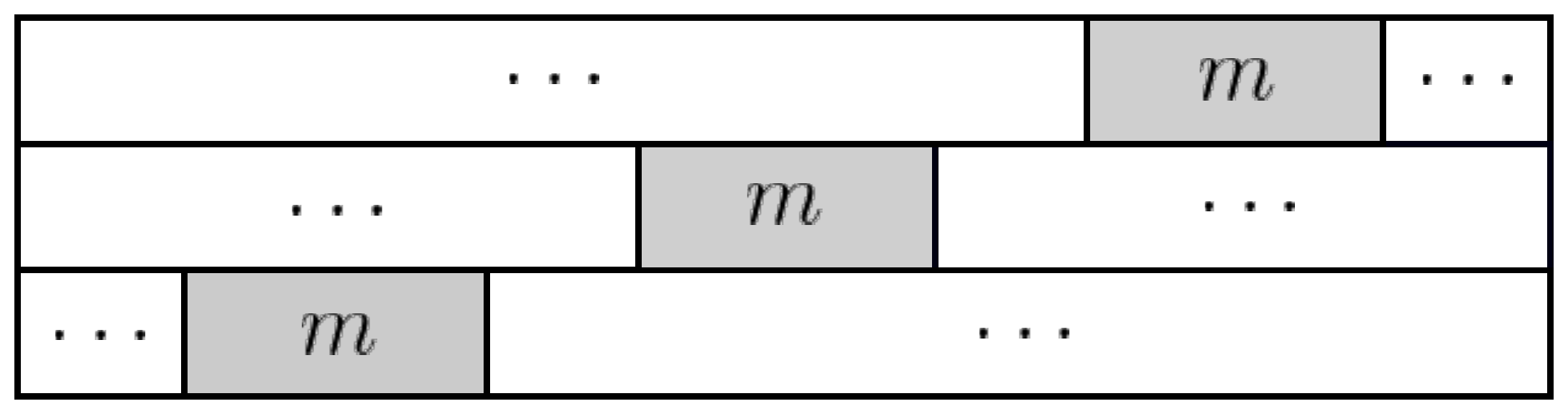}
  \caption{Stair-like distribution of blocks with label $m$ at step $k$. Notice that the spacings might in fact be equal to zero, contrary to what the figure might suggest.}
\label{fig:stairs}
\end{figure}\noindent
Moreover, in every row the labels of blocks are in non-decreasing order (from left to right). In order to prove that Condition 4 is satisfied, it suffices to show that
\begin{equation}
\label{eq:quantity}
\tilde{\#}(m-1)^k_{l_0}- \tilde{\#}(m)^k_{l_0} \geq 0 \mathrm{,}
\end{equation}
whenever it makes sense.
We distinguish two cases - either the column $l_0$ crosses some $m$-block, say at row $i$, (Case I) or it does not (Case II). We start with Case I.\\
\textbf{Case I}\\
There are essentially two subcases presented in Fig. \ref{fig:cases}. Other possibilities fulfilling of Conditions 1 and 2 can be effectively reduced to those cases. The column in consideration is denoted by $l_0$, the columns where an $(m-1)$-block begins and an $m$-block ends are denoted by $l^\prime$ and $l$ respectively.
\begin{figure}[H]
\begin{subfigure}[t]{.4\textwidth}
  \centering
  \includegraphics[width=.8\linewidth]{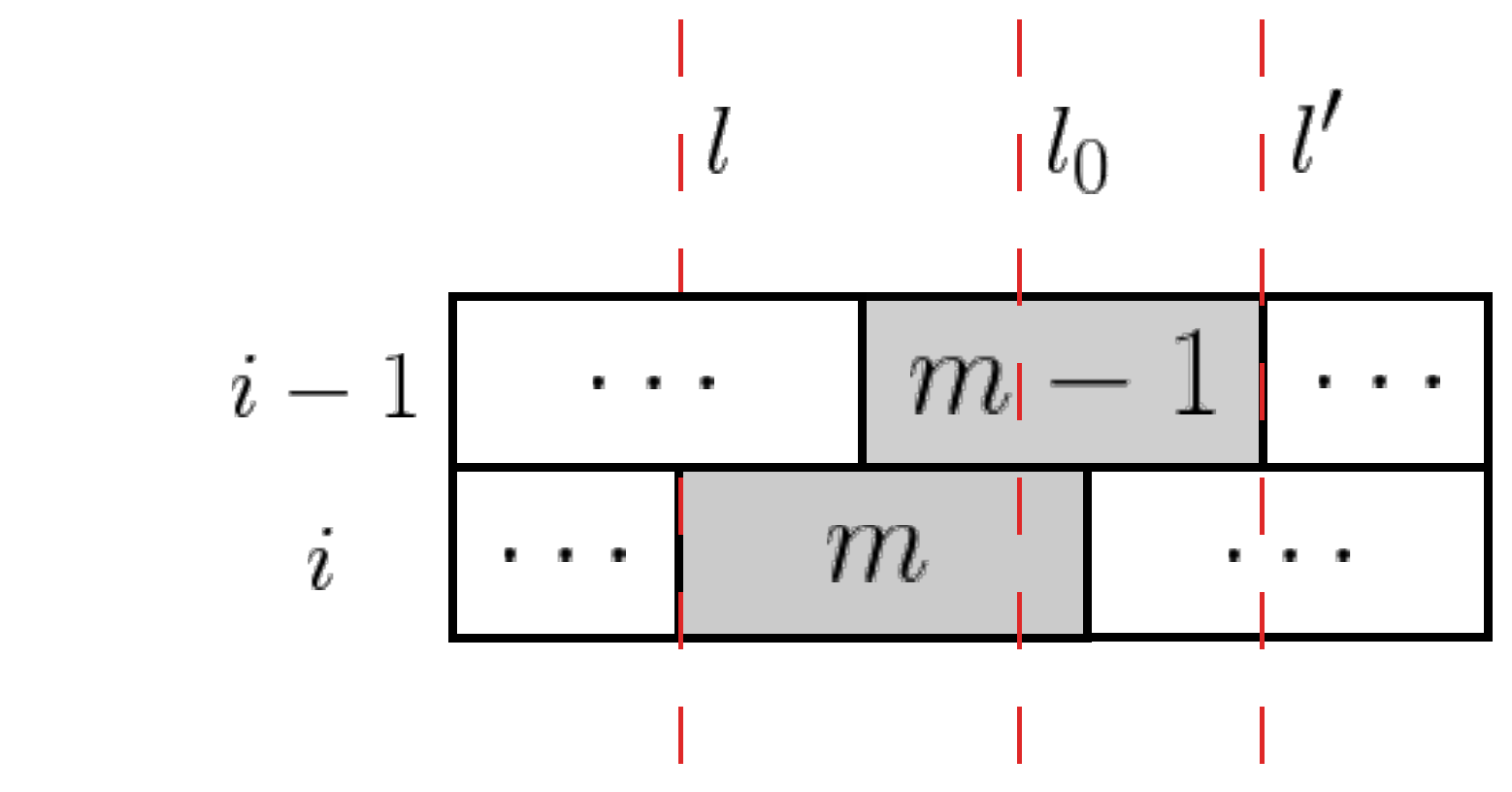}  
  \caption{Subcase I.A - the $m$-block which intersects with column $l_0$ at row $i$ begins before (or just as) the ($m-1$)-block at row $i-1$ ends. }
  \label{fig:casea}
\end{subfigure}
\hspace{1cm}
\begin{subfigure}[t]{.4\textwidth}
  \centering
  \includegraphics[width=.9\linewidth]{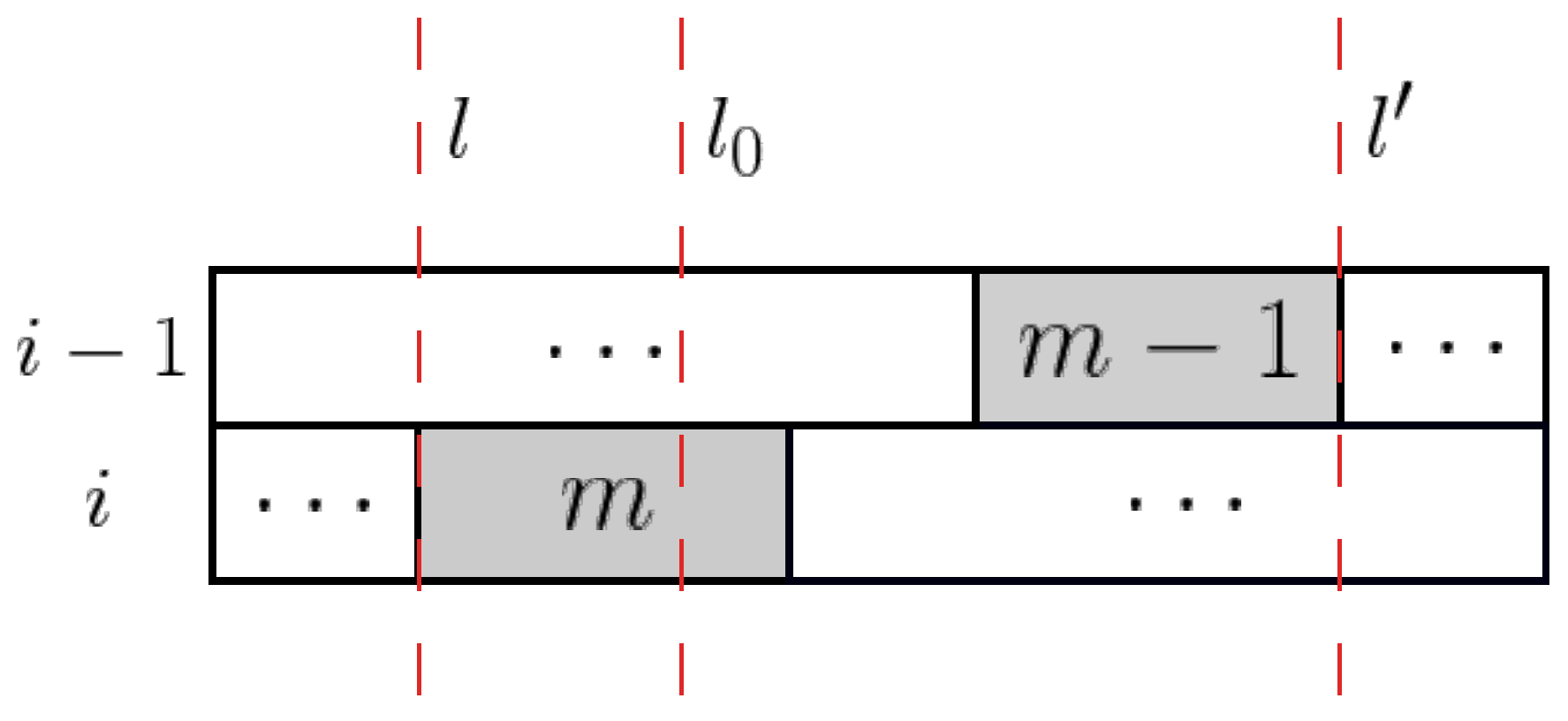}  
  \caption{Subcase I.B - the $m$-block which intersects with column $l_0$ at row $i$ begins after the ($m-1$)-block at row $i-1$ ends.}
  \label{fig:caseb}
\end{subfigure}
\caption{Possible subcases of Case I.}
\label{fig:cases}
\end{figure}\noindent
The situation to the right of a column $l^\prime$ is schematically depicted in Fig. \ref{fig:rightto}.
\begin{figure}[H]
\centering
\includegraphics[width=0.5\linewidth]{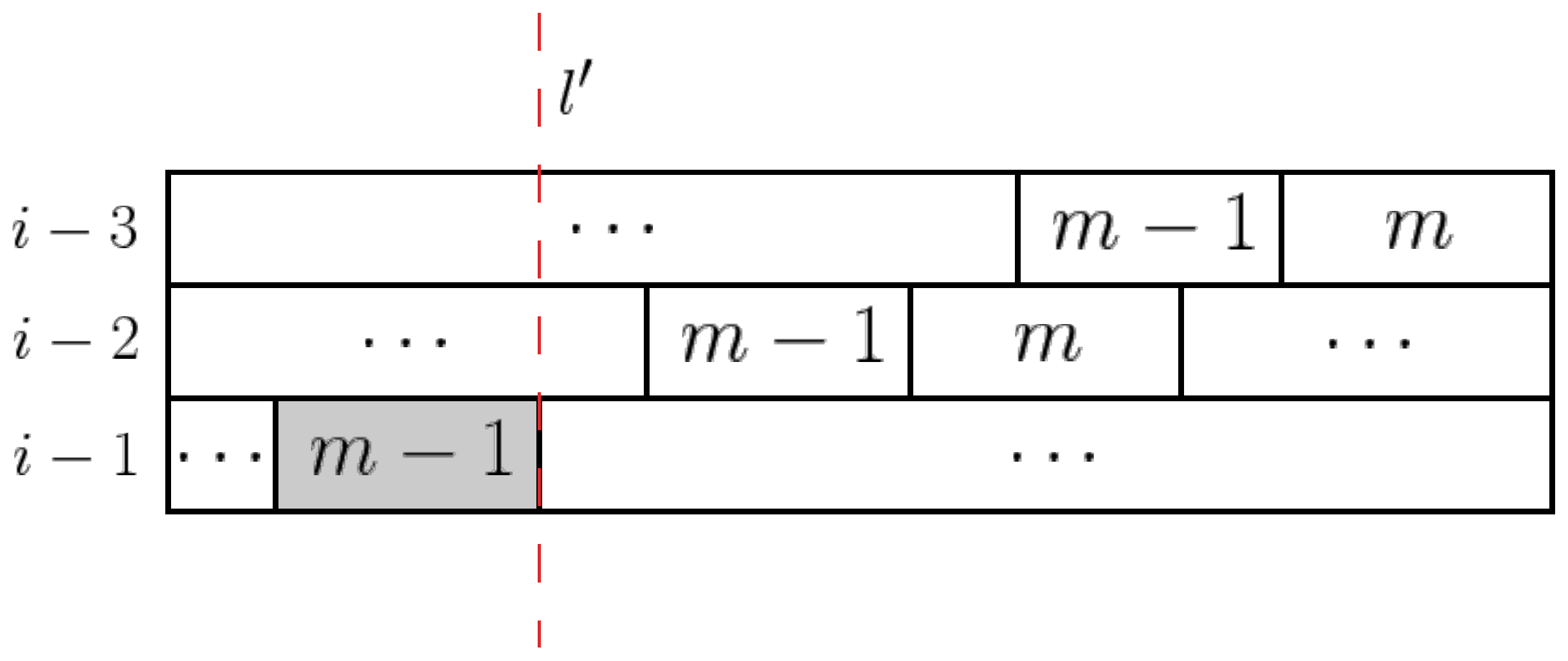}
  \caption{The schematic configuration to the right of the column $l^\prime$. The grey block is an $(m-1)$-block in Fig. \ref{fig:cases}.}
\label{fig:rightto}
\end{figure}\noindent
Since Condition 3 (row-counting) is satisfied, considering row $i-2$ we may write (see Fig. \ref{fig:rightto})
\begin{equation}
 \sum_{j<i-1}\#(m-1)^k_j \geq  \sum_{j<i-1}\#(m)^k_j \mathrm{,} 
\end{equation}
but
\begin{equation}
\tilde{\#}(m-1)^k_{l^\prime}=\sum_{j<i-1}\#(m-1)^k_j\mathrm{,}
\end{equation}
and
\begin{equation}
\tilde{\#}(m)^k_{l^\prime}=\sum_{j<i-1}\#(m)^k_j  \mathrm{.}
\end{equation}
Thus, we obtain
\begin{equation}
\label{eq:totheright}
\tilde{\#}(m-1)^k_{l^\prime} \geq \tilde{\#}(m)^k_{l^\prime}\mathrm{.}
\end{equation}
\textbf{Subcase I.A}\\
Notice that the number of $(m-1)$-boxes added to the left of $l^{\prime}$ is greater or equal to the corresponding number of $m$-boxes, counted up to the column $l$ (see Fig. \ref{fig:cases}), so from (\ref{eq:totheright}) we see that (\ref{eq:quantity}) is non-negative.\\
\textbf{Subcase I.B}\\ 
It suffices to check that (see Fig. \ref{fig:cases})
\begin{equation}
\label{eq:formula}
\sum_{j<i} \#(m-1)^k_j \geq \sum_{j \leq i} \#(m)^k_j \mathrm{.}
\end{equation}
Just as in case of Condition 3, we have three cases depicted in Fig. \ref{fig:rc}. The only difference here is that partial sums of telescopes under consideration are taken not up to the same row, but up to two consecutive rows -- $(i-1)$ for label $m-1$ and $i$ for label $m$. Hence, it is clear that (\ref{eq:quantity}) is satisfied.\\
\textbf{Case II}\\ 
\begin{figure}[H]
\begin{subfigure}[t]{.5\textwidth}
  \centering
  \includegraphics[width=.7\linewidth]{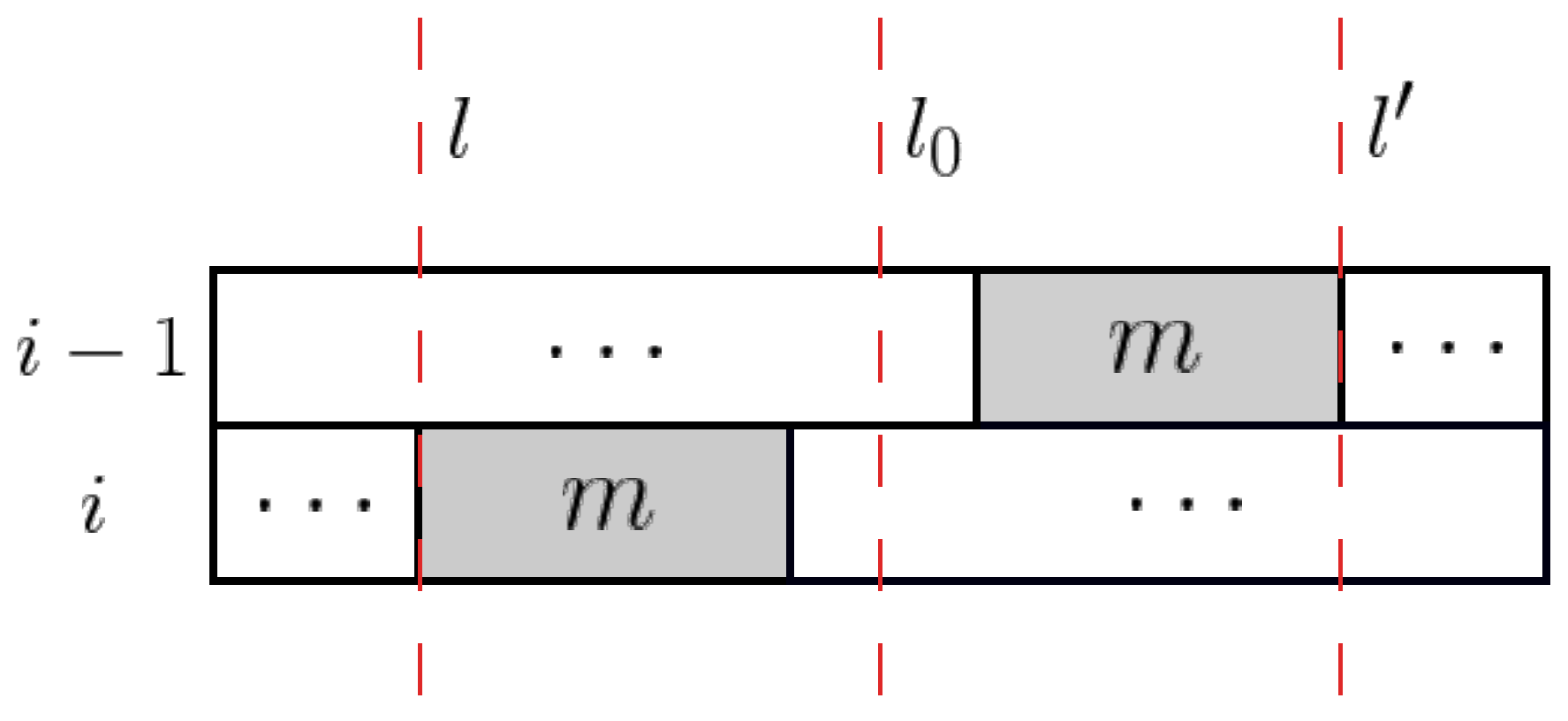}  
  \caption{Configuration in Case II.}
  \label{fig:casebm1}
\end{subfigure}
\begin{subfigure}[t]{.5\textwidth}
  \centering
  \includegraphics[width=.7\linewidth]{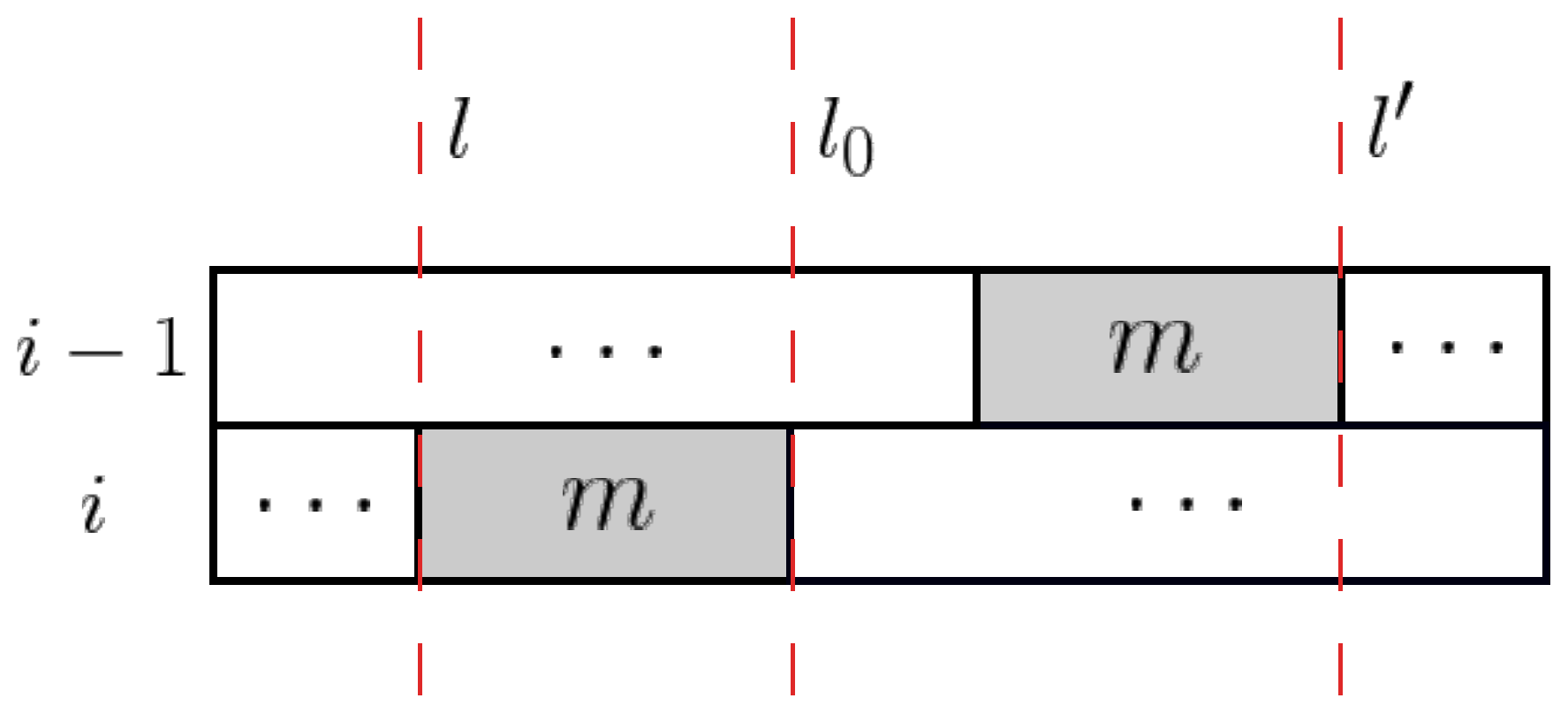}  
  \caption{Reduction to Case I.}
  \label{fig:casebm2}
\end{subfigure}
\caption{Reduction of Case II to Case I - the applied change of the choice of the column $l_0$ does not decrease (\ref{eq:quantity}).}
\label{fig:transform}
\end{figure}\noindent
Note that the the column $l_0$ can be moved to the left without decreasing (\ref{eq:quantity}) (see Fig. \ref{fig:transform}), so that it goes through the last box of the lower $m$-block. Thus, Case II can be reduced to Case I. This ends the whole proof.

\bibliographystyle{unsrt}

\end{document}